\documentclass[reqno,12pt,a4paper]{article}

\usepackage{amsfonts,amssymb,amsthm, 
amsmath,amscd, bbm, bm, mathabx, 
mathrsfs,delarray,subfigure,framed,enumitem,graphicx,extarrows}

\usepackage[dvipsnames]{xcolor}[svgnames,x11names]

\usepackage{hyperref}
\usepackage{cite}
\usepackage{xypic}
\usepackage{asymptote}
\usepackage{sectsty}

\definecolor{light-gray}{gray}{0.96}
\definecolor{shadecolor}{named}{light-gray}

\usepackage[scr=rsfs,cal=dutchcal]{mathalfa}

\usepackage{qcircuit}

\usepackage{fancyhdr}
\fancyheadoffset[RE,LO]{0.\textwidth} 

\hypersetup{
  colorlinks,
  linkcolor=blue,
  linktoc=all,
  citecolor=blue,
   urlcolor=blue
}

\usepackage{sectsty}

\colorlet{sectitlecolor}{blue}
\colorlet{sectboxcolor}{white}
\colorlet{secnumcolor}{blue}

\sectionfont{\color{sectitlecolor}}

\makeatletter
\renewcommand\@seccntformat[1]{%
  \colorbox{sectboxcolor}{\textcolor{secnumcolor}{\csname the#1\endcsname}}%
  \quad
}
\makeatother

\usepackage{etoolbox}
\patchcmd{\thebibliography}{\section*{\refname}}{}{}{}

\newcommand{\rme}{{\mathrm{e}}}

\newcommand{\rmB}{{\mathrm{B}}}
\newcommand{\rmC}{{\mathrm{C}}}

\newcommand{\rmH}{{\mathrm{H}}}

\DeclareSymbolFont{sfletters}{OML}{cmbrm}{m}{it}  

\DeclareMathSymbol{\sfGamma}{\mathord}{sfletters}{"00}
\DeclareMathSymbol{\sfDelta}{\mathord}{sfletters}{"01}
\DeclareMathSymbol{\sfTheta}{\mathord}{sfletters}{"02}
\DeclareMathSymbol{\sfLambda}{\mathord}{sfletters}{"03}
\DeclareMathSymbol{\sfXi}{\mathord}{sfletters}{"04}
\DeclareMathSymbol{\sfPi}{\mathord}{sfletters}{"05}
\DeclareMathSymbol{\sfSigma}{\mathord}{sfletters}{"06}
\DeclareMathSymbol{\sfUpsilon}{\mathord}{sfletters}{"07}
\DeclareMathSymbol{\sfPhi}{\mathord}{sfletters}{"08}
\DeclareMathSymbol{\sfPsi}{\mathord}{sfletters}{"09}
\DeclareMathSymbol{\sfOmega}{\mathord}{sfletters}{"0A}
\DeclareMathSymbol{\sfalpha}{\mathord}{sfletters}{"0B}
\DeclareMathSymbol{\sfbeta}{\mathord}{sfletters}{"0C}
\DeclareMathSymbol{\sfgamma}{\mathord}{sfletters}{"0D}
\DeclareMathSymbol{\sfdelta}{\mathord}{sfletters}{"0E}
\DeclareMathSymbol{\sfepsilon}{\mathord}{sfletters}{"0F}
\DeclareMathSymbol{\sfzeta}{\mathord}{sfletters}{"10}
\DeclareMathSymbol{\sfeta}{\mathord}{sfletters}{"11}
\DeclareMathSymbol{\sftheta}{\mathord}{sfletters}{"12}
\DeclareMathSymbol{\sfiota}{\mathord}{sfletters}{"13}
\DeclareMathSymbol{\sfkappa}{\mathord}{sfletters}{"14}
\DeclareMathSymbol{\sflambda}{\mathord}{sfletters}{"15}
\DeclareMathSymbol{\sfmu}{\mathord}{sfletters}{"16}
\DeclareMathSymbol{\sfnu}{\mathord}{sfletters}{"17}
\DeclareMathSymbol{\sfxi}{\mathord}{sfletters}{"18}
\DeclareMathSymbol{\sfpi}{\mathord}{sfletters}{"19}
\DeclareMathSymbol{\sfrho}{\mathord}{sfletters}{"1A}
\DeclareMathSymbol{\sfsigma}{\mathord}{sfletters}{"1B}
\DeclareMathSymbol{\sftau}{\mathord}{sfletters}{"1C}
\DeclareMathSymbol{\sfupsilon}{\mathord}{sfletters}{"1D}
\DeclareMathSymbol{\sfphi}{\mathord}{sfletters}{"1E}
\DeclareMathSymbol{\sfchi}{\mathord}{sfletters}{"1F}
\DeclareMathSymbol{\sfpsi}{\mathord}{sfletters}{"20}
\DeclareMathSymbol{\sfomega}{\mathord}{sfletters}{"21}
\DeclareMathSymbol{\sfvarepsilon}{\mathord}{sfletters}{"22}
\DeclareMathSymbol{\sfvartheta}{\mathord}{sfletters}{"23}
\DeclareMathSymbol{\sfvarpi}{\mathord}{sfletters}{"24}
\DeclareMathSymbol{\sfvarrho}{\mathord}{sfletters}{"25}
\DeclareMathSymbol{\sfvarsigma}{\mathord}{sfletters}{"26}
\DeclareMathSymbol{\sfvarphi}{\mathord}{sfletters}{"27}

\DeclareMathSymbol{\spartial}{\mathord}{sfletters}{"40}

\DeclareMathSymbol{\sfA}{\mathord}{sfletters}{"41}
\DeclareMathSymbol{\sfB}{\mathord}{sfletters}{"42}
\DeclareMathSymbol{\sfC}{\mathord}{sfletters}{"43}
\DeclareMathSymbol{\sfD}{\mathord}{sfletters}{"44}
\DeclareMathSymbol{\sfE}{\mathord}{sfletters}{"45}
\DeclareMathSymbol{\sfF}{\mathord}{sfletters}{"46}
\DeclareMathSymbol{\sfG}{\mathord}{sfletters}{"47}
\DeclareMathSymbol{\sfH}{\mathord}{sfletters}{"48}
\DeclareMathSymbol{\sfI}{\mathord}{sfletters}{"49}
\DeclareMathSymbol{\sfJ}{\mathord}{sfletters}{"4A}
\DeclareMathSymbol{\sfK}{\mathord}{sfletters}{"4B}
\DeclareMathSymbol{\sfL}{\mathord}{sfletters}{"4C}
\DeclareMathSymbol{\sfM}{\mathord}{sfletters}{"4D}
\DeclareMathSymbol{\sfN}{\mathord}{sfletters}{"4E}
\DeclareMathSymbol{\sfO}{\mathord}{sfletters}{"4F}
\DeclareMathSymbol{\sfP}{\mathord}{sfletters}{"50}
\DeclareMathSymbol{\sfQ}{\mathord}{sfletters}{"51}
\DeclareMathSymbol{\sfR}{\mathord}{sfletters}{"52}
\DeclareMathSymbol{\sfS}{\mathord}{sfletters}{"53}
\DeclareMathSymbol{\sfT}{\mathord}{sfletters}{"54}
\DeclareMathSymbol{\sfU}{\mathord}{sfletters}{"55}
\DeclareMathSymbol{\sfV}{\mathord}{sfletters}{"56}
\DeclareMathSymbol{\sfW}{\mathord}{sfletters}{"57}
\DeclareMathSymbol{\sfX}{\mathord}{sfletters}{"58}
\DeclareMathSymbol{\sfY}{\mathord}{sfletters}{"59}
\DeclareMathSymbol{\sfZ}{\mathord}{sfletters}{"5A}
\DeclareMathSymbol{\sfa}{\mathord}{sfletters}{"61}
\DeclareMathSymbol{\sfb}{\mathord}{sfletters}{"62}
\DeclareMathSymbol{\sfc}{\mathord}{sfletters}{"63}
\DeclareMathSymbol{\sfd}{\mathord}{sfletters}{"64}
\DeclareMathSymbol{\sfe}{\mathord}{sfletters}{"65}
\DeclareMathSymbol{\sff}{\mathord}{sfletters}{"66}
\DeclareMathSymbol{\sfg}{\mathord}{sfletters}{"67}
\DeclareMathSymbol{\sfh}{\mathord}{sfletters}{"68}
\DeclareMathSymbol{\sfi}{\mathord}{sfletters}{"69}
\DeclareMathSymbol{\sfj}{\mathord}{sfletters}{"6A}
\DeclareMathSymbol{\sfk}{\mathord}{sfletters}{"6B}
\DeclareMathSymbol{\sfl}{\mathord}{sfletters}{"6C}
\DeclareMathSymbol{\sfm}{\mathord}{sfletters}{"6D}
\DeclareMathSymbol{\sfn}{\mathord}{sfletters}{"6E}
\DeclareMathSymbol{\sfo}{\mathord}{sfletters}{"6F}
\DeclareMathSymbol{\sfp}{\mathord}{sfletters}{"70}
\DeclareMathSymbol{\sfq}{\mathord}{sfletters}{"71}
\DeclareMathSymbol{\sfr}{\mathord}{sfletters}{"72}
\DeclareMathSymbol{\sfs}{\mathord}{sfletters}{"73}
\DeclareMathSymbol{\sft}{\mathord}{sfletters}{"74}
\DeclareMathSymbol{\sfu}{\mathord}{sfletters}{"75}
\DeclareMathSymbol{\sfv}{\mathord}{sfletters}{"76}
\DeclareMathSymbol{\sfw}{\mathord}{sfletters}{"77}
\DeclareMathSymbol{\sfx}{\mathord}{sfletters}{"78}
\DeclareMathSymbol{\sfy}{\mathord}{sfletters}{"79}
\DeclareMathSymbol{\sfz}{\mathord}{sfletters}{"7A}

\usepackage[LGR,T1]{fontenc}
\usepackage{amsmath,etoolbox}

\newcommand{\declarebsfgreek}[2]{%
  \protected\csdef{bsf#1}{\mathord{\text{\bsfgreekfont#2}}}%
}
\newcommand{\bsfgreekfont}{\usefont{LGR}{cmss}{bx}{it}}

\declarebsfgreek{alpha}{a}
\declarebsfgreek{beta}{b}
\declarebsfgreek{gamma}{g}
\declarebsfgreek{delta}{d}
\declarebsfgreek{epsilon}{e}
\declarebsfgreek{zeta}{z}
\declarebsfgreek{eta}{h}
\declarebsfgreek{theta}{j}
\declarebsfgreek{iota}{i}
\declarebsfgreek{kappa}{k}
\declarebsfgreek{lambda}{l}
\declarebsfgreek{mu}{m}
\declarebsfgreek{nu}{n}
\declarebsfgreek{xi}{x}
\declarebsfgreek{omicron}{o}
\declarebsfgreek{pi}{p}
\declarebsfgreek{rho}{r}
\declarebsfgreek{sigma}{s}
\declarebsfgreek{tau}{t}
\declarebsfgreek{upsilon}{u}
\declarebsfgreek{phi}{f}
\declarebsfgreek{chi}{q}
\declarebsfgreek{psi}{y}
\declarebsfgreek{omega}{w}
\declarebsfgreek{varsigma}{c}

\declarebsfgreek{Gamma}{G}
\declarebsfgreek{Delta}{D}
\declarebsfgreek{Upsilon}{U}
\declarebsfgreek{Omega}{W}
\declarebsfgreek{Theta}{J}
\declarebsfgreek{Lambda}{L}
\declarebsfgreek{Xi}{X}
\declarebsfgreek{Pi}{P}
\declarebsfgreek{Sigma}{S}
\declarebsfgreek{Phi}{F}
\declarebsfgreek{Psi}{Y}

\newcommand{\declarebsfitalic}[2]{%
  \protected\csdef{bsf#1}{\mathord{\text{\bsfitalicfont#2}}}%
}
\newcommand{\bsfitalicfont}{\usefont{T1}{cmss}{bx}{it}}

\declarebsfitalic{a}{a}
\declarebsfitalic{b}{b}
\declarebsfitalic{c}{c}
\declarebsfitalic{d}{d}
\declarebsfitalic{e}{e}
\declarebsfitalic{f}{f}
\declarebsfitalic{g}{g}
\declarebsfitalic{h}{h}
\declarebsfitalic{i}{i}
\declarebsfitalic{j}{j}
\declarebsfitalic{k}{k}
\declarebsfitalic{l}{l}
\declarebsfitalic{m}{m}
\declarebsfitalic{n}{n}
\declarebsfitalic{o}{o}
\declarebsfitalic{p}{p}
\declarebsfitalic{q}{q}
\declarebsfitalic{r}{r}
\declarebsfitalic{s}{s}
\declarebsfitalic{t}{t}
\declarebsfitalic{u}{u}
\declarebsfitalic{v}{v}
\declarebsfitalic{x}{x}
\declarebsfitalic{y}{y}
\declarebsfitalic{z}{z}

\declarebsfitalic{A}{A}
\declarebsfitalic{B}{B}
\declarebsfitalic{C}{C}
\declarebsfitalic{D}{D}
\declarebsfitalic{E}{E}
\declarebsfitalic{F}{F}
\declarebsfitalic{G}{G}
\declarebsfitalic{H}{H}
\declarebsfitalic{I}{I}
\declarebsfitalic{J}{J}
\declarebsfitalic{K}{K}
\declarebsfitalic{L}{L}
\declarebsfitalic{M}{M}
\declarebsfitalic{N}{N}
\declarebsfitalic{O}{O}
\declarebsfitalic{P}{P}
\declarebsfitalic{Q}{Q}
\declarebsfitalic{R}{R}
\declarebsfitalic{S}{S}
\declarebsfitalic{T}{T}
\declarebsfitalic{U}{U}
\declarebsfitalic{V}{V}
\declarebsfitalic{X}{X}
\declarebsfitalic{Y}{Y}
\declarebsfitalic{Z}{Z}

\newcommand{\msA}{{\sf{A}}}

\newcommand{\msG}{{\sf{G}}}

\newcommand{\msU}{{\sf{U}}}

\newcommand{\fkg}{{\mathfrak{g}}}

\newcommand{\fki}{{\mathfrak{i}}}

\newcommand{\fkl}{{\mathfrak{l}}}
\newcommand{\fkm}{{\mathfrak{m}}}

\newcommand{\fks}{{\mathfrak{s}}}

\newcommand{\fkw}{{\mathfrak{w}}}

\newcommand{\bbC}{{\mathbb{C}}}

\newcommand{\bbF}{{\mathbb{F}}}

\newcommand{\bbN}{{\mathbb{N}}}

\newcommand{\bbR}{{\mathbb{R}}}

\newcommand{\bbZ}{{\mathbb{Z}}}

\newcommand{\scA}{{\matheul{A}}}

\newcommand{\scC}{{\matheul{C}}}

\newcommand{\scH}{{\matheul{H}}}

\newcommand{\scK}{{\matheul{K}}}
\newcommand{\scL}{{\matheul{L}}}
\newcommand{\scM}{{\matheul{M}}}

\newcommand{\scP}{{\matheul{P}}}
\newcommand{\scQ}{{\matheul{Q}}}

\newcommand{\clC}{{\mathcal{C}}}

\newcommand{\clG}{{\mathcal{G}}}

\newcommand{\clP}{{\mathcal{P}}}

\newcommand{\clS}{{\mathcal{S}}}

\newcommand{\clh}{{\mathcal{h}}}

\newcommand{\bra}[1]{{\langle{#1}|}}

\newcommand{\ket}[1]{{|{#1}\rangle}}

\newcommand{\braket}[2]{{\langle{#1}|{#2}\rangle}}

\newcommand{\laitrap}{\reflectbox{$\partial$}}      

\usepackage{accents}

\sectionfont{\fontsize{12}{15}\selectfont}
\subsectionfont{\fontsize{12}{15}\selectfont}

\usepackage{enumitem}


\DeclareMathOperator{\ran}{ran} 
\DeclareMathOperator{\id}{id}

\DeclareMathOperator{\Hom}{Hom}

\DeclareMathOperator{\Obj}{Obj}

\DeclareMathOperator{\Vrtc}{Vert}
\DeclareMathOperator{\End}{End}

\DeclareMathOperator{\tr}{tr}

\DeclareMathOperator{\UU}{\sf{U}}

\DeclareMathOperator{\Vrtx}{Vert}
\DeclareMathOperator{\Simp}{Simp}
\DeclareMathOperator{\lk}{lk}
\DeclareMathOperator{\sk}{sk}

\numberwithin{equation}{subsection} 
\numberwithin{subsection}{section} 

\newcommand{\ceqref}[1]{{\textcolor{blue}{\eqref{#1}}}}
\newcommand{\cref}[1]{{\textcolor{blue}{\ref{#1}}}}
\newcommand{\ccite}[1]{{\textcolor{blue}{\!\cite{#1}}}}

\newcommand{\sss}{{\hbox{\large $\sum$}}}
\newcommand{\ppp}{{\hbox{\large $\prod$}}}

\newcommand{\uuu}{{\hbox{\large $\bigcup$}}}
\newcommand{\nnn}{{\hbox{\large $\bigcap$}}}
\newcommand{\ddd}{{\hbox{\large $\bigoplus$}}}

\newcommand{\squuu}{{\hbox{\large $\bigsqcup$}}}

\newcommand{\ul}[1]{{\underline{#1}}}
\newcommand{\ol}[1]{{\overline{#1}}}

\newcommand{\hfpt}{\hspace{.75pt}}
\newcommand{\mhfpt}{\hspace{-.75pt}}

\newcommand\mycom[2]{\genfrac{}{}{-3pt}{}{#1}{#2}}

\font\euler=eusm10 at 12.8 truept
\font\scripteuler=eusm7
\font\scriptscripteuler=eusm5 
\def\eul{\fam=12}
\newcommand{\matheul}[1]{{{\eul #1}}}

\newtheorem{defi}{{\sf Definition}}[subsection]
\newtheorem{prop}{{\sf Proposition}}[subsection]
\newtheorem{theor}{{\sf Theorem}}[subsection]

\newtheorem{exa}{{\sf Example}}[subsection]

\DeclareMathSymbol{*}{\mathbin}{symbols}{"03} 

\begin{document}

\thispagestyle{empty} 

\vskip1.5cm
\begin{large}
{\flushleft\textcolor{blue}{\sffamily\bfseries A new quantum computational 
set--up for algebraic topology via simplicial sets}}
\end{large}
\vskip1.3cm
\hrule height 1.5pt
\vskip1.3cm
{\flushleft{\sffamily \bfseries Roberto Zucchini}\\
\it Department of Physics and Astronomy,\\
Alma Mater Studiorum University of Bologna,\\
I.N.F.N., Bologna division,\\
viale Berti Pichat, 6/2\\
Bologna, Italy\\
Email: \textcolor{blue}{\tt \href{mailto:roberto.zucchini@unibo.it}{roberto.zucchini@unibo.it}}, 
\textcolor{blue}{\tt \href{mailto:zucchinir@bo.infn.it}{zucchinir@bo.infn.it}}}


\vskip.7cm
\vskip.6cm 
{\flushleft\sc
Abstract:} 
In this paper, a quantum computational framework for algebraic 
topology based on simplicial set theory is presented.
This extends previous work, which was limited to simplicial complexes and aimed
mostly to topological data analysis. The proposed set--up applies to any parafinite simplicial set and
proceeds by associating with it a finite dimensional simplicial Hilbert space, whose simplicial operator structure
is studied in some depth. It is shown in particular
how the problem of determining the simplicial set's homology can be solved within the simplicial Hilbert
framework. Further, the conditions under which simplicial set theoretic algorithms can be
implemented in a quantum computational setting with finite resources are examined. 
Finally a quantum algorithmic scheme capable to compute the simplicial homology spaces and Betti numbers
of a simplicial set combining a number of basic quantum algorithms is outlined.

\vspace{2mm}
\par\noindent

\vfil\eject

{\color{blue}\tableofcontents}

\vfil\eject

\vfil\eject

\renewcommand{\sectionmark}[1]{\markright{\thesection\ ~~#1}}

\section{\textcolor{blue}{\sffamily Introduction}}\label{sec:intro}

Computational topology is a branch of topology which aims to the determination
of topological invariants of topological spaces using the methods of algebraic topology
and the algorithmic techniques
furnished by computer science. It comprises areas as diverse as computational
3--manifold theory, computational knot theory, computational homotopy theory and computational
homology theory.

In the last three decades, there has been a growing interest in computational topology 
in relation to topological data analysis, a method of analysing large, incomplete 
and noisy high dimensional sets of data by determining their intrinsic topological properties in general
and their homology in particular \ccite{Edelsbrunner:2002tps,Zomorodian:2005cph,Carlsson:2009tad,Zomorodian:tda2012}. 
The methodology has the potential of extracting meaningful information about data sets of interest
on one hand but poses formidable computational challenges on the other.

Quantum computing may provide
new powerful means to speedup the implementation of the algorithms of computational topology in general 
and of topological data analysis in particular. This paper aims to explore this possibility relying upon  the
methods of simplicial set theory.


\subsection{\textcolor{blue}{\sffamily Simplicial approaches to computational topology}}\label{subsec:motivation}

In computational topology, a wide range of topological spaces embedded
in Euclidean spaces are analyzed by means of suitable abstract simplicial complexes
associated to samplings of them, such as the \v Cech \ccite{Alexandroff:1924chc},
Vietoris--Rips \ccite{Vietoris:1927vrc} and witness complexes \ccite{de Silva:2004wwc} 
to mention the most popular. 

A simplicial complex is a set of vertices, edges, triangles, tetrahedrons and higher dimensional 
polytopes called collectively simplices fitting with each other through
their boundaries in a topologically meaningful manner\ccite{Lefschet:1930top}.
The simplices of the complex build up a topological space and the complex realizes a
generalized triangulation of such a space. 
The way the simplices of a simplicial complex join together is described combinatorially 
by an associated abstract simplicial complex\ccite{Spanier:1966att}. An abstract simplicial complex
is a family of non empty finite subsets of a vertex set, called abstract simplices, that contains all
the singleton sets of the vertices and is closed under subset taking.
Upon identifying abstract simplices with simplices, an abstract simplicial complex turns into a
combinatorial blueprint for a topological space, its topological realization. 

Abstract simplicial complex theory turns out to be particularly effective in the study of
embedded topological spaces. The representation of such spaces it provides is intuitive and
reasonably simple. However, by its very nature, it falls short of satisfying key requirements
and certain of its features limit its applicability to topological spaces of other types,
as summarized in the following points.

\begin{enumerate}[label=\roman*),font=\itshape]

\vspace{-1mm}\item The definitions of product of two simplicial complexes and quotient
of a simplicial complex by a subcomplex are elaborate and involved reflecting the intricacy
of making such operations enjoy desired compatibility properties with topological realization.


\vspace{-2mm}\item Simplices with identified faces cannot occur in a simplicial complex, restricting the range
of types of simplices available.

\vspace{-2mm}\item Distinct simplices in a simplicial complex cannot share the same set of faces, limiting the 
combinatorial range of simplicial complex theory. 

\vspace{-2mm}\item The simplicial complexes which are employed to describe even relatively simple
topological spaces contain as a rule a very large number of simplices.

\vspace{-2mm}\item The reduction methods which have been devised 
to curtail the size of these complexes while preserving their topological properties,
such as Whitehead's simplicial contraction \ccite{Whitehead:1939spn}, are often of
limited usefulness.

\vspace{-1mm}

\end{enumerate}

Alternative simplicial approaches to computational topology free of the shortcomings
pointed out above would deserve to be considered. Simplicial set theory \ccite{Eilenberg:sst1950}
is a generalization of simplicial complex theory which meets such demand.
As simplicial complexes, simplicial sets have topological realizations and can therefore
be employed as combinatorial models of topological spaces. However, 
simplicial sets enjoy the desirable properties listed below simplicial complexes do not.

\begin{enumerate}[label=\roman*),font=\itshape]

\vspace{-1mm}\item The product 
of two simplicial sets and the quotient of a simplicial set by a subset
can be defined straightforwardly and the simplicial sets resulting from such operations are compatible
with topological realization.

\vspace{-2mm}\item Simplices with identified faces are allowed in a simplicial set. 

\vspace{-2mm}\item Distinct simplices sharing the same set of faces can occur in a simplicial set.

\vspace{-2mm}\item Precisely because of the greater wealth of simplex types and simplex combinatorics,
simplicial sets furnish leaner and more succinct simplicial models of topological spaces.

\vspace{-2mm}\item A wider range of reduction techniques, such as 
simplicial contraction previously mentioned and simplicial collapse,
allow to further streamline such models.

\vspace{-1mm}

\end{enumerate}

The essential feature distinguishing simplicial sets from simplicial complexes is the incorporation 
of degenerate simplices. These are simplices 
whose formal dimension is higher than their effective one. The nicer properties of simplicial sets
as compared to simplicial complexes listed above are possible precisely thanks to the inclusion of degenerate
simplices. The reasons for this are multiple.

\begin{enumerate}[label=\roman*),font=\itshape]

\vspace{-1mm}\item A particular non degenerate simplex can have a degenerate face and be a face a
degenerate simplex.

\vspace{-2mm}\item Degenerate simplices allow for the gluing of a non degenerate simplex 
through its boundary to another non degenerate simplex of arbitrarily lower dimension.

\vspace{-2mm}\item Degenerate simplices enter in an essential way in the basic simplicial set theoretic operations
already mentioned. In a product of two simplicial sets, a non degenerate
simplex may have degenerate components. In a quotient of a simplicial set by a subset, non degenerate
simplices are replaced by degenerate ones. Other examples could be mentioned.

\vspace{-1mm}

\end{enumerate}

\noindent 
The degenerate simplices are therefore indispensable constitutive elements of a simplicial set along with the non degenerate
ones. This notwithstanding, degenerate simplices are hidden in the simplicial set's topological realization and further
they do not contribute to its homology. This however does not imply that
such simplices can simply be dropped.  In fact, the indiscriminate removal of the degenerate simplices
leaves in general an incomplete and/or inconsistent simplicial construct.

We illustrate the points made above with a few examples. 
A two dimensional torus can be represented as a simplicial set with one vertex, three edges
with coinciding ends 
and two triangles with coinciding vertices and sharing the same edges,
while as a simplicial complex as simple as possible requires seven vertices, twenty--one edges and fourteen 
triangles. The minimal simplicial complex describing a three dimensional sphere requires five vertices,
ten edges, ten triangles, five tetrahedrons, while as a simplicial set, only one vertex and one tetrahedron
with collapsed faces are sufficient.
These examples, albeit elementary, indicate that in general
the number of non degenerate simplices required to model a topological space as a simplicial set is indeed
considerably smaller than as a simplicial complex. 
The infinitely many degenerate simplices accompanying the non degenerate ones in the simplicial set
do not offset this advantage. The degenerate simplices are topologically invisible 
in the simplicial set and there are methods to effectively dispose of them.

It is important to realize that, in an appropriate sense, simplicial complexes are 
special cases of simplicial sets. 
A simplicial set describes a simplicial complex precisely when each non degenerate
simplex has distinct vertices and no two non degenerate simplices share the same vertices.
Instances of simplicial sets not satisfying these restrictive conditions are routine. 
Simplicial sets are therefore more general than simplicial complexes and for this reason
have in principle a broader scope. 


Our main proposition is that the investigation of the potential implementation of
simplicial set theoretic algorithms in computational topology
as an useful addition and complement to simplicial complex theoretic ones is a worthwhile endeavour.

\begin{enumerate}[label=\roman*),font=\itshape]

\vspace{-1mm}\item Simplicial sets have the potential of providing
a more efficient combinatorial codification of topological spaces 
ideally suited for algorithmic implementation.

\vspace{-2mm}\item
Techniques of 
computational topology employing simplicial sets and not just simplicial complexes
may have a wider range of applications. Indeed, simplicial set theory can be used to
describe and generalize a variety of combinatorial--topological structures such as
directed multigraphs, partially ordered sets, categories and more generally $\infty$--categories.

\vspace{-1mm}
  
\end{enumerate}

\noindent
This point of view has been forcefully  advocated by Perry \ccite{Perry:act2003} and Zomorodian
\cite{Zomorodian:tds2010} based on the above and similar considerations.
Independently, simplicial sets have found applications also
in computational geometry, a discipline distinct from but overlapping with computational topology
with relevant applications in geometric modelling, computer graphics, computer aided design and manufacturing. etc.
Indeed computational geometry has topological and metric aspects, the first of which are amenable to
the methods provided by simplicial set theory \ccite{Lang:1995mss,Lienhardt:1997sgm}



We conclude by observing that although most of its applications 
are concerned with homology computation, simplicial set theory enters noticeably also
in homotopy computation, see e.g. \ccite{Krcal:2013cht,Awodey:2015htt,Filakovsky:2018srh}.

 Basic notions and facts of simplicial set theory employed in this paper
are reviewed in sect. \cref{sec:sset} to reader's benefit and to fix notation and terminology.


\subsection{\textcolor{blue}{\sffamily A quantum framework for simplicial sets
}}\label{subsec:qufrmkssets}

The present endeavour is inspired by the seminal work \ccite{Lloyd:2014lgz} by Lloyd {\it et al}, 
where quantum algorithms for calculating  eigenvectors and eigenvalues of combinatorial
Laplacians and finding Betti numbers in persistent homology were developed achieving
an exponential speedup over the analogous classical algorithms
used in topological data analysis. 
This opened a new field in quantum computing, quantum topological data analysis, whose
development intensified in recent years
\ccite{Gunn:2019rqa,Gyurik:2020qat,Ubaru:2021dae,Hayakawa:2022qpb,McArdle:2022sqa,Berry:2022qqa,Black:2023isp}.
A critical evaluation of this promising quantum computational framework from the perspective
of complexity theory was presented in ref. \ccite{Schmidhuber:2022ctl}.
An interesting relationship between it 
and supersymmetric quantum mechanics was studied in refs. \ccite{Crichigno:2020vue,Cade:2021jhc,Crichigno:2022wyj}. 

The authors of the references cited in the previous paragraph
used quantum algorithms suited for simplicial complexes having in mind their potential applications to
topological data analysis. For the reasons explained in subsect. \cref{subsec:motivation} above, investigating
whether it is possible to adapt and extend such a quantum computational approach to simplicial sets
may be a worthy proposition.  In the present paper, we attempt to do so. 

We provide now a heuristic illustration of our formulation.
A simplicial set can be described as a collection of simplices subdivided according to their dimension
and of face and degeneracy maps which indicate which simplices are the faces and degeneracies of any
given simplex. In the quantum set--up we are going to present, 
a given simplicial set is inscribed in a {\it simplicial Hilbert space} by viewing the
simplices as simplex vectors forming a distinguished orthonormal basis and the face and degeneracy maps
as face and degeneracy operators acting on simplex vectors in a way that precisely correlates to that
the face and degeneracy maps act on simplices. 

The dagger structure of the simplicial Hilbert space set--up 
brings in the adjoints of the face and degeneracy operators. These encode relevant 
features of the underlying simplicial set and can be used to obtain alternative reformulations of standard
problems of computational topology. In particular, the problem of the determination of the simplicial
homology spaces with complex coefficients of the simplicial set
is recast as that of the computation of the kernels of certain simplicial Hilbert Laplacians
as in ref. \ccite{Lloyd:2014lgz}. 
The homologically irrelevant degenerate simplex subspaces can furthermore be effectively disposed of
in principle by employing appropriate search algorithms. 

The framework outlined above, which is expounded in sect. \cref{sec:setup}, affords also the derivation
of a number of technical results and leads to novel theoretical constructs. We introduce in particular
the notion of simplicial quantum circuit, a special kind of quantum circuit which performs
homological computations. A related approach was presented by Schreiber and Sati 
in ref. \ccite{Sati:2023xqk}

In sect. \cref{sec:simpcircappls}, we examine some of the problems which may arise in the implementation
of simplicial set theoretic topological algorithms in a realistic quantum computer relying on
the quantum computational set--up of sect. \cref{sec:setup}.  
The issues analyzed here range from truncation and skeletonization of simplicial sets, 
as a means of modelling finite storage resources capable of assembling simplicial data up a
certain finite dimension, and the digital encoding of a truncated simplicial set, by counting and parametrizing
of simplices, to their translation into a quantum computational framework.
Mainly for illustrative purposes, we also outline a quantum algorithmic scheme capable
in principle to compute the simplicial homology spaces and Betti numbers
of a simplicial set along the lines of that worked out in \ccite{Lloyd:2014lgz} combining
a number of basic quantum algorithms.

\subsection{\textcolor{blue}{\sffamily Scope and limitations of the present work}}\label{subsec:scope}

We conclude this introductory section with some remarks clarifying the scope and limitations of the  
present work. 

This paper is interdisciplinary in that it combines elements from computational algebraic topology and
quantum computation scattered in the literature proposing a unified view and a synthesis.
It hopefully is of some interest to algebraic topologists, willing to know how their discipline
may find application in quantum computational topology, and quantum computational topology specialists,
wishing to have an understanding of their field from a more foundational perspective.

The paper has a theoretical and physical mathematical outlook. It illustrates a quantum computational framework
for algebraic topology based on simplicial set theory, which may be viewed as an abstract model
of a {\it simplicial quantum computer}. No new 
quantum algorithms solving specific problems of algebraic topology
with a better performance than classical ones are presented, but hopefully the ground
is prepared for the future development and study of such algorithms. The paper contains also
original research work, but the results presented have
mostly a speculative bearing for the time being. Time will say whether these ideas will
turn out to be useful in practice. 

The paper focuses on homology for its practical
relevance and as an illustration of the efficacy of the formal apparatus devised,
though application of it to homotopy is conceivable and presumably attainable. 
The simplicial complex based quantum algorithm of ref.
\ccite{Lloyd:2014lgz} employed for homological computations in topological data analysis
and its derivations and refinements elaborated in later literature, however, 
will likely remain the most competitive option in the foreseeable future. 

In sect. \cref{sec:outlook}, we list the problems which are still open and
discuss the outlook of our work.



\noindent
\markright{\textcolor{blue}{\sffamily Conventions}} 


\noindent
\textcolor{blue}{\sffamily Conventions.}
Following a widely  used convention of algebraic topology and computer science,
we denote by $\bbN$ the set of all non negative integers. So, $0\in\bbN$.
We indi\-cate by $|A|$ the cardinality of a set $A$. To avoid possible confusion,
we denote the topological realization of a simplicial set $\sfX$ by ${}^\sharp\sfX$
rather than $|\sfX|$, as usually done in the mathematics.
$\Obj_\clC$ and $\Hom_\clC$ denote the object class and homomorphism  set of a category $\clC$.
Composition $\circ$ of maps, when occurring,  is usually left understood.
Finally, we adopt Dirac's bra--ket notation throughout this paper.

\vfil\eject

\renewcommand{\sectionmark}[1]{\markright{\thesection\ ~~#1}}

\section{\textcolor{blue}{\sffamily Simplicial sets}}\label{sec:sset}

In this section, we shall review the main aspects of simplicial set theory.
This topic has an elegant formulation in the framework of category theory. In what follows,
however, we shall pursue a more direct approach which by its combinatorial 
nature is especially suitable for the algorithmic methods of computational topology.
The presentation is kept as concrete as possible. Only a basic knowledge of category theory is assumed.
Introductory accounts of simplicial set theory are provided by refs. \ccite{Riehl:lis2008,Sergeraert:cht2008,
Friedman:sst2012,Bergner:2022scb}. 
Standard references on the subject are \ccite{Curtis:sht1971,May:sot1993}.

In intuitive terms, a simplicial set $\sfX$ is a collection of sets
$\sfX_n$ with $n\in\bbN$, whose elements
are to be thought of as $n$--simplices, equipped with mappings which establish:

\begin{enumerate}

\vspace{-.9mm} \item which $n-1$--simplices are faces of which $n$--simplices

\vspace{-1.8mm} \item which $n+1$--simplices are degeneracies of which $n$--simplices. 

\vspace{-.9mm}

\end{enumerate}

A simplicial set is an abstract combinatorial blueprint for constructing a topological space. Indeed,
every simplicial set $\sfX$ has a topological realization (more commonly but less precisely called
geometric realization) ${}^\sharp\sfX$, a topological space of a special kind called $CW$ complex. ${}^\sharp\sfX$
is built by associating with any $n$--simplex $\sigma_n$ of $\sfX$
a copy of the standard topological $n$--simplex
\begin{equation}
{}^\sharp\Delta^n=\left\{(t_0,\ldots,t_n)\in\bbR^{n+1}|0\leq t_i\leq 1,\sss_it_i=1\right\}
\label{}
\end{equation}
and then gluing together all the topological simplices generated in this fashion along their boundaries
in a way that precisely correlates to that the simplices of $\sfX$ which they are associated with are
related in consequence of the simplicial set's face and degeneracy maps. 

It is precisely because a topological space can be encoded in the simplicial set of which
it is the topological realization that many notions of ordinary algebraic topology have an analog
in simplicial set theory. In particular, the homotopy and homology of a topological space are modelled
by the homotopy and homology of the underlying simplicial set. In this paper, we concentrate on homology
for its interest and greater simplicity.

\vfill\eject

\subsection{\textcolor{blue}{\sffamily Simplicial sets}}\label{subsec:smplsets}

In this subsection, we review the main notions of simplicial set theory. 

\begin{defi} \label{def:sset}
A simplicial set $\sfX$ is a collection of 
sets $\sfX_n$, $n\in\bbN$,
and mappings $d_{ni}:\sfX_n\rightarrow \sfX_{n-1}$, $n\geq 1$, $i=0,\dots,n$, and 
$s_{ni}:\sfX_n\rightarrow \sfX_{n+1}$, $n\geq 0$, $i=0,\dots,n$, obeying the simplicial relations
\begin{subequations}
\label{smplsets1/5}
\begin{align}
&d_{n-1i}d_{nj}=d_{n-1j-1}d_{ni}&& \text{if ~$0\leq i,j\leq n$, $i<j$}, 
\label{smplsets1}
\\
&d_{n+1i}s_{nj}=s_{n-1j-1}d_{ni}&& \text{if ~$0\leq i,j\leq n$, $i<j$}, 
\label{smplsets2}
\\
&d_{n+1i}s_{nj}=\id_n&&\text{if ~$0\leq j\leq n$, $i=j,\,j+1$}, 
\label{smplsets3}
\\
&d_{n+1i}s_{nj}=s_{n-1j}d_{ni-1}&& \text{if ~$0\leq i,j\leq n+1$, $i>j+1$}, 
\label{smplsets4}
\\
&s_{n+1i}s_{nj}=s_{n+1j+1}s_{ni}&& \text{if ~$0\leq i,j\leq n$, $i\leq j$}. 
\label{smplsets5}
\end{align}
\end{subequations}
\end{defi}
\noindent
For each of these relations, there are restrictions on the range of allowed values of $n$ which follow from
$d_{ni}$ and $s_{ni}$ being defined for $n\geq 1$ and $n\geq 0$, respectively.
These conditions are evident from inspection and will not be stated explicitly.
By the third relation, the maps $d_{ni}$ are surjective whilst the maps $s_{ni}$ are injective.

The integer $n$ is called simplicial degree.
The set $\sfX_n$ comprises the $n$-simplices of $\sfX$. The maps $d_{ni}$, $s_{ni}$ are the face
and degeneracy maps of $\sfX_n$.  In the mathematical 
literature, the dependence
of $d_{ni}$, $s_{ni}$ on $n$ is usually left implicit. In the applications to quantum computation treated
in this paper, this is not always possible without yielding ambiguous or incomplete expressions.
We have therefore decided to indicate it explicitly at the cost of somewhat complicating
the notation.

The simplicial set can be represented diagrammatically as
\begin{equation}
\xymatrix@C=2.5pc{
&\cdots \ar@/^.6pc/[r]\ar@/^-.6pc/[r] \ar@/^.2pc/[r]\ar@/^-.2pc/[r]
&\sfX_2\ar@/^.4pc/[l]\ar@/^-.4pc/[l]\ar[l]
\ar[r]\ar@/^.6pc/[r]\ar@/^-.6pc/[r]
&\sfX_1\ar@/^.3pc/[l]\ar@/^-.3pc/[l]
\ar@/^.6pc/[r]\ar@/^-.6pc/[r]&\sfX_0\ar[l]},
\end{equation}
where the rightward/leftward arrows stand for the face/degeneracy maps. 

\begin{defi} \label{def:smor}
A morphism $\phi:\sfX\rightarrow\sfX'$ of simplicial sets $\sfX$, $\sfX'$
consists in a collection of maps $\phi_n:\sfX_n\rightarrow\sfX'{}_n$ with $n\geq 0$ obeying 
\begin{subequations}
\label{smplsets6/7}
\begin{align}&\phi_{n-1}d_{ni}=d'{}_{ni}\phi_n && \text{if ~$0\leq i\leq n$}, 
\label{smplsets6}
\\
&\phi_{n+1}s_{ni}=s'{}_{ni}\phi_n && \text{if ~$0\leq i\leq n$}.
\label{smplsets7}
\end{align}
\end{subequations}
\end{defi}

\noindent
The morphism fits in a commutative diagram of the form 
\begin{equation}
\vspace{.125cm}
\xymatrix@C=2.5pc{
&\cdots \ar@/^.6pc/[r]\ar@/^-.6pc/[r] \ar@/^.2pc/[r]\ar@/^-.2pc/[r]
&\sfX_2\ar@/^.4pc/[l]\ar@/^-.4pc/[l]\ar[l]
\ar[r]\ar@/^.6pc/[r]\ar@/^-.6pc/[r]\ar[d]^{\phi_0}
&\sfX_1\ar@/^.3pc/[l]\ar@/^-.3pc/[l]
\ar@/^.6pc/[r]\ar@/^-.6pc/[r]\ar[d]^{\phi_1}&\sfX_0\ar[l]\ar[d]^{\phi_0}\\
&\cdots \ar@/^.6pc/[r]\ar@/^-.6pc/[r] \ar@/^.2pc/[r]\ar@/^-.2pc/[r]
&\sfX'{}_2\ar@/^.4pc/[l]\ar@/^-.4pc/[l]\ar[l]
\ar[r]\ar@/^.6pc/[r]\ar@/^-.6pc/[r]
&\sfX'{}_1\ar@/^.3pc/[l]\ar@/^-.3pc/[l]
\ar@/^.6pc/[r]\ar@/^-.6pc/[r]&\sfX'{}_0\ar[l]}.
\vspace{.125cm}
\end{equation}


Simplicial sets can be formed using other simplicial sets as building blocks
via certain elementary operations. Two such operations will be relevant in our analysis. 

Let $\sfX$, $\sfX'$ be simplicial sets.

\begin{defi} \label{def:simplprod}
The Cartesian product $\sfX\times\sfX'$ of $\sfX$, $\sfX'$
is the simplicial set defined as follows. The $n$--simplex set of $\sfX\times\sfX'$ is the Cartesian
product $\sfX\times\sfX'{}_n=\sfX_n\times\sfX'{}_n$. The face and degeneracy maps of $\sfX\times\sfX'$
at degree $n$ are the Cartesian product maps $d\times d'{}_{ni}=d_{ni}\times d'{}_{ni}$ and
$s\times s'{}_{ni}=s_{ni}\times s'{}_{ni}$.
The Cartesian product of two morphisms $\phi:\sfX\rightarrow \sfX''$, $\psi:\sfX'\rightarrow \sfX'''$,
of simplicial sets is the simplicial set morphism $\phi\times\psi:\sfX\times\sfX'\rightarrow\sfX''\times\sfX'''$
defined by $\phi\times\psi_n=\phi_n\times\psi_n$ at degree $n$. 
\end{defi}

\noindent
Above, $\times$ denotes the Cartesian multiplication monoidal product of the category $\ul{\rm Set}$
of sets and functions. 

\begin{defi} \label{def:simplcup}
The disjoint union $\sfX\sqcup\sfX'$  of $\sfX$, $\sfX'$ is the simplicial set 
defined as follows. The $n$--simplex set of $\sfX\sqcup\sfX'$ is the disjoint union 
$\sfX\sqcup\sfX'{}_n=\sfX_n\sqcup\sfX'{}_n$.
The face and degeneracy maps of $\sfX\sqcup\sfX'$
at degree $n$ are the disjoint union maps $d\sqcup d'{}_{ni}=d_{ni}\sqcup d'{}_{ni}$ and
$s\sqcup s'{}_{ni}=s_{ni}\sqcup s'{}_{ni}$.
The disjoint union of two morphisms $\phi:\sfX\rightarrow \sfX''$, $\psi:\sfX'\rightarrow \sfX'''$,
of simplicial sets is the simplicial set morphism $\phi\sqcup\psi:\sfX\sqcup\sfX'\rightarrow\sfX''\sqcup\sfX'''$
defined by $\phi\sqcup\psi_n=\phi_n\sqcup\psi_n$ at degree $n$. 
\end{defi}

\noindent
Here, $\sqcup$ denotes the disjoint union monoidal product of the category $\ul{\rm Set}$.

With the operations of Cartesian product and disjoint union
simplicial sets and morphisms form a bimonoidal category $\ul{\rm sSet}$.
We shall treat $\ul{\rm sSet}$ as a strict bimonoidal category, not completely rigorously,
neglecting the fact the Cartesian product and disjoint union of sets and functions 
are associative and unital only up to natural isomorphism only. 

\begin{defi} \label{def:parafinite}
A simplicial set $\sfX$ is called parafinite if the $n$--simplex set $\sfX_n$ is finite
for all $n$.
\end{defi}

\noindent
{\it In this paper, we shall deal mainly with such simplicial sets}. They form a full bimo\-noidal subcategory
$\ul{\rm pfsSet}$  of $\ul{\rm sSet}$.

The following seemingly trivial instance of simplicial set 
is together with some of its variants sometimes useful in some constructions.

\begin{exa} \label{ex:discretesset} The discrete simplicial set of an ordinary set.

\noindent
{\rm 
Any ordinary non empty set $A$ can be identified with the simplicial set $\sfD A$
with $n$-simplex set $\sfD_nA=A$ and face and degeneracy maps $d_{ni}=s_{ni}=\id_A$. 
Such a simplicial set is called {\it discrete} and sometimes also {\it simplicially constant}.
If the set $A$ is finite, then $\sfD A$ is parafinite as a simplicial set.
}
\end{exa}

While the constructions elaborated in this work apply to any parafinite simplicial set,  
there are specific simplicial sets for which they exhibit special properties. 
Nerves of categories and simplicial sets of ordered simplicial complexes are among these.

\begin{exa}\label{ex:nerve} The nerve of a category.

\noindent  
{\rm The {\it nerve} $\sfN\clC$ of a category $\clC$ is defined as follows. The $0$--simplex set of $\sfN\clC$
is just the set of objects of $\clC$: $\sfN_0\clC=\Obj_\clC$. Thus, a $0$--simplex is just an object $x$
of $\clC$. For $n\geq 1$, the $n$--simplex set of $\sfN\clC$ consists of the ordered $n$ element sequences
of composable morphisms of $\clC$:  
$\sfN_n\clC=\Hom_\clC\times_{\Obj_\clC}\cdots\times_{\Obj_\clC}\Hom_\clC$ ($n$ factors). Therefore, an $n$--simplex
$\sigma_n\in\sfN_n\clC$ is representable as
\begin{equation}
\sigma_n=(f_1,\ldots,f_n), 
\label{smplsets8}
\end{equation}
where $f_1,\ldots,f_n\in\Hom_\clC$ are morphisms such that $t(f_k)=s(f_{k+1})$ for $1\leq k\leq n-1$,
$s$, $t$ denoting the source and target maps of $\clC$.
The face maps $d_{ni}:\sfN_n\clC\rightarrow\sfN_{n-1}\clC$ read as follows
\begin{align}
&d_{10}(f_1)=t(f_1), \quad d_{11}(f_1)=s(f_1) 
\label{smplsets9}
\\
\intertext{for $n=1$ and}
&d_{n0}(f_1,\ldots,f_n)=(f_2,\ldots,f_n), 
\label{smplsets10}
\\
&d_{ni}(f_1,\ldots,f_n)=(f_1,\ldots,f_{i-1},f_{i+1}\circ f_i,f_{i+2},\ldots,f_n)\quad\text{for $0<i<n$},
\nonumber
\\
&d_{nn}(f_1,\ldots,f_n)=(f_1,\ldots,f_{n-1}) 
\nonumber
\intertext{for $n\geq 2$. The degeneracy maps $s_{ni}:\sfN_n\clC\rightarrow\sfN_{n+1}\clC$ take the form}
&s_{00}x=\id_x
\label{smplsets11}
\\
\intertext{for $n=0$ and}
&s_{ni}(f_1,\ldots,f_n)=(f_1,\ldots,f_i,\id_{t(f_i)},f_{i+1},\ldots,f_n)
\label{smplsets12}
\end{align}
for $n\geq 1$.
A category $\clC$ is called finite if its object and morphism sets are finite. In that case, its nerve
$\sfN\clC$ is a parafinite simplicial set.

A groupoid is a category $\clG$ all of whose morphisms are invertible. The nerve $\sfN\clG$
of a groupoid $\clG$ exhibits as a consequence special properties not found in generic categories.
}
\end{exa}

\begin{exa} \label{ex:simp} The simplicial set of an ordered abstract simplicial complex.

\noindent 
{\rm An {\it abstract simplicial complex} $\clS$ consists of a set of vertices, $\Vrtx_\clS$,
and a set of simplices, $\Simp_\clS$, constituted by finite non-empty subsets of $\Vrtx_\clS$
satisfying the following requirements. (1) If $\sigma\in \Simp_\clS$ and $\emptyset\neq\tau\subseteq\sigma$,
then $\tau\in \Simp_\clS$, that is any non empty subset of a simplex is a simplex. 
(2) If $v\in \Vrtx_\clS$, then $\{v\}\in \Simp_\clS$, so every singleton of a vertex is a simplex. 
An $n$--simplex is a simplex of $\Simp_\clS$ formed by $n+1$ distinct vertices of $\Vrtx_\clS$. 
$\clS$ is said to be {\it ordered} if $\Vrtx_\clS$ \linebreak is endowed with a total ordering $\leq$.
An $n$--simplex $\sigma_n$ is then representable as an increasing sequence of $n+1$
vertices: $\sigma_n=(v_0,\ldots,v_n)$ with $v_0,\ldots,v_n\in\Vrtx_\clS$ and $v_0<\ldots<v_n$. 

With any ordered abstract simplicial complex $\clS$, there is associated a simplicial set $\sfK\clS$
defined as follows. For $n\geq 1$, the $n$--simplex set consists of the non 
decreasing sequences of $n+1$ vertices of $\clS$ whose underlying vertex set is a simplex of $\clS$.
Therefore, an $n$--simplex $\sigma_n\in\sfK_n\clS$ has a representation of the form
\begin{equation}
\sigma_n=(v_0,\ldots,v_n),
\label{smplsets13}
\end{equation}
where $v_0,\ldots,v_n\in\Vrtx_\clS$ are vertices with $v_0\leq\ldots\leq v_n$ and $|\sigma_n|\in \Simp_\clS$,
$|\sigma_n|$ being the set constituted by the distinct vertices $v_k$. 
The face maps $d_{ni}:\sfK_n\clS\rightarrow\sfK_{n-1}\clS$ are given by the expression 
\begin{align}
&d_{ni}(v_0,\ldots,v_n)=(v_0,\ldots v_{i-1},v_{i+1},\ldots,v_n).
\label{smplsets14}
\intertext{The degeneracy maps $s_{ni}:\sfK_n\clS\rightarrow\sfK_{n+1}\clS$ read as follows}
&s_{ni}(v_0,\ldots,v_n)=(v_0,\ldots v_i,v_i,\ldots,v_n).
\label{smplsets15}
\end{align}
$\clS$ is called finite when $\Vrtx_\clS$ is a finite set. In that case, its simplicial set 
$\sfK\clS$ is parafinite. 
}
\end{exa}

A distinguishing feature of simplicial sets when compared to simplicial complexes
is the appearance of degenerate simplices.

Let $\sfX$ be a simplicial set.

\begin{defi} \label{defi:degsiplx}
 An $n$--simplex $\sigma_n\in\sfX_n$ is said to be degenerate if there is some $n-1$--simplex $\tau_{n-1}\in\sfX_{n-1}$
 and index $i$ with $0\leq i\leq n-1$ such that $\sigma_n=s_{n-1i}\tau_{n-1}$.
\end{defi}

\noindent $0$-simplices are regarded as non degenerate. We denote by ${}^s\sfX_n$
the subset of the degenerate simplices of $\sfX_n$.

\begin{exa} \label{ex:discretedeg} The degenerate simplices of the discrete simplicial set of a set.

\noindent
{\rm In the discrete simplicial set $\sfD A$ of a non empty set $A$ (cf. ex. \cref{ex:discretesset})
${}^s\sfD_nA=\sfD_nA$ for $n>0$: all positive degree simplices are degenerate.
}
\end{exa}

\begin{exa}\label{ex:nervedeg} The degenerate simplices of the nerve of a category.

\noindent  
{\rm For $n>0$, the degenerate $n$--simplex set ${}^s\sfN_n\clC$ of the nerve 
$\sfN\clC$ of a category $\clC$ (cf. ex. \cref{ex:nerve}) consists of all simplices
$(f_1,\ldots,f_n)$ at least one of whose components $f_i$ is an identity morphisms
}
\end{exa}

\begin{exa} \label{ex:simpdeg} The degenerate simplices of the simplicial set of a simplicial complex.

\noindent 
{\rm For $n>0$, the degenerate $n$--simplex set ${}^s\sfK_n\clS$ of the simplicial set $\sfK\clS$ of
an ordered abstract simplicial complex $\clS$ (cf. ex. \cref{ex:simp}) consists of all simplices
$(v_0,\ldots,v_n),$ at least two of whose components $v_i$ are equal.
}
\end{exa}

Note that a degenerate simplex can have a non degenerate face and viceversa.



\subsection{\textcolor{blue}{\sffamily Simplicial objects}}\label{subsec:smplobjs}

Simplicial objects in a general category $\clC$ are defined in analogy to and generalize
simplicial sets.


\begin{defi} \label{def:sobj}
A simplicial object $\sfX$ in $\clC$ is a collection of non empty objects $\sfX_n$, $n\in\bbN$,
of $\clC$ and morphisms $d_{ni}:\sfX_n\rightarrow \sfX_{n-1}$, $n\geq 1$, $i=1,\dots,n$, and 
$s_{ni}:\sfX_n\rightarrow \sfX_{n+1}$, $n\geq 0$, $i=1,\dots,n$, of $\clC$ obeying the simplicial relations
\ceqref{smplsets1/5}. 
\end{defi}

\noindent The $d_{ni}$, $s_{ni}$ are called face and degeneracy morphism of $\sfX$. 

Instances of simplicial objects are encountered in many areas of mathematics.

\begin{exa} A simplicial set.

\noindent
{\rm A simplicial set $\sfX$ is just a simplicial object in the category
$\underline{\rm Set}$ of sets and functions.}
\end{exa}  

\begin{exa} \label{exa:simpgr} A simplicial group.

\noindent
{\rm A simplicial group $\sfX$
is a simplicial object in the category $\underline{\rm Grp}$ of groups and group homomorphisms for which the objects
$\sfX_n$ are groups and the face and degeneracy morphisms are group morphisms between them.}
\end{exa}

\begin{exa} A simplicial manifold.

\noindent
{\rm A simplicial manifold $\sfX$
is a simplicial object in the category $\underline{\rm Mnfld}$ of smooth manifolds and manifold maps for which the objects
$\sfX_n$ are smooth manifolds and the face and degeneracy morphisms are smooth maps between them.}
\end{exa}

In this paper, we shall deal specifically with simplicial Hilbert spaces.
This kind of simplicial objects 
emerge quite natural in the construction of sect. \cref{sec:setup}.

\begin{exa} A simplicial Hilbert space.

\noindent
{\rm  A simplicial Hilbert space is a simplicial object in
the category $\underline{\rm Hilb}$
of finite dimensional Hilbert spaces and linear maps.}
\end{exa}

There exists an obvious notion of simplicial object morphism which generalizes that of simplicial set morphism
of def. \cref{def:smor}. 

\begin{defi} \label{def:sobjmor}
A morphism $\phi:\sfX\rightarrow\sfX'$ of simplicial objects $\sfX$, $\sfX'$ of $\clC$
is a collection of morphisms $\phi_n:\sfX_n\rightarrow\sfX'{}_n$ with $n\geq 0$ obeying
the relations \ceqref{smplsets6/7}. 
\end{defi}

\noindent
Note that every simplicial object in a concrete category $\clC$ is also a simplicial set.
In this paper, we consider mainly simplicial objects of this kind.


\subsection{\textcolor{blue}{\sffamily Simplicial homology}}\label{subsec:simplhomol}

Homology is a basic structure playing an important role in our analysis. We introduce the notion
first from a purely algebraic point of view. Later, we concentrate on simplicial homology.

\begin{defi} An abstract chain complex $(\sfA,\delta)$ 
is a sequence of Abelian groups and group morphisms of the form 
\begin{equation}
\xymatrix@C=2.5pc
{\cdots
\ar[r]^{\delta_3}
&\sfA_2
\ar[r]^{\delta_2}
&\sfA_1\ar[r]^{\delta_1}&\sfA_0}
\label{nwhom1}
\end{equation}
such that the homological relations 
\begin{equation}
\delta_n\delta_{n+1}=0
\label{nwhom2}
\end{equation}
with $n\geq 1$ are satisfied.
\end{defi}

\noindent
The index $n$ labelling the segments of the sequence \ceqref{nwhom1} is called chain degree. $\sfA_n$
and $\delta_n$ are named respectively chain group and boundary morphism at degree $n$. 

By virtue of \ceqref{nwhom2}, we have that $\ran\delta_{n+1}\subseteq\ker\delta_n$. The sequence \ceqref{nwhom1}
would be exact if $\ran\delta_{n+1}=\ker\delta_n$, but this is not the case in general.  
The homology of $(\sfA,\delta)$ measures the failure of \ceqref{nwhom1} to be exact.

\begin{defi} For $n\geq 0$, the homology group of degree $n$ of the chain complex $(\sfA,\delta)$ is
the quotient group
\begin{equation}
\rmH_n(\sfA,\delta)=\ker\delta_n/\ran\delta_{n+1}.
\label{nwhom3}
\end{equation}
\end{defi}

\noindent Above, it is conventionally assumed that $\ker\delta_0=\sfA_0$. The homology
groups $\rmH_n(\sfA,\delta)$ constitute the homology $\rmH(\sfA,\delta)$ of $(\sfA,\delta)$.

There is a natural notion of morphism of chain complexes.  

\begin{defi} \label{def:chcxmorph}
A morphism $g:(\sfA,\delta)\rightarrow(\sfA',\delta')$ of chain complexes 
is a diagram of Abelian groups and group morphisms of the form
\begin{equation}
\xymatrix@C=2.5pc@R=2.pc
{\cdots
\ar[r]^{\delta_3}
&\sfA_2\ar[r]^{\delta_2}\ar[d]^{g_2}
&\sfA_1\ar[r]^{\delta_1}\ar[d]^{g_1}
&\sfA_0\ar[d]^{g_0}\\
\cdots
\ar[r]^{\delta'{}_3}
&\sfA'{}_2
\ar[r]^{\delta'{}_2}
&\sfA'{}_1\ar[r]^{\delta'{}_1}&\sfA'{}_0
}
\label{nwhom4}
\end{equation}
obeying the commutativity conditions 
\begin{equation}
g_{n-1}\delta_n=\delta'{}_ng_n
\label{nwhom5}
\end{equation}
with $n\geq 1$. 
\end{defi}

\noindent
Chain complexes and chain complex morphisms form so a category. 

A chain complex morphism $g:(\sfA,\delta)\rightarrow(\sfA',\delta')$ induces a 
homology morphisms $g_{*n}:\rmH_n(\sfA,\delta)\rightarrow\rmH_n(\sfA',\delta')$ for each $n$.
Homology has thus functorial properties. 


Cohomology is the dual notion of homology. All the basic definitions of cohomology can be obtained
by those of homology roughly by inverting all the arrows. In particular, for any
abstract cochain complex $(\sfM,\sigma)$ one can define the cohomology groups $\rmH^n(\sfM,\sigma)$ for
any degree $n\geq 0$. We leave to the reader the straightforward task of spelling this out in full detail. 

The above homological algebraic framework can be employed in simplicial set theory leading to the
formulation of simplicial homology.  

Let $\sfG$ be an Abelian simplicial group (cf. subsect. \cref{subsec:smplobjs}). Using the face maps of
$\sfG$ as constitutive elements, one can construct a series of Abelian group morphisms 
$d_n:\sfG_n\rightarrow \sfG_{n-1}$ for $n\geq 1$ viz
\begin{align}
&d_n=\mycom{{}_\sss}{{}_{0\leq i\leq n}}(-1)^id_{ni}.
\label{nwhom11}
\end{align}
These Abelian groups and group morphisms fit in the diagram 
\begin{equation}
\xymatrix@C=2.5pc
{\cdots
\ar[r]^{d_3}&\sfG_2\ar[r]^{d_2}&\sfG_1\ar[r]^{d_1}&\sfG_0}.
\label{nwhom16}
\end{equation}
By the simplicial relations \ceqref{smplsets1}, the $d_n$ obey further the homological relations 
\begin{align}
&d_nd_{n+1}=0,  
\label{nwhom13}
\end{align}
By \ceqref{nwhom13}, the diagram \ceqref{nwhom16} constitutes a chain complex $(\sfG,d)$,
the simplicial chain complex of $\sfG$. For each $n$, $\sfG_n$ and $d_n$ are respectively
the group of simplicial chains and the simplicial boundary morphism at degree $n$. 

The homology $\rmH(\sfG)\equiv\rmH(\sfG,d)$ is called the simplicial homology of $\sfG$.
$\rmH(\sfG)$ characterizes $\sfG$ and provides valuable information on the data which underlie
and specify it. 

Denote by ${}^s\sfG_n$ the subgroup of $\sfG_n$ generated by the degenerate $n$--simplices
(cf. subsect. \cref{subsec:smplsets}).
The quotient group $\ol{\sfG}_n=\sfG_n/{}^s\sfG_n$ is the normalized $n$--chain group. 
Owing to relations \ceqref{smplsets2}--\ceqref{smplsets4}, $d_n{}^s\sfG_n\subseteq{}^s\sfG_{n-1}$ and 
therefore an induced map $\ol{d}_n:\ol{\sfG}_n\rightarrow\ol{\sfG}_{n-1}$ is defined.
$\ol{d}_n$ obeys the homological relation \ceqref{nwhom13}. The normalized chain complex 
\begin{equation}
\xymatrix@C=2.7pc
{\cdots
\ar[r]^{\ol{d}_3}
&\ol{\sfG}_2
\ar[r]^{\ol{d}_2}
&\ol{\sfG}_1\ar[r]^{\ol{d}_1}&\ol{\sfG}_0}
\label{simplhomol6}
\end{equation}
is in this way constructed. The homology $\rmH(\ol{\sfG})\equiv\rmH(\ol{\sfG},\ol{d})$ of $(\ol{\sfG},\ol{d})$
is the normalized simplicial homology of $\sfG$. The following theorem establishes that normalized simplicial homology
is just another incarnation of simplicial homology. 

\begin{theor} \label{pop:eilmclane}
(Normalization theorem \ccite{Eilenberg:1953norm}) The simplicial homology $\rmH(\sfG)$
and normalized simplicial homology $\rmH(\ol{\sfG})$ of $\sfG$ are isomorphic. 
\end{theor}

\noindent 
with $\rmC(\sfX,\msA)$ replaced 
This result reveals that the degenerate chains are homologically irrelevant in the computation
of the simplicial homology $\rmH(\sfG)$ and can be used in principle to simplify the computation of this latter.

Let $\sfX$ be a simplicial set and $\msA$ an Abelian group. For $n\in\bbN$, the 
group of $n$--chains of $\sfX$ with coefficients in $\msA$ is the Abelian group 
\begin{equation}
\sfC_n(\sfX,\msA)=\bbZ[\sfX_n]\otimes\msA,
\label{simplhomol1}
\end{equation}
where $\bbZ[\sfX_n]$ denotes the free Abelian group generated by the $n$--simplex set $\sfX_n$.
The face and degeneracy maps $d_{ni}$, $s_{ni}$ of $\sfX$ extend uniquely to Abelian group
morphisms $d_{ni}:\sfC_n(\sfX,\msA)\rightarrow\sfC_{n-1}(\sfX,\msA)$,
$s_{ni}:\sfC_n(\sfX,\msA)\rightarrow\sfC_{n+1}(\sfX,\msA)$. These extensions obey the simplicial
relations \ceqref{smplsets1/5}. So, the groups $\sfC_n(\sfX,\msA)$ and the morphisms $d_{ni}$, $s_{ni}$
build up a simplicial Abelian group $\sfC(\sfX,\msA)$. By the general construction described above,
it is then possible to construct via \ceqref{nwhom11} boundary morphisms 
\begin{equation}
\partial_n\sigma_n=\mycom{{}_\sss}{{}_{0\leq i\leq n}}(-1)^id_{ni}\sigma_n
\label{simplhomol2}
\end{equation}
obeying the homological relations
\begin{equation}
\partial_n\partial_{n+1}=0.
\label{simplhomol3}
\end{equation}
With $\sfX$ and $\msA$, so, there is associated a chain complex $(\sfC(\sfX,\msA),\partial)$.

\begin{defi} \label{def:simplhomolx}
The simplicial homology $\rmH(\sfX,\msA)$ of $\sfX$ with coefficients in $\msA$ is the simplicial homology
$\rmH(\sfC(\sfX,\msA))$ of the simplicial Abelian group $\sfC(\sfX,\msA)$. 
The normalized simplicial homology $\ol{\rmH}(\sfX,\msA)$ of $\sfX$ with coefficients in $\msA$ 
is the associated normalized simplicial homology $\rmH(\ol{\sfC}(\sfX,\msA))$. 
\end{defi}

By the normalization theorem \ceqref{pop:eilmclane}, the simplicial and normalized simplicial homologies
$\rmH(\sfX,\msA)$, $\ol{\rmH}(\sfX,\msA)$ are isomorphic. The computation of the latter is however generally 
simpler than that of the former. For this reason, simplicial homology is sometimes defined directly
as normalized simplicial homology in the mathematical literature. 

\begin{exa}\label{ex:hnerve} The homology of the delooping of a group. 

\noindent
{\rm A group $\msG$ can be regarded as a groupoid $\rmB\msG$ with a single object whose morphisms are
in one--to--one correspondence with the elements of $\msG$, the so called delooping of $\msG$.
The simplicial homology of the nerve $\sfN\rmB\msG$ of $\rmB\msG$ with coefficients in $\bbZ$,
$\rmH(\sfN\rmB\msG,\bbZ)$ can be shown to be isomorphic to the group homology $\rmH(\msG)$. 
}
\end{exa}

\begin{exa} \label{ex:hsimp} The homology of the simplicial set of an ordered abstract simplicial complex.

\noindent
{\rm It is possible to define the simplicial homology of a simplicial complex with a given coefficient Abelian group
on the same lines as that of a simplicial set. The simplices of an ordered simplicial complex $\clS$ are precisely the non  
degenerate simplices of the associated simplicial set $\sfK\clS$ studied in ex. \cref{ex:simp}.
By the normalization 
theorem \cref{pop:eilmclane}, the simplicial homology $\rmH(\clS,\msA)$ is then isomorphic
to the simplicial homology $\rmH(\sfK\clS,\msA)$.}  
\end{exa}

The above set--up can be refined by working with vector spaces or modules instead than Abelian groups
in obvious fashion, resulting in (co)homology spaces or modules, respectively.

The Betti numbers with coefficients in a field $\bbF$ of a simplicial set $\sfX$ are
\begin{equation}
\beta_n(\sfX,\bbF)=\dim\rmH_n(\sfX,\bbF).
\label{betti}
\end{equation} 
The Betti numbers are topological invariants of the topological space ${}^\sharp\sfX$
realizing $\sfX$ and characterize it based on its connectivity. 
For $\bbF=\bbR,\bbC$ and low degrees, the Betti numbers have simple intuitive topological interpretations:
$\beta_0(\sfX,\bbF)$ is the number of connected components of ${}^\sharp\sfX$,
$\beta_1(\sfX,\bbF)$ is the number of holes of ${}^\sharp\sfX$, 
$\beta_2(\sfX,\bbF)$ is the number of voids of ${}^\sharp\sfX$ etc. 
If ${}^\sharp\sfX$ is a $d$--dimensional topological manifold, $\beta_n(\sfX,\bbF)=0$
for $n>d$.

Much of computational topology aims to the computation of
the Betti numbers for the important topological information they furnish about ${}^\sharp\sfX$.  
For instance, finding out that the Betti numbers $\beta_n(\sfX,\bbR)$ vanish for $n>d$ for some $d$
is indication that ${}^\sharp\sfX$ is a $d$--dimensional manifold. 


\vfill\eject

\renewcommand{\sectionmark}[1]{\markright{\thesection\ ~~#1}}

\section{\textcolor{blue}{\sffamily Quantum simplicial set framework}}\label{sec:setup} 

In this section we shall work out and study in detail the quantum simplicial set theoretic set--up 
outlined in the introduction. 

The quantum simplicial set framework furnishes a natural backdrop for the theoretical analysis
and eventual implementation of simplicial quantum algorithms for computational topology. It is essentially an instance
of quantum basis encoding of classical data, where the latter are just basic simplicial data. By virtue
of it, a given parafinite simplicial set is encoded into a finite dimensional simplicial Hilbert
space much as a qubit register is into a finite dimensional Hilbert space. Correspondingly, the simplices and face and
degeneracy maps of the simplicial set convert into the basis vectors and face and degeneracy operators of the simplicial
Hilbert space. The Hilbert dagger structure provides however us also with the adjoints of these operators,
making possible novel constructions. 

Extending the quantum simplicial framework beyond the range of parafinite simplicial sets is not feasible
because the intrinsic limitations of implementable computation: any conceivable simplicial computer
can operate only on a finite number of simplices of each given degree. Therefore, parafinite simplicial sets
are the most general kind of simplicial sets that can be handled by such a device.

A criticism that can be levelled at the framework is that it still incorporates the simplex Hilbert spaces storing
the simplicial data of all degrees, while a realistic computer can manage the simplicial data up to a finite maximum 
degree. In fact, doing so is only a convenient abstraction to simplify the analysis. 
In sect. \cref{sec:simpcircappls}, we shall show how to deal this problem by simplicial truncation or skeletonization
of the underlying simplicial set. 

The quantum simplicial framework will enable us to analyze the simplicial homology of the simplicial Hilbert space,
show how it is controlled by appropriate simplicial Hilbert Laplacians and prove its isomorphism 
to the simplicial homology of the underlying simplicial set. It will also gives us the means to construct
the simplicial Hilbert space's appropriate form of normalized simplicial homology and show its isomorphism to the
simplicial set's normalized simplicial homology. Last but not least, it will suggest us a suitable notion of simplicial
quantum circuits as the kind of quantum circuits capable of performing simplicial computations. Though we concentrate
on homology, applications also to homotopy are conceivable. 








\subsection{\textcolor{blue}{\sffamily Hilbert space encoding of a simplicial set}}\label{subsec:qusmplx}

In this subsection, we shall show how a given parafinite simplicial set $\sfX$
(cf. defs. \cref{def:sset}, \cref{def:parafinite}) can be encoded in a 
simplicial finite dimensional Hilbert structure. 

\begin{defi}
 For $n\in\bbN$, the $n$--simplex Hilbert space $\scH_n$ is the Hilbert space generated 
 by the $n$--simplex set $\sfX_n$.
\end{defi}

\noindent
$\scH_n$ has thus a canonical orthonormal basis $\ket{\sigma_n}$ labelled by the $n$--simplices $\sigma_n\in \sfX_n$.
In the following, we shall refer to the basis $\ket{\sigma_n}$
as the $n$--simplex basis. It plays a role analogous to the computational basis
in familiar quantum computing.
Since $\sfX_n$ is a finite set, $\scH_n$ is finite dimensional.

The face and degeneracy maps $d_{ni}$ and $s_{ni}$ induce face and degeneracy operators
characterized by their action on the vectors of the $n$--simplex basis.
\begin{defi}
The face operators 
$D_{ni}:\scH_n\rightarrow \scH_{n-1}$, $i=0,\dots,n$ and $n\geq 1$,
and degeneracy operators $S_{ni}:\scH_n\rightarrow \scH_{n+1}$ and $i=0,\dots,n$ and $n\geq 0$.
$D_{ni}$, $S_{ni}$ are 
\begin{align}
&D_{ni}=\mycom{{}_\sss}{{}_{\sigma_n\in \sfX_n}}\ket{d_{ni}\sigma_n}\hfpt\bra{\sigma_n},
\label{qusmplx1}
\\
&S_{ni}=\mycom{{}_\sss}{{}_{\sigma_n\in \sfX_n}}\ket{s_{ni}\sigma_n}\hfpt\bra{\sigma_n}.
\label{qusmplx2}
\end{align}
\end{defi}

The Hilbert dagger structure of the $\scH_n$ allows us to define the adjoint operators  
$D_{ni}{}^+:\scH_{n-1}\rightarrow \scH_n$, $S_{ni}{}^+:\scH_{n+1}\rightarrow \scH_n$
of $D_{ni}$, $S_{ni}$. They are given by 
\begin{align}
&D_{ni}{}^+
=\mycom{{}_\sss}{{}_{\sigma_{n-1}\in \sfX_{n-1}}}\mycom{{}_\sss}{{}_{\omega_n\in\sfD_{ni}(\sigma_{n-1})}}\ket{\omega_n}\hfpt\bra{\sigma_{n-1}},
\label{qusmplx3}
\\
&S_{ni}{}^+
=\mycom{{}_\sss}{{}_{\sigma_{n+1}\in \sfX_{n+1}}}\mycom{{}_\sss}{{}_{\omega_n\in\sfS_{ni}(\sigma_{n+1})}}\ket{\omega_n}\hfpt\bra{\sigma_{n+1}},
\label{qusmplx4}
\end{align}
where $\sfD_{ni}(\sigma_{n-1})$, $\sfS_{ni}(\sigma_{n+1})\subset \sfX_n$ are the $n$-simplex subsets
\begin{align}
&\sfD_{ni}(\sigma_{n-1})=\{\omega_n\in \sfX_n|d_{ni}\omega_n=\sigma_{n-1}\},
\label{qusmplx5}
\\
&\sfS_{ni}(\sigma_{n+1})=\{\omega_n\in \sfX_n|s_{ni}\omega_n=\sigma_{n+1}\}.  
\label{qusmplx6}
\end{align}
An important problem of the theory is the determination of the content of 
$\sfD_{ni}(\sigma_{n-1})$, $\sfS_{ni}(\sigma_{n+1})$ for any assignment of $\sigma_{n-1},\,\sigma_{n+1}$.
Through the $\sfD_{ni}(\sigma_{n-1})$, $\sfS_{ni}(\sigma_{n+1})$,
the adjoint operators $D_{ni}{}^+$, $S_{ni}{}^+$ encode basic features of and provide important
information about the underlying simplicial set $\sfX$. Not much can be said about
$\sfD_{ni}(\sigma_{n-1}),\,\sfS_{ni}(\sigma_{n+1})$ in general. The following general properties
hold anyway. 
Since $d_{ni}$ is surjective while $s_{ni}$ generally is not, 
$\sfD_{ni}(\sigma_{n-1})$ is always non empty while $\sfS_{ni}(\sigma_{n+1})$ may be empty.
Further, as $s_{ni}$ is injective while $d_{ni}$ generally is not, $|\sfS_{ni}(\sigma_{n+1})|\leq 1$
while $|\sfD_{ni}(\sigma_{n-1})|\geq 1$. 

The basic simplicial relations \ceqref{smplsets1/5}
obeyed by the face and degeneracy maps $d_{ni}$, $s_{ni}$ imply that 
the face and degeneracy operators $D_{ni}$, $S_{ni}$ and their adjoints $D_{ni}{}^+$, $S_{ni}{}^+$
satisfy a host of {\it exchange identities} relating products of pairs of these latter.
These characterize to a considerable extent
the simplicial Hilbert framework constructed in this paper 
and are therefore analyzed in detail in the rest of this subsection. 

The exchange identities come in pairs related by adjunction.
We shall show explicitly only one the two relations of each pair leaving to the reader the rather straightforward
task of writing down the other. 

The exchange identities involving the $D_{ni}$, $S_{ni}$ ensue directly from \ceqref{qusmplx1}, \ceqref{qusmplx1}
and the simplicial relations \ceqref{smplsets1/5}
\begin{subequations}
\label{qusmplx7/11}
\begin{align}
&D_{n-1i}D_{nj}-D_{n-1j-1}D_{ni}=0&& \text{for ~$0\leq i,j\leq n$, $i<j$}, 
\label{qusmplx7}
\\
&D_{n+1i}S_{nj}-S_{n-1j-1}D_{ni}=0&& \text{for ~$0\leq i,j\leq n$, $i<j$}, 
\label{qusmplx8}
\\
&D_{n+1i}S_{nj}=1_n&&\text{for ~$0\leq j\leq n$, $i=j,\,j+1$}, 
\label{qusmplx9}
\\
&D_{n+1i}S_{nj}-S_{n-1j}D_{ni-1}=0&& \text{for ~$0\leq i,j\leq n+1$, $i>j+1$}, 
\label{qusmplx10}
\\
&S_{n+1i}S_{nj}-S_{n+1j+1}S_{ni}=0&& \text{for ~$0\leq i,j\leq n$, $i\leq j$},
\label{qusmplx11}
\end{align}
\end{subequations}
where $1_n=1_{\scH_n}$. 
The \ceqref{qusmplx7/11} are identical in form to the simplicial relations
and for this reason are referred to as the simplicial Hilbert identities \ceqref{smplsets1/5}. 
The exchange relations involving the $D_{ni}{}^+$, $S_{ni}{}^+$
are obtained from the \ceqref{qusmplx7/11} by adjunction. They are given by 
the \ceqref{qusmplx7/11} except for the order of the factors of each term, which is reversed. 
They are formally equal to the simplicial theoretic cosimplicial relations
and so are called the cosimplicial Hilbert identities.

The exchange identities involving one operator of each of the operator sets $D_{ni}$, $S_{ni}$
and $D_{ni}{}^+$, $S_{ni}{}^+$ are not so easily obtained and do not take a form analogous to that of
the simplicial and cosimplicial Hilbert identities. These mixed identities exhibit however
an analogous structure. Explicitly, they read as 
\begin{subequations}
\label{qusmplx12/15}
\begin{align}
&D_{ni}{}^+D_{nj}-D_{n+1j+1}D_{n+1i}{}^+=\varDelta^{DD}{}_{nij}&& \text{for ~$0\leq i,j\leq n$, $i\leq j$}, 
\label{qusmplx12}
\\
&D_{n+2i}{}^+S_{nj}-S_{n+1j+1}D_{n+1i}{}^+=\varDelta^{DS}{}_{nij}&& \text{for ~$0\leq i,j\leq n$, $i\leq j$}, 
\label{qusmplx13}
\\
&S_{n-2i}{}^+D_{nj}-D_{n-1j-1}S_{n-1i}{}^+=\varDelta^{SD}{}_{nij}&& \text{for ~$0\leq i,j\leq n$, $i+1<j$}, 
\label{qusmplx14}
\\
&S_{ni}{}^+S_{nj}-S_{n-1j-1}S_{n-1i}{}^+=\varDelta^{SS}{}_{nij}&& \text{for ~$0\leq i,j\leq n$, $i<j$}, 
\label{qusmplx15}
\end{align}
\end{subequations}
where
\begin{subequations}
\label{qusmplx16/19}
\begin{align}
&\varDelta^{DD}{}_{nij}
=-\mycom{{}_\sss}{{}_{\sigma_n\in \sfX_n}}\mycom{{}_\sss}{{}_{\omega_n\in\sfD_{ni}(d_{nj}\sigma_n)}}
\ket{\omega_n}(|\sfD_{n+1i}(\sigma_n)\cap\sfD_{n+1j+1}(\omega_n)|-1)\bra{\sigma_n},
\label{qusmplx16}
\\
&\varDelta^{DS}{}_{nij}
=-\mycom{{}_\sss}{{}_{\sigma_n\in \sfX_n}} \mycom{{}_\sss}{{}_{\omega_{n+2}\in\sfD_{n+2i}(s_{nj}\sigma_n)}}
\ket{\omega_{n+2}}(|\sfD_{n+1i}(\sigma_n)\cap\sfS_{n+1j+1}(\omega_{n+2})|-1)\bra{\sigma_n},
\label{qusmplx17}
\\
&\varDelta^{SD}{}_{nij}
=-\mycom{{}_\sss}{{}_{\sigma_n\in \sfX_n}} \mycom{{}_\sss}{{}_{\omega_{n-2}\in\sfS_{n-2i}(d_{nj}\sigma_n)}}
\ket{\omega_{n-2}}(|\sfS_{n-1i}(\sigma_n)\cap\sfD_{n-1j-1}(\omega_{n-2})|-1)\bra{\sigma_n},
\label{qusmplx18}
\\
&\varDelta^{SS}{}_{nij}
=-\mycom{{}_\sss}{{}_{\sigma_n\in \sfX_n}}\mycom{{}_\sss}{{}_{\omega_n\in\sfS_{ni}(s_{nj}\sigma_n)}}
\ket{\omega_n}(|\sfS_{n-1i}(\sigma_n)\cap\sfS_{n-1j-1}(\omega_n)|-1)\bra{\sigma_n}.
\label{qusmplx19}
\end{align}
\end{subequations}
Recall that $|I|$ denotes the cardinality of a set $I$.
The expressions shown result from straightforward computations relying solely on the expressions
\ceqref{qusmplx1}, \ceqref{qusmplx2} and \ceqref{qusmplx3}, \ceqref{qusmplx4} of the operators
$D_{ni}$, $S_{ni}$ and $D_{ni}{}^+$, $S_{ni}{}^+$ and the basic simplicial relations \ceqref{smplsets1/5}.


\begin{defi}
The defects of the simplicial Hilbert structure are the 
four operators
$\varDelta^{DD}{}_{nij}:\scH_n\rightarrow \scH_n$, $0\leq i,j\leq n$, $i\leq j$,
$\varDelta^{DS}{}_{nij}:\scH_n\rightarrow \scH_{n+2}$, $0\leq i,j\leq n$, $i\leq j$,
$\varDelta^{SD}{}_{nij}:\scH_n\rightarrow \scH_{n-2}$, $0\leq i,j\leq n$, $i+1<j$, 
and $\varDelta^{SS}{}_{nij}:\scH_n\rightarrow \scH_n$, $0\leq i,j\leq n$, $i<j$,
given by the \ceqref{qusmplx16/19}. 
\end{defi}

A basic problem of the theory is determining under which conditions some or all 
defects vanish and identifying the simplicial sets for which such conditions are
fulfilled. For the $SS$ defects $\varDelta^{SS}{}_{nij}$, the problem however does not arise 
by the following theorem. 

\begin{theor}\label{prop:vanidect1} (No degeneracy defect theorem) We have
\begin{align}
&\varDelta^{SS}{}_{nij}=0&& \text{for ~$0\leq i,j\leq n$, $i<j$}.
\label{qusmplx20}
\end{align}
\end{theor}

\begin{proof}
See app. \cref{app:vanidect1} for the proof.
\end{proof}

The defects $\varDelta^{DD}{}_{nij}$, $\varDelta^{DS}{}_{nij}$, $\varDelta^{SD}{}_{nij}$ conversely are non zero in general. 
The exchange identities of the operators $S_{ni}$, $S_{ni}{}^+$ have in this way a simple
universal form akin to that of simplicial and cosimplicial Hilbert identities. 
The exchange identities of the operators $D_{ni}$, $D_{ni}{}^+$, $S_{ni}$, $D_{ni}{}^+$ and $D_{ni}$, $S_{ni}{}^+$ 
instead do not.

The concrete form the $\varDelta^{DD}{}_{nij}$, $\varDelta^{DS}{}_{nij}$, $\varDelta^{SD}{}_{nij}$
take depends on the underlying simplicial set $\sfX$. 
In fact, they encode features of $\sfX$ not deducible from
the simplicial relations and hence specific of $\sfX$.
For distinguished instances of simplicial sets, there exist perfectness results
establishing the vanishing of some of these defects. In this regard, the following definition is apposite. 


\begin{defi}\label{def:perfect}
The simplicial set $\sfX$ is said to be semi perfect if
\begin{align}
&\varDelta^{DS}{}_{nij}=0 && \text{for ~$0\leq i,j\leq n$, $i\leq j$},
\label{qusmplx22}
\\
&\varDelta^{SD}{}_{nij}=0 && \text{for ~$0\leq i,j\leq n$, $i+1<j$}.
\label{qusmplx23}
\intertext{$\sfX$ is said to be quasi perfect if \ceqref{qusmplx22}, \ceqref{qusmplx23} hold and further}
&\varDelta^{DD}{}_{nij}=0 && \text{for ~$0\leq i,j\leq n$, $i<j$}.
\label{qusmplx21}
\intertext{$\sfX$ is said to be perfect if \ceqref{qusmplx22}, \ceqref{qusmplx23} hold and further}
&\varDelta^{DD}{}_{nij}=0 && \text{for ~$0\leq i,j\leq n$, $i\leq j$}.
\label{qusmplx24}
\end{align}
\end{defi}

\begin{exa}\label{ex:nerveperf} The nerve of a category.

\noindent  
{\rm Nerves of categories are a special type of simplicial sets introduced in ex. \cref{ex:nerve}.

\begin{prop} (Perfectness proposition for nerves of categories) \label{prop:vanidect2a}
  The nerve $\sfN\clC$ of a finite category $\clC$ is quasi perfect. If $\clC$ is also 
  a groupoid then $\sfN\clC$ is perfect.
\end{prop}

\begin{proof}
See app. \cref{app:vanidect2} for the proof. 
\end{proof}
}
\end{exa}    

\begin{exa} \label{ex:simpperf} The simplicial set of an ordered abstract simplicial complex.

\noindent 
{\rm
Simplicial sets of an ordered abstract simplicial complex are another distinctive
type of simplicial sets introduced in ex. \cref{ex:simp}.

\begin{prop} (Perfectness proposition for simplicial sets of simplicial complexes) \label{prop:vanidect2b}
 The simplicial set $\sfK\clS$ of an ordered finite abstract simplicial complex $\clS$ is semi perfect.
\end{prop}

\begin{proof}
See again app. \cref{app:vanidect2} for the proof. 
\end{proof}
}
\end{exa}    

The mixed exchange identities \ceqref{qusmplx12/15} do not cover all possible products of one of the operators 
$D_{ni}$, $S_{ni}$ and one of the adjoint operators $D_{ni}{}^+$, $S_{ni}{}^+$.
The missing products are $D_{n+1i}D_{n+1i}{}^+$, $0\leq i\leq n+1$, $D_{n+2i+1}{}^+S_{ni}$, $0\leq i\leq n$,
$S_{n+1i}D_{n+1i}{}^+$, $0\leq i\leq n+1$ and $S_{ni}{}^+S_{ni}$, $0\leq i\leq n$. Some of 
these products will reemerge as elemental terms in the expression of the simplicial Hilbert Laplacians
studied in subsect. \cref{subsec:smplxlapl}. 

Every morphism $\phi:\sfX\rightarrow\sfX'$ of parafinite simplicial sets $\sfX$, $\sfX'$
(cf. def. \cref{def:smor}) also has a simplicial Hilbert encoding. 


\begin{defi} \label{def:morphop}
The morphism operators $\varPhi_n:\scH_n\rightarrow\scH'{}_n$, $n\geq 0$, of $\phi$ 
are 
\begin{equation}
\varPhi_n=\mycom{{}_\sss}{{}_{\sigma_n\in \sfX_n}}\ket{\phi_n\sigma_n}\bra{\sigma_n}.
\label{hsmor0}
\end{equation}
\end{defi}

By the relations \ceqref{smplsets6/7} and \ceqref{qusmplx1}, \ceqref{qusmplx2}, the $\varPhi_n$ satisfy 
\begin{subequations}
\label{hsmor1/2}
\begin{align}
&\varPhi_{n-1}D_{ni}-D'{}_{ni}\varPhi_n=0 && \text{if ~$0\leq i\leq n$}, 
\label{hsmor1}
\\
&\varPhi_{n+1}S_{ni}-S'{}_{ni}\varPhi_n=0 && \text{if ~$0\leq i\leq n$}.
\label{hsmor2}
\end{align}
\end{subequations}
The \ceqref{hsmor1/2} are identical in form to the simplicial morphism relations
and for this reason are referred to as the simplicial Hilbert morphism identities.

Again, the Hilbert dagger structure of the $\scH_n$, $\scH'{}_n$ allows us to define the adjoint operators  
$\varPhi_n{}^+:\scH'{}_n\rightarrow\scH_n$, $n\geq 0$, 
which in terms of the simplex basis read 
\begin{equation}
\varPhi_n{}^+=\mycom{{}_\sss}{{}_{\sigma'{}_n\in \sfX'{}_n}}
\mycom{{}_\sss}{{}_{\omega_n\in\sfX_n,\phi_n\omega_n=\sigma'{}_n}}\ket{\omega_n}\bra{\sigma'{}_n}\pagebreak 
\label{hsmor3}
\end{equation}
The $\varPhi_n{}^+$ obey the identities following from the \ceqref{hsmor1/2} by adjunction.
They have the same form as the \ceqref{hsmor1/2} except for the reversed order of the factors 
and are therefore called cosimplicial Hilbert morphism relations.


\subsection{\textcolor{blue}{\sffamily The simplicial Hilbert space of a simplicial set }}\label{subsec:hilbfunct}

In this subsection, we shall show how the Hilbert space encoding of a simplicial set described in subsect.
\cref{subsec:qusmplx} can be naturally represented as a simplicial Hilbert space. 

The category $\ul{\rm fdsHilb}$ of finite dimensional simplicial Hilbert spaces and operators
is described as follows. An object $\scH$ of $\ul{\rm fdsHilb}$ is a collection $\{\scH_n, D_{ni}, S_{ni}\}$
consisting of finite dimensional Hilbert spaces $\scH_n$ together with  face and degeneracy operators 
$D_{ni}:\scH_n\rightarrow\scH_{n-1}$, $S_{ni}:\scH_n\rightarrow\scH_{n+1}$
obeying the simplicial Hilbert identities \ceqref{qusmplx7/11}. A morphisms
$\varPhi:\scH\rightarrow\scH'$ of $\ul{\rm fdsHilb}$ is a collection of linear operators
$\varPhi_n:\scH_n\rightarrow\scH'{}_n$ satisfying the simplicial Hilbert morphism identities 
\ceqref{hsmor1/2}.

The category $\ul{\rm fdsHilb}$ is bimonoidal, its two monoidal products being given by
degreewise direct product and sum $\otimes$ and $\oplus$. Explicitly, $\otimes$ and $\oplus$
act as follows. 
Let $\scH$, $\scH'$ be 
finite dimensional simplicial Hilbert spaces. Then, $\scH\otimes\scH'$ is the simplicial Hilbert space
with $\scH\otimes\scH'{}_n=\scH_n\otimes\scH'{}_n$, $D\otimes D'{}_{ni}=D_{ni}\otimes D'{}_{ni}$,
$S\otimes S'{}_{ni}=S_{ni}\otimes S'{}_{ni}$. 
Similarly, $\scH\oplus\scH'$ is the simplicial Hilbert space
with $\scH\oplus\scH{}_n=\scH_n\oplus\scH'{}_n$, $D\oplus D'{}_{ni}=D_{ni}\oplus D'{}_{ni}$,
$S\oplus S'{}_{ni}=S_{ni}\oplus S'{}_{ni}$. Let $\varPhi:\scH\rightarrow\scH''$,
$\varPhi':\scH'\rightarrow\scH'''$ be morphisms of finite dimensional simplicial Hil\-bert spaces.
Then,  $\varPhi\otimes\varPhi':\scH\otimes\scH'\rightarrow\scH''\otimes\scH'''$ is the
simplicial Hilbert space morphism given by 
$\varPhi\otimes\varPhi'{}_n=\varPhi_n\otimes\varPhi'{}_n$. Equally, $\varPhi\oplus\varPhi':
\scH\oplus\scH'\rightarrow\scH''\oplus\scH'''$ is the simplicial Hilbert space morphism such that
$\varPhi\oplus\varPhi'{}_n=\varPhi_n\oplus\varPhi'{}_n$.  

We shall treat $\ul{\rm fdsHilb}$ as a strict bimonoidal category, not completely rigorously,
neglecting the fact the direct multiplication and summation of Hilbert spaces and operators
are associative and unital only up to natural isomorphisms only. 

In subsect. \cref{subsec:qusmplx}, 
we have detailed a construction  that associates with every parafinite simplicial set $\sfX$ the 
simplex Hilbert spaces $\scH_n$ and the face and degeneracy operators $D_{ni}$, $S_{ni}$
given by eqs. \ceqref{qusmplx1}, \ceqref{qusmplx2}.  
The simplicial Hilbert identities \ceqref{qusmplx7/11} obeyed by the $D_{ni}$, $S_{ni}$ 
entail that the Hilbert data collection $\{\scH_n, D_{ni}, S_{ni}\}$
constitutes a finite dimensional simplicial Hilbert space $\scH$.
We have also shown that with any morphism $\phi:\sfX\rightarrow\sfX'$ of parafinite simplicial sets
there are associated linear operators $\varPhi_n$ mapping $\scH_n$ to  $\scH'{}_n$ given by eqs.
\ceqref{hsmor0}. 
By the simplicial Hilbert morphism relations \ceqref{hsmor1/2} the $\varPhi_n$ obey,
the operator collection $\{\varPhi_n\}$ forms 
a morphism $\varPhi:\scH\rightarrow\scH'$ of the simplicial Hilbert spaces
$\scH$, $\scH'$ of $\sfX$, $\sfX'$.
The correspondence $\sfX\mapsto\scH$ and $\phi:\sfX\rightarrow\sfX'\mapsto \varPhi:\scH\rightarrow\scH'$
is further compatible with morphism composition
and identity assignment. We have thus constructed a functor from the category $\ul{\rm pfsSet}$ of parafinite
simplicial sets into the category of finite dimensional simplicial Hilbert spaces $\ul{\rm fdsHilb}$. 

\begin{defi} \label{def:hsqfunct}
The simplicial Hilbert functor is the functor $\clh:\ul{\rm pfsSet}\rightarrow\ul{\rm fdsHilb}$
described in the previous paragraph.
\end{defi}

The functor $\clh$ enjoys a nice property.

\begin{theor} \label{prop:hqsfbimod}
The simplicial Hilbert functor $\clh$ is bimonoidal.
\end{theor}

\noindent
Essentially, this states that $\clh$  maps the bimonoidal product structure of $\ul{\rm pfsSet}$,
consisting of Cartesian multiplication and disjoint union (cf. subsect.
\cref{subsec:smplsets}), into that of $\ul{\rm fdsHilb}$, comprising direct multiplication and summation.
The most significant categorical features of parafinite
simplicial sets are so reproduced in the appropriate form by finite dimensional simplicial Hilbert spaces.

\begin{proof}
We provide only a sketch of the proof. 
Let $\sfX$, $\sfX'$ be 
simplicial sets and
let $\scH=\clh(\sfX)$, $\scH'=\clh(\sfX')$ be their simplicial Hilbert spaces. 
Consider now the Cartesian product $\sfX\times\sfX'$ of $\sfX$, $\sfX'$
(cf. def. \cref{def:simplprod}) and its associated simplicial Hilbert space
$\scH\times\scH':=\clh(\sfX\times\sfX')$. There exists a simplicial Hilbert isomorphism
$\varLambda:\scH\times\scH'\xrightarrow{\simeq}\scH\otimes\scH'$ of $\scH\times\scH'$ 
and the direct product $\scH\otimes\scH'$ of $\scH$, $\scH'$.
At degree $n$, $\varLambda$ is defined by $\varLambda_n\ket{\sigma_n,\sigma'{}_n}=\ket{\sigma_n}\otimes\ket{\sigma'{}_n}$
with $\sigma_n\in\sfX_n$, $\sigma'{}_n\in\sfX'{}_n$. 
This furnishes the identification $\scH\times\scH'\simeq\scH\otimes\scH'$
by which we conclude that $\clh(\sfX\times\sfX')\simeq\clh(\sfX)\otimes\clh(\sfX')$.
Consider likewise the disjoint union $\sfX''=\sfX\sqcup\sfX'$ of $\sfX$, $\sfX'$ (cf. def. \cref{def:simplcup})
and its associated simplicial Hilbert space $\scH\sqcup\scH':=\clh(\sfX\sqcup\sfX')$. 
There exists a simplicial Hilbert isomorphism
$\varSigma:\scH\sqcup\scH'\xrightarrow{\simeq}\scH\oplus\scH'$ of 
$\scH\sqcup\scH'$  and the direct sum $\scH\oplus\scH'$ of $\scH$, $\scH'$.
At degree $n$, $\varSigma$ is defined by the expressions
$\varSigma_n\ket{\sigma_n}=\ket{\sigma_n}\oplus 0$, 
$\varSigma_n\ket{\sigma'{}_n}=0\oplus\ket{\sigma'{}_n}$
for $\sigma_n\in\sfX_n$, $\sigma'{}_n\in\sfX'{}_n$ respectively. This leads to the 
identification $\scH\sqcup\scH'\simeq\scH\oplus\scH'$ from which it is concluded that 
$\clh(\sfX\sqcup\sfX')\simeq\clh(\sfX)\oplus\clh(\sfX')$.

Let $\phi:\sfX\rightarrow\sfX''$, $\phi':\sfX'\rightarrow\sfX'''$ be morphisms of simplicial sets and
let $\varPhi=\clh(\phi)$, $\varPhi'=\clh(\phi')$ be the associated simplicial Hilbert space morphisms, so that, 
setting $\scH=\clh(\sfX)$, $\scH'=\clh(\sfX')$,
$\scH''=\clh(\sfX'')$, $\scH'''=\clh(\sfX''')$, we have $\varPhi:\scH\rightarrow\scH''$, \linebreak  
$\varPhi':\scH'\rightarrow\scH'''$. Consider now the Cartesian product
$\varphi\times\varphi':\sfX\times\sfX'\rightarrow\varphi':\sfX''\times\sfX'''$ and disjoint union of
$\varphi\sqcup\varphi':\sfX\sqcup\sfX'\rightarrow\varphi':\sfX''\sqcup\sfX'''$
of $\varphi$, $\varphi'$ (cf. defs. \cref{def:simplprod} and \cref{def:simplcup})
and their associated simplicial Hilbert space morphisms
$\varPhi\times\varPhi':=\clh(\varphi\times\varphi')$ and $\varPhi\sqcup\varPhi':=\clh(\varphi\sqcup\varphi')$.
Then, we have $\varPhi\times\varPhi':\scH\times\scH'\rightarrow\scH''\times\scH'''$ and 
$\varPhi\sqcup\varPhi':\scH\sqcup\scH'\rightarrow\scH''\sqcup\scH'''$. Using the isomorphisms
we introduced in the previous paragraph, 
$\varLambda:\scH\times\scH'\xrightarrow{\simeq}\scH\otimes\scH'$,
$\varLambda':\scH''\times\scH'''\xrightarrow{\simeq}\scH''\otimes\scH'''$
and $\varSigma:\scH\sqcup\scH'\xrightarrow{\simeq}\scH\oplus\scH'$,
$\varSigma:\scH''\sqcup\scH'''\xrightarrow{\simeq}\scH''\oplus\scH'''$, 
we find then that $\varPhi\times\varPhi'\simeq\varPhi\otimes\varPhi'$ and
$\varPhi\sqcup\varPhi'\simeq\varPhi\oplus\varPhi'$. It follows that $\clh(\varphi\times\varphi')\simeq
\clh(\varphi)\otimes\clh\varphi')$ and $\clh(\varphi\sqcup\varphi')\simeq\clh(\varphi)\oplus\clh(\varphi')$
as required. 
\end{proof} 

The category $\ul{\rm fdcsHilb}$ of finite dimensional cosimplicial Hilbert spaces and operators 
is defined analogously to its simplicial counterpart.
An object $\scH$ of $\ul{\rm fdsHilb}$ is a collection $\{\scH_n, D_{cni}, S_{cni}\}$
consisting of finite dimensional Hilbert spaces $\scH_n$ and coface and codegeneracy operators 
$D_{cni}:\scH_{n-1}\rightarrow\scH_n$, $S_{cni}:\scH_{n+1}\rightarrow\scH_n$
obeying obeying the cosimplicial Hilbert identities, relations of the same form as the
\ceqref{qusmplx7/11} except for the order of the factors which is inverted. A morphisms
$\varPhi:\scH'\rightarrow\scH$ of $\ul{\rm fdsHilb}$ is a collection of linear operators
$\varPhi_{cn}:\scH'{}_n\rightarrow\scH_n$ satisfying the cosimplicial Hilbert morphism identities,
relations of the same form as the \ceqref{hsmor1/2} except again for the reversed order of the factors.

Just as $\ul{\rm fdsHilb}$, the category $\ul{\rm fdcsHilb}$ is bimonoidal, its two monoidal products
being given again by degreewise direct product and sum $\otimes$ and $\oplus$. The explicit expressions  
$\otimes$ and $\oplus$ take in $\ul{\rm fdcsHilb}$ is essentially the same as that they do in
$\ul{\rm fdsHilb}$, which was detailed above.

It is easy to see that the dagger autofunctor ${}^+$ of the finite dimensional Hilbert space category
$\ul{\rm fdHilb}$, which implements operator adjunction, 
induces a dagger isofunctor ${}^+:\ul{\rm fdsHilb}\rightarrow\ul{\rm fdcsHilb}^{\rm op}$, where
the superscript ${}^{\rm op}$ denotes opposite of a category. ${}^+$ associates with every simplicial Hilbert space
$\scH=\{\scH_n, D_{ni}, S_{ni}\}$ its adjoint cosimplicial Hilbert space $\scH^+=\{\scH_n, D_{ni}{}^+, S_{ni}{}^+\}$
and with every simplicial Hilbert operator $\varPhi=\{\varPhi_n\}$ 
its adjoint cosimplicial Hilbert operator $\varPhi^+=\{\varPhi_n{}^+\}$ 
with $\varPhi^+:\scH'^+\rightarrow\scH^+$ if $\varPhi:\scH\rightarrow\scH'$. 
The functor ${}^+$ preserves direct multiplication and summation and is therefore bimonoidal.
The categories $\ul{\rm fdsHilb}$, $\ul{\rm fdcsHilb}^{\rm op}$ can therefore be identified.

We can compose the simplicial Hilbert functor $\clh$ and the dagger isofunctor ${}^+$ just introduced
to obtain a functor $\clh_c:\ul{\rm pfsSet}\rightarrow\ul{\rm fdcsHilb}^{\rm op}$
from the parafinite simplicial set category $\ul{\rm pfsSet}$ into the opposite finite
dimensional cosimplicial Hilbert space category $\ul{\rm fdcsHilb}^{\rm op}$.
By means of $\clh_c$, we can associate with any parafinite simplicial set $\sfX$ a finite dimensional cosimplicial 
Hilbert space $\scH^+$, the adjoint of the simplicial Hilbert space $\scH$ assigned to $\sfX$
by $\clh$. Similarly, we can associate with any morphism 
$\phi:\sfX\rightarrow\sfX'$ of parafinite simplicial sets a cosimplicial Hilbert operator
$\varPhi^+:\scH'^+\rightarrow\scH^+$, the adjoint of the simplicial Hilbert operator
$\varPhi:\scH\rightarrow\scH'$ assigned to $\phi$ by $\clh$. We shall call $\clh_c$ the cosimplicial Hilbert functor. 

In this way, parafinite simplicial sets have both a simplicial and a cosimplicial encoding
related by the dagger isofunctor.
This dagger structure provides the formal framework for the
analysis of the unitarity of the simplicial Hilbert operators associated with the simplicial set morphisms
arising in reversible simplicial computation (cf. subsect. \cref{subsec:simpcirc} below). 

We conclude this subsection with the following remark. 
Let $\sfX$ be a parafinite simplicial set with associated simplicial Hilbert space $\scH$. 
Suppose that the Hilbert structure of the spaces $\scH_n$ is forgotten so that the $\scH_n$ are regarded 
just as sets and the operators $D_{ni}$, $S_{ni}$ as maps.
Then, the simplicial Hilbert encoding maps $\varkappa_n:\sfX_n\rightarrow\scH_n$ given by
$\varkappa_n(\sigma_n)=\ket{\sigma_n}$ are the components of a simplicial set monomorphism
$\varkappa:\sfX\rightarrow\scH$ as follows from \ceqref{qusmplx1}, \ceqref{qusmplx2}. The morphism 
$\varkappa$ will enter the discussion of truncation and skeletonization
of simplicial sets in sect. \cref{sec:simpcircappls}.


\subsection{\textcolor{blue}{\sffamily The simplicial Hilbert Hodge Laplacians
and their properties}}\label{subsec:smplxlapl}

In this subsection, we shall study in some depth the simplicial Hilbert Hodge Laplacians and
their properties having in mind the problem of the computation of the simplicial homology spaces
of a simplicial set analyzed later in subsect. \cref{subsec:smplxhomol}. We shall
do so from a perspective more general than that strictly required by such problem, considering
three kinds of Laplacians. 
The reason for proceeding like so is twofold. 
First, as a way of achieving a broader and more complete understanding of
the quantum simplicial operator framework developed in the preceding subsections. Second, for its potential relevance in a
reinterpretation of the simplicial Hilbert structure as an instance of $N=4$ supersymmetric quantum mechanics
on the lines of the analogous formulation of the quantum simplicial complex framework of ref. \ccite{Lloyd:2014lgz}
as an instance of $N=2$ supersymmetric quantum mechanics
proposed and studied in refs. \ccite{Crichigno:2020vue,Cade:2021jhc,Crichigno:2022wyj}, though we shall delve
into this matter in the present work.

We begin by introducing simplicial Hilbert homological operators 
which will be key in the study of simplicial Hilbert homology later in subsect. \cref{subsec:smplxhomol}

\begin{defi}\label{def:qdnqsn}
The simplicial Hilbert face boundary operators $Q_{Dn}:\scH_n\rightarrow \scH_{n-1}$, $n\geq 1$, 
and degeneracy coboundary operators 
$Q_{Sn}:\scH_n\rightarrow \scH_{n+1}$, $n\geq 0$ are given by 
\begin{align}
&Q_{Dn}=\mycom{{}_\sss}{{}_{0\leq i\leq n}}(-1)^iD_{ni},
\label{smplxhomol1}
\\  
&Q_{Sn}=\mycom{{}_\sss}{{}_{0\leq i\leq n}}(-1)^iS_{ni}.
\label{smplxhomol2}
\end{align}
\end{defi}

\noindent
The $Q_{Dn}$, $Q_{Sn}$ are the building blocks of the simplicial Hilbert Hodge Laplacians.

\begin{defi}\label{def:hodlap}
The face, mixed and degeneracy simplicial Hilbert Hodge Laplacians are the operators
$H_{DDn}:\scH_n\rightarrow\scH_n$, $n\geq 0$, $H_{SDn}:\scH_n\rightarrow\scH_{n-2}$, $n\geq 2$,
and $H_{SSn}:\scH_n\rightarrow\scH_n$, $n\geq 0$ given by
\begin{align}
&H_{DDn}=Q_{Dn}{}^+Q_{Dn}+Q_{Dn+1}Q_{Dn+1}{}^+,
\label{smplxlapl1}
\\
&H_{SDn}=Q_{Sn-2}{}^+Q_{Dn}+Q_{Dn-1}Q_{Sn-1}{}^+, 
\label{smplxlapl2}
\\
&H_{SSn}=Q_{Sn}{}^+Q_{Sn}+Q_{Sn-1}Q_{Sn-1}{}^+.
\label{smplxlapl3}
\end{align}
\end{defi}

\noindent Above, it is tacitly understood that the first term in the right hand side of
\ceqref{smplxlapl1} and the second term in the right hand side of
\ceqref{smplxlapl3} are absent when $n=0$. 
We observe that the $H_{DDn}$, $H_{SSn}$ are Hermitian while the $H_{SDn}$ are not. 


The $H_{DDn}$, $H_{DSn}$, $H_{SSn}$ can be expressed
through the basic face and degeneracy operators $D_{ni}$, $S_{ni}$ and their adjoints
$D_{ni}{}^+$, $S_{ni}{}^+$ on account of \ceqref{smplxhomol1}, \ceqref{smplxhomol2}.
The resulting expressions of $H_{DDn}$, $H_{DSn}$, $H_{SSn}$ exhibit a similar structure:
\begin{align}
&H_{DDn}=\varUpsilon_{DDn}+\varUpsilon_{DDn}{}^++H^0{}_{DDn},
\label{susysmplx7}
\\
&H_{SDn}=\varUpsilon_{SDn}+\varUpsilon_{DSn-2}{}^++H^0{}_{SDn},
\label{susysmplx8}
\\
&H_{SSn}=
H^0{}_{SSn}.
\label{susysmplx12}
\end{align}
Here, $\varUpsilon_{DDn}$, $\varUpsilon_{DSn}$, $\varUpsilon_{SDn}$, 
are operators directly related to the defects
$\varDelta^{DD}{}_{nij}$, $\varDelta^{DS}{}_{nij}$, $\varDelta^{SD}{}_{nij}$ 
introduced in subsect. \cref{subsec:qusmplx} and given by eqs. \ceqref{qusmplx16}--\ceqref{qusmplx18},
\begin{subequations}
\label{susysmplx13/15}
\begin{align}
&\varUpsilon_{DDn}=\mycom{{}_\sss}{{}_{0\leq i,j\leq n,i<j}}(-1)^{i+j}\varDelta^{DD}{}_{nij},
\label{susysmplx13}
\\
&\varUpsilon_{DSn}=\mycom{{}_\sss}{{}_{0\leq i,j\leq n,i\leq j}}(-1)^{i+j}\varDelta^{DS}{}_{nij},
\label{susysmplx14}
\\
&\varUpsilon_{SDn}=\mycom{{}_\sss}{{}_{0\leq i,j\leq n,i+1<j}}(-1)^{i+j}\varDelta^{SD}{}_{nij}.
\label{susysmplx15}
\end{align}
\end{subequations}
Above, $\varUpsilon_{DD0}=0$ by convention.
We note here that $\varUpsilon_{DSn}$, $\varUpsilon_{SDn}$ vanish when $\sfX$ is a semi perfect simplicial set
and that $\varUpsilon_{DDn}$ also vanishes when $\sfX$ is a quasi perfect simplicial set (cf. def. \cref{def:perfect}).
A term of the form $\varUpsilon_{SSn}+\varUpsilon_{SSn}{}^+$ depending in an analogous manner on the
defects $\varDelta^{SS}{}_{nij}$, $\varDelta^{SS}{}_{nij}{}^+$ does not appear in the expression
of $H_{SSn}$ in \ceqref{susysmplx12}, because the $\varDelta^{SS}{}_{nij}$ always vanish
by the no degeneracy defect theorem 
\cref{prop:vanidect1}. 
The operators $H^0{}_{DDn}$, $H^0{}_{SDn}$,
$H^0{}_{SSn}$ instead are not reducible to the defects and are genuinely new. They provide additional structure
to our quantum simplicial framework. 
In particular, they involve
all the missing products are $D_{n+1i}D_{n+1i}{}^+$, $0\leq i\leq n+1$, $D_{n+2i+1}{}^+S_{ni}$, $0\leq i\leq n$,
$S_{n+1i}D_{n+1i}{}^+$, $0\leq i\leq n+1$ and $S_{ni}{}^+S_{ni}$, $0\leq i\leq n$ not covered by the mixed
exchange identities \ceqref{qusmplx12/15}.

The operator $H^0{}_{DDn}$ is Hermitian. $H^0{}_{DDn}$ can be expressed through two sets of elementary operators. 
The first set consists of the operators 
\begin{equation}
\varOmega_{ni}=D_{n+1i}D_{n+1i}{}^+
\label{qusmplx38}
\end{equation}
with $0\leq i\leq n+1$. The $\varOmega_{ni}$ are clearly Hermitian. A simple application of the
basic expressions \ceqref{qusmplx1}, \ceqref{qusmplx3} shows that the $\varOmega_{ni}$ are diagonal
in the simplex basis $\ket{\sigma_n}$, 
\begin{equation}
\varOmega_{ni}=\mycom{{}_\sss}{{}_{\sigma_n\in \sfX_n}}\ket{\sigma_n}|\sfD_{n+1i}(\sigma_n)|\bra{\sigma_n}.
\label{qusmplx39}
\end{equation}
So, for every $i$ and $n$--simplex $\sigma_n$, \pagebreak $\varOmega_{ni}$ counts the number of $n+1$-simplices $\omega_{n+1}$
whose $i$--face is $\sigma_n$. The second set is constituted by the operators
\begin{align}
&\varTheta_{ni}=D_{ni}{}^+D_{ni},
\label{qusmplx40}
\\
&\varGamma_{ni}=D_{n+1i+1}D_{n+1i}{}^+
\label{qusmplx41}
\end{align}
with $0\leq i\leq n$. Above, we conventionally set $\varTheta_{00}=0$. The $\varTheta_{ni}$ are Hermitian, whilst
the $\varGamma_{ni}$ are not. Another straightforward application of \ceqref{qusmplx1}, \ceqref{qusmplx3} furnishes
the following formulae:
\begin{align}
&\varTheta_{ni}=\mycom{{}_\sss}{{}_{\sigma_n\in\sfX_n}}
\mycom{{}_\sss}{{}_{\omega_n\in\sfD_{ni}(d_{ni}\sigma_n)}}\ket{\omega_n}\bra{\sigma_n},
\label{qusmplx42}
\\  
&\varGamma_{ni}=\mycom{{}_\sss}{{}_{\sigma_n,\omega_n\in\sfX_n}}\ket{\omega_n}
|\sfD_{n+1i}(\sigma_n)\cap\sfD_{n+1i+1}(\omega_n)|\bra{\sigma_n}.
\nonumber
\label{qusmplx43}
\end{align}
Hence, for each $i$, $\varTheta_{ni}$ detects all the pairs $\sigma_n$, $\omega_n$ 
of $n$--simplices sharing the $i$--face whilst $\varGamma_{ni}$ provides information about the number 
of $n+1$-simplices $\omega_{n+1}$ having $\sigma_n$, $\omega_n$ as their $i$, $i+1$--th faces.
We note that for $n\geq 1$, owing to the simplicial relation \ceqref{smplsets1},
the effective summation range of $\varGamma_{ni}$ consists of pairs
$\sigma_n,\omega_n\in\sfX_n$ such that $d_{ni}\omega_n=d_{ni}\sigma_n$ and thus it is contained
in the summation range of $\varTheta_{ni}$.

The operator $H^0{}_{DDn}$ is given by a sum of operator products of the kind appearing in the right hand side
of eqs. \ceqref{qusmplx38}, \ceqref{qusmplx40}, \ceqref{qusmplx41}. $H^0{}_{DDn}$ is in this way 
expressible in terms of the operators $\varOmega_{ni}$, $\varTheta_{ni}$, 
\begin{equation}
H^0{}_{DDn}=\mycom{{}_\sss}{{}_{0\leq i\leq n+1}}\varOmega_{ni}
+\mycom{{}_\sss}{{}_{0\leq i\leq n}}(\varTheta_{ni}-\varGamma_{ni}-\varGamma_{ni}{}^+).
\label{qusmplx44}
\end{equation}
$H^0{}_{DDn}$ so encodes all the information about face relations in the underlying simplicial set $\sfX$ the
$\varOmega_{ni}$, $\varTheta_{ni}$ and $\varGamma_{ni}$ do. 

The operators $H^0{}_{SDn}$ and the Hermitian operators $H^0{}_{SSn}$ have a more elementary structure.
They are reducible to a common set of elementary orthogonal projectors as we now show.

For $0\leq i\leq n$, the operator $S_{ni}$ is an isometry of $\scH_n$ into $\scH_{n+1}$,
\begin{equation}
S_{ni}{}^+S_{ni}=1_n.
\label{qusmplx31}
\end{equation}
This follows immediately from the expressions \ceqref{qusmplx2}, \ceqref{qusmplx4} of
$S_{ni}$, $S_{ni}{}^+$ and the fact
that the sets $\sfS_{ni}(\sigma_{n+1})$ defined in \ceqref{qusmplx6} contain at most one element. 
Consequently, $S_{ni}S_{ni}{}^+$ is the orthogonal projector on the range $\ran S_{ni}$ of $S_{ni}$. 

In general, the Hermitian operators 
\begin{equation}
\varPi_{ni}=S_{n-1i}S_{n-1i}{}^+=S_{ni+1}{}^+S_{ni}=S_{ni}{}^+S_{ni+1},
\label{qusmplx32}
\end{equation}
where $0\leq i\leq n-1$, are orthogonal projectors in $\scH_{n}$.
The identity of the three expressions of $\varPi_{ni}$ follows from the  $S$--$S^+$ exchange identities
\ceqref{qusmplx15}. It is simple to verify using \ceqref{qusmplx2}, \ceqref{qusmplx4} that 
the $\varPi_{ni}$ are diagonal in the simplex basis $\ket{\sigma_n}$, 
\begin{equation}
\varPi_{ni}=\mycom{{}_\sss}{{}_{\sigma_n\in \sfX_n}}\ket{\sigma_n}|\sfS_{n-1i}(\sigma_n)|\bra{\sigma_n}.
\label{qusmplx33}
\end{equation}
Since $|\sfS_{n-1i}(\sigma_n)|\leq 1$, for fixed $i$ $\varPi_{ni}$ detects whether a given $n$--simplex $\sigma_n$
lies in the range of $s_{n-1i}$ or not. 
The $\varPi_{ni}$ evidently commute pairwise.
The $\varPi_{ni}$ do not furnish however a resolution of the identity
of $\scH_{n}$ because $\varPi_{ni}\varPi_{nj}\neq 0$ in general for $i\neq j$. 
Via the projectors $\varPi_{ni}$, the adjoint degeneracy operators $S_{ni}{}^+$ are reducible
to the face operators $D_{ni}$, since
\begin{equation}
S_{ni}{}^+=D_{n+1i}\varPi_{n+1i}=D_{n+1i+1}\varPi_{n+1i},
\label{qusmplx35}
\end{equation}
as follows immediately from \ceqref{qusmplx9}.

The operators $H^0{}_{SDn}$, $H^0{}_{SSn}$ are expressible in a simple manner in terms of the projectors
$\varPi_{ni}$, 
\begin{align}
&H^0{}_{DSn}=\mycom{{}_\sss}{{}_{0\leq i\leq n-1}}D_{n-1i}D_{ni+1}\varPi_{ni}
-\mycom{{}_\sss}{{}_{0\leq i\leq n-2}}D_{n-1i}\varPi_{n-1i}D_{ni+1},
\label{qusmplx36}
\\
&H^0{}_{SSn}=(n+1)1_n-\mycom{{}_\sss}{{}_{0\leq i\leq n-1}}\varPi_{ni}.
\label{qusmplx37}
\end{align}
The second term in the right hand side of \ceqref{qusmplx37} is conventionally set to $0$ for $n=0$.
The verification of these identities is straightforward enough
from \ceqref{qusmplx31}, \ceqref{qusmplx35}.


\subsection{\textcolor{blue}{\sffamily Simplicial Hilbert homology}}\label{subsec:smplxhomol}


In subsect. \cref{subsec:simplhomol}, we showed that a chain complex
can be associated with any simplicial group. 
This scheme can be applied in particular to the simplicial Hilbert space of
a parafinite simplicial set introduced in subsects. \cref{subsec:qusmplx} and \cref{subsec:hilbfunct}, 
adding new elements to our analysis. In fact, 
the richness of the operator structure of the quantum simplicial framework enables one to \pagebreak 
introduce several types of simplicial and cosimplicial Hilbert homology and cohomology. 
Since simplicial Hilbert theory is just a
special codification of standard simplicial set theory, we expect that eventually we shall recover
the ordinary simplicial homology of the underlying simplicial set, if we succeed, and no more.
This is indeed the case: all the homologies and cohomologies that can be constructed turn indeed out
to be either trivial or isomorphic to simplicial homology.
This subsection is devoted to the illustration of such construction.


Let $\sfX$ be a parafinite simplicial set and let $\scH$ be its associated simplicial Hilbert space.
The application of the homological set--up of subsect. \cref{subsec:simplhomol} to $\scH$ seen as
a simplicial Abelian group yields a simplicial chain complex, the simplicial Hilbert face 
chain complex $(\scH,Q_D)$. Its chain spaces are the Hilbert spaces $\scH_n$; its 
boundary operators are the simplicial Hilbert face boundary operators $Q_{Dn}$, $n\geq 1$, defined in eq. 
\ceqref{smplxhomol1}. As is readily verified also using the simplicial identities \ceqref{qusmplx7},
the $Q_{Dn}$ obey indeed the basic homological relations \ceqref{nwhom2} reading currently as  
\begin{equation}
Q_{Dn}Q_{Dn+1}=0,  
\label{smplxhomol3}
\end{equation}
Associated with 
$(\scH,Q_D)$ there are then the
simplicial Hilbert face homology spaces $\rmH_{Dn}(\scH)=\ker Q_{Dn}/\ran Q_{Dn+1}$ with $n\geq 0$,
where $\ker Q_{D0}=\scH_0$ by convention.

Similarly, the application of the homological set--up to
the cosimplicial Hilbert space $\scH^+$ seen as
a cosimplicial Abelian group yields a cosimplicial cochain complex, the cosimplicial Hilbert face 
cochain complex $(\scH^+,Q_D{}^+)$. Its cochain spaces are the Hilbert spaces $\scH_n$; its 
coboundary operators are the adjoints $Q_{Dn}{}^+$ of the simplicial Hilbert face boundary operators $Q_{Dn}$.
The $Q_{Dn}{}^+$ obey indeed cohomological relations following from the \ceqref{smplxhomol3}
by adjunction and identical to these but for the order of the factors of the operator products.
Associated with $(\scH^+,Q_D{}^+)$ there are then the cosimplicial Hilbert face cohomology 
spaces $\rmH_D{}^n(\scH^+)=\ker Q_{Dn+1}{}^+/\ran Q_{Dn}{}^+$ with $n\geq 0$, where $\ran Q_{D0}{}^+=0$
by convention.

Consider a morphism $\phi:\sfX\rightarrow\sfX'$ of the parafinite simplicial sets $\sfX$, $\sfX'$.
As we have seen in subsect. \cref{subsec:hilbfunct}, 
with $\phi$ there is associated a morphism $\varPhi:\scH\rightarrow\scH'$
of the simplicial Hilbert spaces $\scH$, $\scH'$ of $\sfX$, $\sfX'$. In the spirit of 
the homological theory of subsect. \cref{subsec:simplhomol}, this can be regarded as a
morphism of simplicial Abelian groups. A morphism 
of chain complexes is then yielded, the associated morphism of simplicial Hilbert chain complexes 
$\varPhi:(\scH,Q_D)\rightarrow(\scH',Q'{}_D)$.  Its components are the operators $\varPhi_n$ given by eq. \ceqref{hsmor0}, 
By virtue of \ceqref{smplxhomol1}, the $\varPhi_n$ obey indeed the identities \ceqref{nwhom5}, reading here as 
\begin{equation}
\varPhi_{n-1}Q_{Dn}=Q'{}_{Dn}\varPhi_n.
\label{smplxhomol12}
\end{equation}
Similarly regarding the morphism $\varPhi^+:\scH'^+\rightarrow\scH^+$ of the cosimplicial Hilbert spaces
$\scH'^+$, $\scH^+$ of $\sfX'$, $\sfX$ as a morphism of cosimplicial Abelian groups, we find
a morphism of cochain complexes, the associated morphism of cosimplicial Hilbert cochain complexes 
$\varPhi^+:(\scH'^+,Q'{}_D{}^+)\rightarrow(\scH^+,Q_D{}^+)$. Its components $\varPhi_n{}^+$, given by \ceqref{hsmor3},
satisfy indeed the adjoints of relations \ceqref{smplxhomol12}.
$\varPhi$, $\varPhi^+$ give hereby rise to
morphisms $\varPhi_{*n}:\rmH_{Dn}(\scH)\rightarrow\rmH_{Dn}(\scH')$ and
$\varPhi^{+*n}:\rmH_D{}^n(\scH'^+)\rightarrow\rmH_D{}^n(\scH^+)$ of the associated 
Hilbert face homology and cohomology spaces. 

The computation of the homology/cohomology spaces $\rmH_{Dn}(\scH)$, $\rmH_D{}^n(\scH^+)$
can be carried out by mimicking that of the de Rham cohomology spaces of closed Riemannian manifolds
in Hodge theory: it reduces to the determination of
the ker\-nels of appropriate simplicial Hilbert Hodge Laplacians \ccite{Lloyd:2014lgz}.

\begin{theor} \label{prop:hodge} (Simplicial Hilbert face Hodge theorem) For $n\geq 0$, 
the isomorphism
\begin{equation}
\rmH_{Dn}(\scH)\simeq\rmH_D{}^n(\scH^+)\simeq\ker H_{DDn} \quad \text{with $n\geq 0$}
\label{smplxhomol16}
\end{equation}
holds, where $H_{DDn}$ is the face simplicial Hilbert Hodge Laplacian (cf.  eq. \ceqref{smplxlapl1}). 
\end{theor}

\begin{proof}
The proof of theor. \cref{prop:hodge} is  based on the finite dimensional 
Hodge theorem. Although this theorem is well--known, we provide a simple proof of it in
app. \cref{app:hodge} for the reader's benefit.
\end{proof}

On account of \ceqref{simplhomol2}, \ceqref{qusmplx1} and \ceqref{smplxhomol1}, the face chain complex $(\scH,Q_D)$
is manifestly isomorphic to the simplicial chain complex $(\sfC(\sfX,\bbC),\partial)$ with complex coefficients
of the simplicial set $\sfX$. At degree $n$, the isomorphism is given by the chain encoding map
$\varkappa_n:\sfC_n(\sfX,\bbC)\rightarrow\scH_n$  engendered by the simplicial encoding map
introduced at the end of subsect. \cref{subsec:hilbfunct}.
Such isomorphism 
translates into one of the corresponding homology spaces leading to the following.


\begin{theor} \label{prop:hlbtosmplxhom}
The simplicial Hilbert face homology and the cosimplicial face cohomology are 
isomorphic to the simplicial homology 
of the underlying simplicial set $\sfX$ with complex coefficients: for $n\geq 0$
\begin{equation}
\rmH_{Dn}(\scH)\simeq\rmH_D{}^n(\scH^+)\simeq\rmH_n(\sfX,\bbC). 
\label{smplxhomo-3}
\end{equation}
Consequently, one has \hphantom{xxxxxxx}
\begin{equation}
\rmH_n(\sfX,\bbC)\simeq\ker H_{DDn}. 
\label{smplxhomo-3/1}
\end{equation}
\end{theor}


Consider again a morphism $\phi:\sfX\rightarrow\sfX'$ of the parafinite simplicial sets $\sfX$, $\sfX'$.
$\phi$ gives rise to a morphism $\phi:(\sfC(\sfX,\bbC),\partial)\rightarrow(\sfC(\sfX',\bbC),\partial')$
of the complex simplicial chain complexes 
of $\sfX$, $\sfX'$ (cf. subsect. \cref{subsec:simplhomol}, def. \cref{def:chcxmorph})
and by virtue of this a morphism $\phi_{*n}:\rmH_n(\sfX,\bbC)\rightarrow \rmH_n(\sfX',\bbC)$ of the associated complex
simplicial homology spaces for each $n$. Evidently, the simplicial
Hilbert homology and cohomology space morphism $\varPhi_{*n}:\rmH_{Dn}(\scH)\rightarrow\rmH_{Dn}(\scH')$ and
$\varPhi^{+*n}:\rmH_D{}^n(\scH'^+)\rightarrow\rmH_D{}^n(\scH^+)$ we have constructed earlier
are the simplicial Hilbert encoding
of the simplicial homology morphisms $\phi_{*n}$. 

As anticipated at the beginning of this subsection, the quantum simplicial framework
is characterized by further homology and cohomology spaces, which we briefly illustrate next.
Such spaces can be shown to be trivial, as expected also on general grounds. The uninterested reader
can skip this discussion and move directly to the last paragraph of this subsection, if he/she wishes so. 

In subset. \cref{subsec:smplxlapl}, we have introduced also the simplicial Hilbert degeneracy coboundary operators
$Q_{Sn}$. Using the simplicial identities \ceqref{qusmplx11}, it is not difficult to show that
the operators $Q_{Sn}$ obey the basic cohomological relations
\begin{equation}
Q_{Sn+1}Q_{Sn}=0
\label{smplxhomol5}
\end{equation}
Exploiting the simplicial identities \ceqref{qusmplx8}--\ceqref{qusmplx10}, one finds in addition that
the $Q_{Sn}$ satisfy a further relation,
\begin{equation}
Q_{Sn-1}Q_{Dn}+Q_{Dn+1}Q_{Sn}=0,  
\label{smplxhomol4}
\end{equation}
involving the simplicial Hilbert face boundary operators $Q_{Dn}$ considered earlier. \pagebreak 
In this wise, the $Q_{Dn}$ are part of a broader homological structure including
the $Q_{Sn}$. 

By virtue of relations \ceqref{smplxhomol5}, the simplicial Hilbert space $\scH$ of $\sfX$ underlies
the simplicial Hilbert degeneracy cochain complex $(\scH,Q_S)$ with cochain spaces $\scH_n$ and
coboundary operators $Q_{Sn}$. Associated with this there are the simplicial Hilbert 
degeneracy cohomology spaces $\rmH_S{}^n(\scH)=\ker Q_{Sn}/\ran Q_{Sn-1}$ for any $n\geq 0$,
where $\ran Q_{S-1}=0$ conventionally.

Similarly, by the adjoint of relations \ceqref{smplxhomol5},
the cosimplicial Hilbert space $\scH^+$ of $\sfX$ supports the cosimplicial Hilbert degeneracy
chain complex $(\scH^+,Q_S{}^+)$ with chain spaces $\scH_n$ and
boundary operators $Q_{Sn}{}^+$. Associated with this there are the cosimplicial Hilbert 
degeneracy homology spaces  $\rmH_{Sn}(\scH^+)=\ker Q_{Sn-1}{}^+/\ran Q_{Sn}{}^+$
for any $n\geq 0$ with $\ker Q_{S-1}{}^+=\scH_0$.

Let $\phi:\sfX\rightarrow\sfX'$ be a morphism of the parafinite simplicial sets $\sfX$, $\sfX'$.
The components $\varPhi_n$ of the attached simplicial Hilbert space
morphism $\varPhi:\scH\rightarrow\scH'$ obey relations \ceqref{hsmor2}. Consequently, 
by \ceqref{smplxhomol2}, the $\varPhi_n$ obey also the identities 
\begin{equation}
\varPhi_{n+1}Q_{Sn}=Q'{}_{Sn}\varPhi_n. 
\label{smplxhomol13}
\end{equation}
Owing to \ceqref{nwhom5}, the $\varPhi_n$ define then a morphism
$\varPhi:(\scH,Q_S)\rightarrow(\scH',Q'{}_S)$ of cochain complexes.
In the same way, by virtue of the Hilbert dagger structure,  the components $\varPhi_n{}^+$ of the adjoint morphism
$\varPhi^+:\scH'^+\rightarrow\scH^+$ give rise to a morphism $\varPhi^+:(\scH'^+,Q'{}_S{}^+)\rightarrow(\scH^+,Q_S{}^+)$ 
of chain complexes. One has in this wise morphisms $\varPhi_*{}^n:\rmH_S{}^n(\scH)\rightarrow\rmH_S{}^n(\scH')$ and
$\varPhi^{+*}{}_n:\rmH_{Sn}(\scH'^+)\rightarrow\rmH_{Sn}(\scH^+)$ of the associated simplicial and cosimplicial 
Hilbert degeneracy cohomology and homology spaces. 

In spite of the formal similarities of the homology/cohomology spaces $\rmH_{Dn}(\scH)$, $\rmH_D{}^n(\scH^+)$
and the cohomology/homology spaces $\rmH_S{}^n(\scH)$, $\rmH_{Sn}(\scH^+)$, while the former are generally
non trivial, the latter always are, in accordance to our expectations, as we show next.

The computation of the spaces $\rmH_S{}^n(\scH)$, $\rmH_{Sn}(\scH^+)$
is again reduced to the determination of the kernels of appropriate simplicial Hilbert Hodge Laplacians.
The triviality of the $\rmH_S{}^n(\scH)$, $\rmH_{Sn}(\scH^+)$ is a immediate consequence of that of such kernels.
\vfill

\begin{theor} \label{prop:hstriv}  (Hilbert degeneracy cohomology and homology triviality theorem)
The simplicial Hilbert degeneracy cohomology and homology spaces are trivial, 
\begin{equation}
\rmH_S{}^n(\scH)\simeq\rmH_{Sn}(\scH^+)\simeq 0
\label{smplxhomo-4}
\end{equation}
for $n\geq 0$. 
\end{theor}

\begin{proof}
By the finite dimensional Hodge theorem, proven in app. \cref{app:hodge}, we have 
\begin{equation}
\rmH_S{}^n(\scH)\simeq\rmH_{Sn}(\scH^+)\simeq\ker H_{SSn}
\label{smplxhomo-4/p}
\end{equation}
for $n\geq 0$, where $H_{SSn}$ is the degeneracy simplicial Hilbert Hodge Laplacian (cf.  eq. \ceqref{smplxlapl3}). 
The result follows by showing that $\ker H_{SSn}=0$. See app. \cref{app:hstriv}.
\end{proof}

A simplicial degeneracy cochain complex $(\sfC(\sfX,\bbC),\laitrap)$
can be built also in standard simplicial theory with no reference to its eventual Hilbert space encoding
alongside with the face chain complex $(\sfC(\sfX,\bbC),\partial)$.
While the complex $(\sfC(\sfX,\bbC),\partial)$ is contemplated and analyzed in simplicial theory, to our knowledge 
the complex $(\sfC(\sfX,\bbC),\laitrap)$ has not appeared and found any application so far.  
The reason for this is presumably that the 
homology of $(\sfC(\sfX,\bbC),\partial)$, which is just the complex simplicial homology $\rmH(\sfX,\bbC)$ of $\sfX$,
is significant and generally non trivial, while the cohomology of $(\sfC(\sfX,\bbC),\laitrap)$, as we have shown,
vanishes. Interestingly, we have been able to provide a completely quantum Hilbert space proof
of this fact. 

From now on, for the reasons explained above, we concentrate on the simplicial Hilbert face homology,
which we shall call simply simplicial Hilbert homology.


\subsection{\textcolor{blue}{\sffamily Normalized simplicial Hilbert homology}}\label{subsec:normal}

The determination of the simplicial homology $\rmH(\sfX,\bbC)$ of a parafinite simplicial set $\sfX$
via that of the isomorphic simplicial Hilbert homology $\rmH_D(\scH)$ is computationally
more costly than necessary, 
as it involves also the subspaces of the simplex Hilbert spaces spanned by the degenerate simplices
(cf. subsect. \cref{subsec:smplsets}, def. \cref{defi:degsiplx}), which are homologically
irrelevant by the normalization theorem \cref{pop:eilmclane}. In this subsection, we shall explain how this redundant 
degenerate structure can be disposed of in our formulation opening a cheaper route to the homology computation. 

\begin{defi} \label{def:degsub} For $n\geq 0$, the degenerate $n$--simplex space is 
\begin{equation}
{}^s\!\scH_n=\sss_{i=0}^{n-1}\ran S_{n-1i},
\label{normal0}
\end{equation}
where ${}^s\!\scH_0=0$ by convention. 
\end{defi}

\noindent
The expression in the right hand side of \ceqref{normal0} denotes the linear span of the
ranges of the operators $S_{n-1i}$. ${}^s\!\scH_n$ is therefore the subspace of $\scH_n$ spanned
by the degenerate $n$--simplex vectors as alluded by its name. 

By the isomorphism of the chain complexes $(\sfC(\sfX,\bbC),\partial)$ and $(\scH,Q_D)$
disclosed in subsect. \cref{subsec:smplxhomol}, the fact that the simplicial boundary operators
of $\sfX$ preserve the subspaces of complex 
degenerate chains of $\sfX$ recalled in subsect. \cref{subsec:simplhomol} translates into the property
that the face boundary operators of $\scH$ preserve the degenerate simplex subspaces of $\scH$.
So, for each $n\geq 1$ $Q_{Dn}{}^s\!\scH_n\subseteq{}^s\!\scH_{n-1}$. Let $\ol{\scH}_n=\scH_n/{}^s\!\scH_n$. 
An operator $\ol{Q}_{Dn}:\ol{\scH}_n\rightarrow\ol{\scH}_{n-1}$ is then induced by $Q_{Dn}$, 
which obeys the homological relation $\ol{Q}_{Dn-1}\ol{Q}_{Dn}=0$ (cf. eq.\ceqref{smplxhomol3}).
We have in this way a chain complex
$(\ol{\scH},\ol{Q}_D)$ 
called the abstract normalized simplicial Hilbert face complex of $\sfX$ below.
The associated abstract normalized simplicial Hilbert
homology spaces are 
$\rmH_{Dn}(\ol{\scH})$ $=\ker \ol{Q}_{Dn}/\ran \ol{Q}_{Dn+1}$ with $n\geq 0$
(with $\ker\ol{Q}_{D0}=\ol{\scH}_0$).

\begin{theor} For every $n\geq 0$, one has 
\begin{equation}
\rmH_{Dn}(\ol{\scH})\simeq\rmH_n(\sfX,\bbC). 
\label{normal-3}
\end{equation}
\end{theor}

\begin{proof}
The chain complexes $(\sfC(\sfX,\bbC),\partial)$ and $(\scH,Q_D)$ are isomorphic.
Further, by\ceqref{normal0}, for each $n$ we have ${}^s\sfC_n(\sfX,\bbC)\simeq{}^s\!\scH_n$.
The normalized chain complexes $(\ol{\sfC}(\sfX,\bbC),\ol{\partial})$
and $(\ol{\scH},\ol{Q}_D)$ are so also isomorphic and so $\ol{\rmH}_n(\sfX,\bbC)\simeq\rmH_{Dn}(\ol{\scH})$.
By the normalization theorem \cref{pop:eilmclane}, 
$\rmH_n(\sfX,\bbC)\simeq\ol{\rmH}_n(\sfX,\bbC)$. \ceqref{normal-3} then follows. 
\end{proof} 




The isomorphism \ceqref{normal-3} offers an alternative way of computing the simplicial homology
of $\sfX$, which is less expensive in that it does away with degenerate simplices. The modding out of
these latter is however hardly implementable algorithmically by its abstract form. 
A full operator formulation is necessary for that purpose.

With the above in mind, we introduce the orthogonal projector $\varPi_n$ on the degenerate $n$--simplex space 
${}^s\!\scH_n$. An expression of $\varPi_n$ can be obtained in terms of the orthogonal projectors $\varPi_{ni}$
introduced in \ceqref{qusmplx32}: one has 
\begin{equation}
\varPi_n=1_n-\mycom{{}_\ppp}{{}_{0\leq i\leq n-1}}(1_n-\varPi_{ni})
\label{normal2}
\end{equation}
for $n\geq 0$ with $\varPi_0=0$ by convention. The reader is referred to app.
\cref{app:deginv} for some details about the derivation of this formula. 

We now introduce another Hilbert complex closely related 
to the normalized Hilbert complex $(\ol{\scH},\ol{Q}_D)$. 
For $n\geq 1$, let ${}^c\!\scH_n={}^s\!\scH_n{}^\perp$, where ${}^\perp$ denotes orthogonal complement,
and let ${}^cQ_{Dn}:{}^c\!\scH_n\rightarrow{}^c\!\scH_{n-1}$ be
defined by 
\begin{equation}
{}^cQ_{Dn}=(1_{n-1}-\varPi_{n-1})Q_{Dn}\big|_{{}^c\!\scH_n}.
\label{normal4}
\end{equation}
The ${}^cQ_{Dn}$ satisfy the homological relation ${}^cQ_{Dn-1}{}^cQ_{Dn}=0$.
The reader is referred again to app.
\cref{app:deginv} for some details about the derivation of this identity.
We have consequently a chain complex $({}^c\!\scH,{}^cQ_D)$
which we shall denominate the concrete normalized simplicial Hilbert face complex
of $\sfX$ in the following. The associated concrete normalized simplicial Hilbert
homology spaces with $n\geq 0$ are defined as $\rmH_{Dn}({}^c\!\scH)=\ker {}^cQ_{Dn}/\ran {}^cQ_{Dn+1}$
(with $\ker{}^cQ_{D0}={}^c\!\scH_0$). 


The following result shows the isomorphism of the two homologies we have introduced above. 

\begin{prop} \label{prop:hilbnorm/1} 
The abstract and concrete simplicial Hilbert face homologies are isomorphic: for every $n\geq 0$, it holds that 
\begin{equation}
\rmH_{Dn}(\ol{\scH})\simeq\rmH_{Dn}({}^c\!\scH).
\label{normal7}
\end{equation}  
\end{prop}

\begin{proof}
The proof of the isomorphism \ceqref{normal7} can be achieved by constructing a chain equivalence
of the abstract and concrete Hilbert complexes $(\ol{\scH},Q_D)$, $({}^c\!\scH{}, {}^cQ_D)$.
The chain equivalence consists of a sequence of chain operators $I_n:\ol{\scH}_n\rightarrow{}^c\!\scH_n$,
$J_n:{}^c\!\scH_n\rightarrow\ol{\scH}_n$, $n\geq 0$, such that the composite operators
$J_nI_n$, $I_nJ_n$ are chain homotopic to $\ol{1}_n$, ${}^c1_n$, respectively.
The property of $I_n$, $J_n$ being chain operators is just $I_n$, $J_n$ satisfying the relations 
\begin{subequations}
\label{normal8/9}
\begin{align}
&I_{n-1}\ol{Q}_{Dn}={}^cQ_{Dn}I_n,
\label{normal8}
\\
&J_{n-1}{}^cQ_{Dn}=\ol{Q}_{Dn}J_n
\label{normal9}
\end{align}
\end{subequations}
for $n\geq 1$. The chain homotopy of $J_nI_n$, $I_nJ_n$ and $1_n$, ${}^c1_n$ descends from the existence of operators
$\ol{W}_n:\ol{\scH}_n\rightarrow\ol{\scH}_{n+1}$,
${}^cW_n:{}^c\!\scH_n\rightarrow{}^c\!\scH_{n+1}$
such that
\begin{subequations}
\label{normal10/11}
\begin{align}
&J_nI_n-\ol{1}_n=\ol{Q}_{Dn+1}\ol{W}_n+\ol{W}_{n-1}\ol{Q}_{Dn},
\label{normal10}
\\
&I_nJ_n-{}^c1_n={}^cQ_{Dn+1}{}^cW_n+{}^cW_{n-1}{}^cQ_{Dn}
\label{normal11}
\end{align}
\end{subequations}
for all $n\geq 0$, where the second term in the right hand side of both relations is absent when $n=0$.
The diagram
\begin{equation}
\xymatrix@C=2.9pc @R=2.3pc
{\cdots\ar[r]^{\ol{Q}_{D3}~~~~~}
&\ol{\scH}_2
\ar[r]^{\ol{Q}_{D2}}\ar@/_/[d]_{I_2}\ar@/^.75pc/[l]^{\ol{W}_2}
&\ol{\scH}_1\ar[r]^{\ol{Q}_{D1}}\ar@/_/[d]_{I_1}\ar@/^.75pc/[l]^{\ol{W}_1}
&\ol{\scH}_0\ar@/_/[d]_{I_0}\ar@/^.75pc/[l]^{\ol{W}_0}\\
\cdots
\ar[r]_{{}^cQ_{D3}~~}
&{}^c\!\scH_2\ar[r]_{{}^cQ_{D2}~}\ar@/_/[u]_{J_2}\ar@/_.75pc/[l]_{{}^cW_3}
&{}^c\!\scH_1\ar[r]_{{}^cQ_{D1}~}\ar@/_/[u]_{J_1}\ar@/_.75pc/[l]_{{}^cW_1}
&{}^c\!\scH_0\ar@/_/[u]_{J_0}\ar@/_.75pc/[l]_{{}^cW_0}}.
\label{normal12}
\end{equation}
represents graphically the operator structure described above.

The chain equivalence $I_n$, $J_n$ has the following explicit form. 
$I_n$ is the operator from $\ol{\scH}_n$ to ${}^c\!\scH_n$
induced by the orthogonal projector $1_n-\varPi_n$ by virtue of the fact that  ${}^s\!\scH_n=\ker(1_n-\varPi_n)$. 
$J_n$ is the canonical projection of ${}^c\!\scH_n$ onto $\ol{\scH}_n$. In app. \cref{app:deginv}
it is shown that $I_n$, $J_n$ are both chain operators and are chain homotopic to $\ol{1}_n$, ${}^c1_n$, as required.
\end{proof}

\noindent
The isomorphism \ceqref{normal7} presumably reflects an equivalence of categories
of finite dimensional simplicial Hilbert spaces, $\ul{\rm fdsHilb}$, and the category
chain complexes of finite dimensional Hilbert spaces, $\ul{\rm ChfdHilb}$, 
as a version of the Dold--Kan correspondence \ccite{Dold:1957hsp,Dold:1957hnf,Kan:1958fic}
\footnote{$\vphantom{dot{\dot{f}}}$
  The author thanks U. Schreiber for suggesting this to him.}.

Because of the isomorphism \ceqref{normal7}, we shall no longer
distinguish the abstract and concrete simplicial Hilbert homologies.

The following theorem is the main result of this subsection.

\begin{theor} \label{prop:hilbnorm} (Normalized simplicial Hilbert homology theorem) For all $n\geq 0$,
\begin{equation}
\rmH_{Dn}({}^c\!\scH)\simeq\rmH_n(\sfX,\bbC). 
\label{normal13}
\end{equation}  
\end{theor}

\begin{proof}
The homology isomorphism \ceqref{normal13} follows readily from the isomorphisms
\ceqref{normal-3} and \ceqref{normal7}.
\end{proof}

The isomorphism \ceqref{normal13} provides an alternative pathway to the determination of the complex
simplicial homology of the simplicial set $\sfX$ grounded on normalized simplicial Hilbert homology.
As for the non normalized homology studied in subsect. \cref{subsec:smplxhomol}, the normalized homology
can be computed via finite dimensional Hodge theory. 

\begin{defi} \label{def:normsmplxlapl}
The normalized simplicial Hilbert Laplacians are
the operators ${}^cH_{DDn}:{}^c\!\scH_n\rightarrow{}^c\!\scH_n$, $n\geq 0$, given by 
\begin{equation}
{}^cH_{DDn}={}^cQ_{Dn}{}^+{}^cQ_{Dn}+{}^cQ_{Dn+1}{}^cQ_{Dn+1}{}^+.
\label{normal14}
\end{equation}
\end{defi}

\noindent
Above, it is tacitly understood that the first term in the right hand side of
\ceqref{normal14} is absent when $n=0$.

The following theorem, like theor. \cref{prop:hodge}, relates the normalized simplicial Hilbert
homology spaces to the kernels of the normalized simplicial Hilbert Laplacians.

\begin{theor} \label{prop:normhodge} (Normalized simplicial Hilbert Hodge theorem)
The isomorphism
\begin{equation}
\rmH_{Dn}({}^c\!\scH)\simeq\ker{}^cH_{DDn} 
\label{normal15}
\end{equation}
holds for every $n\geq 0$. Consequently
\begin{equation}
\rmH_n(\sfX,\bbC)\simeq\ker{}^cH_{DDn} 
\label{normal15/1}
\end{equation}
\end{theor}

\begin{proof}
The proof of theor. \cref{prop:normhodge} is based again on the finite dimensional 
Hodge theorem reviewed in app. \cref{app:hodge}. 
\end{proof}

The isomorphism \ceqref{normal15} provides a potentially more efficient way 
of computing the simplicial homology $\rmH(\sfX,\bbC)$ of $\sfX$ with complex coefficients
than the isomorphism \ceqref{smplxhomo-3/1}, as by virtue of it non degenerate simplices have been
effectively excised.

We conclude this subsection presenting explicit expressions of some of the 
operators introduced above for their relevance and later usefulness. In what follows, 
${}^c\sfX_n=\sfX_n\setminus{}^s\sfX_n$ denotes the set of non degenerate
$n$--simplices of $\sfX$. 
The normalized simplicial Hilbert face boundary operator ${}^cQ_{Dn}$ reads as 
\begin{equation}
{}^cQ_{Dn}
=\mycom{{}_\sss}{{}_{0\leq i\leq n}}(-1)^i
\mycom{{}_\sss}{{}_{\sigma_n\in{}^c\sfX_n,\hfpt d_{ni}\sigma_n\in{}^c\sfX_{n-1}}}
\ket{d_{ni}\sigma_n}\bra{\sigma_n}\big|_{{}^c\!\scH_n}.
\label{normal17}
\end{equation}
The adjoint ${}^cQ_{Dn}{}^+$ of ${}^cQ_{Dn}$ is similarly given by 
\begin{equation}
{}^cQ_{Dn}{}^+
=\mycom{{}_\sss}{{}_{0\leq i\leq n}}(-1)^i
\mycom{{}_\sss}{{}_{\sigma_{n-1}\in{}^c\sfX_{n-1}}}\mycom{{}_\sss}{{}_{\omega_n\in\sfD_{ni}(\sigma_{n-1})\cap{}^c\sfX_n}}
\ket{\omega_n}\bra{\sigma_{n-1}}\big|_{{}^c\!\scH_{n-1}}.
\label{normal18}
\end{equation}
Finally, the normalized simplicial Hilbert Laplacian ${}^cH_{DDn}$ takes the form
\begin{align}
{}^cH_{DDn}
&=\mycom{{}_\sss}{{}_{0\leq i,j\leq n}}(-1)^{i+j}
\mycom{{}_\sss}{{}_{\sigma_n\in{}^c\sfX_n,\,d_{nj}\sigma_n\in{}^c\sfX_{n-1}}}
\mycom{{}_\sss}{{}_{\omega_n\in\sfD_{ni}(d_{nj}\sigma_n)\cap{}^c\sfX_n}}  
\ket{\omega_n}\bra{\sigma_n}\big|_{{}^c\!\scH_n}
\label{normal19}
\\
&+\mycom{{}_\sss}{{}_{0\leq i,j\leq n+1}}(-1)^{i+j}
\mycom{{}_\sss}{{}_{\sigma_n\in{}^c\sfX_n}}
\mycom{{}_\sss}{{}_{\omega_{n+1}\in\sfD_{n+1j}(\sigma_n)\cap{}^c\sfX_n,\,d_{n+1i}\omega_{n+1}\in{}^c\sfX_n}}
\ket{d_{n+1i}\omega_{n+1}}\bra{\sigma_n}\big|_{{}^c\!\scH_n}.
\nonumber
\end{align}
These expressions follow by straightforward calculations
from relations \ceqref{qusmplx1}, \ceqref{qusmplx3} and \ceqref{smplxhomol1}
and the identity $\varPi_n=\sum_{\sigma_n\in{}^c\sfX_n}\ket{\sigma_n}\bra{\sigma_n}$.



\subsection{\textcolor{blue}{\sffamily Simplicial quantum circuits}}\label{subsec:simpcirc}

In this section, we introduce the notion of simplicial quantum circuit, a special kind of quantum
circuit naturally emerging in the quantum simplicial set--up and capable in principle of performing
meaningful simplicial computations, and study its properties.

Our treatment will be admittedly idealized.  We shall assume, in fact,
not very realistically, that no ancilla registers and no intermediate
measurements are involved.

We consider again a parafinite simplicial set $\sfX$ and the associated simplicial Hilbert
space $\scH$.

The simplicial quantum register of $\sfX$ is a pre-Hilbert space $\scH^{(\infty)}$ that stores all the simplicial
data of $\sfX$ in the same way as a quantum register is a Hilbert space $\bbC^{2\otimes n}$ that stores all the
configurations of a classical $n$ bit string. Mathematically, $\scH^{(\infty)}$ is the infinite dimensional
pre--Hilbert space
\begin{equation}
\scH^{(\infty)}=\mycom{{}_\ddd}{{}_{0\leq n < \infty}}\scH_n,
\label{newcirc1}
\end{equation}
where the direct summation is purely algebraic. 
  
A simplicial quantum circuit is a quantum circuit based on the register $\scH^{(\infty)}$,
whose functioning is compatible with the structure of the underlying simplicial set $\sfX$.
Mathematically, so, a simplicial quantum circuit is a unitary operator
$U\in\msU(\scH^{(\infty)})$ that satisfies certain simplicial conditions. 

\begin{defi} \label{defi:ssqc}
A simple simplicial quantum circuit is a collection of unitary operators $U_n\in\msU(\scH_n)$ with $n\in\bbN$ such that 
\begin{align}
&U_{n-1}D_{ni}=D_{ni}U_n&& \text{for ~$0\leq i\leq n$}, 
\label{simpcirc1}
\\  
&U_{n+1}S_{ni}=S_{ni}U_n&& \text{for ~$0\leq i\leq n$}.
\label{simpcirc2}  
\end{align}
\end{defi}

\noindent
In the language of simplicial Hilbert theory, a simple simplicial quantum circuit $\{U_n\}$ is therefore
a simplicial unitary operator of the simplicial Hilbert space $\scH$ (cf. eqs. \ceqref{hsmor1/2}).
Intuitively, for each $n$ the operator $U_n$ embodies a quantum circuit implementing a reversible
computation involving the simplices of $\sfX_n$. For the above notion to be really meaningful,
the computations performed for the various values of $n$ should have the same simplicial nature 
and be compatible with the simplicial face and degeneracy relations occurring between the underlying
simplicial data. These properties are precisely codified by relations \ceqref{simpcirc1}, \ceqref{simpcirc2}.

A simple quantum circuit $\{U_n\}$ can be thought of as a whole collection of simplicial quantum gates
of the form 
\begin{equation}
U^{(n)}=U_n\oplus\mycom{{}_\ddd}{{}_{0\leq n' < \infty,n'\neq n}}1_{n'}. 
\label{newcirc2}
\end{equation}
The unitary operator $U\in\msU(\scH^{(\infty)})$ corresponding to the circuit is
\begin{equation}
U=\mycom{{}_\ppp}{{}_{0\leq n<\infty}}U^{(n)}=\mycom{{}_\ddd}{{}_{0\leq n<\infty}}U_n.
\label{newcirc3}
\end{equation}
  
We notice that simple simplicial quantum circuits form a group $\UU(\scH)$ under simplicial
degreewise multiplication and inversion.

Simple simplicial quantum circuits can perform only computations at fixed simplicial degree, an
important limitation. We need more general circuits for more general computations.
The simplicial conditions which a general simplicial quantum circuit obeys should be an appropriate
generalization of those obeyed by simple circuits. To formulate it, we need to introduce an appropriate notation.

For a finite subset $A\subset\bbN$ with $A\neq\emptyset$, the simplicial $A$--subregister is the 
finite dimensional Hilbert space \hphantom{xxxxxxxx}
\begin{equation}
\scH_A=\mycom{{}_\ddd}{{}_{n\in A}}\scH_n\subset\scH^{(\infty)}. \vphantom{\bigg]}
\label{newcirc4}
\end{equation}

For a finite subset $A\subset\bbN$ with $A\neq\emptyset$, we let
$F_A$ be the set of all mappings $\alpha:A\rightarrow\varSigma$ with the property that $\alpha_0=+1$ when $0\in A$,
where $\varSigma=\{-1,+1\}$ is the sign alphabet. We also set $\bbN_n=\{n'|n'\in\bbN, 0\leq n'\leq n\}$, where $n\in\bbN$.
For $\alpha\in F_A$ and $i\in\prod_{n\in A}\bbN_n$, we define the operator
$X^{(\alpha)}{}_{Ai}:\scH_A\rightarrow\scH_{A+\alpha}$ by 
\begin{equation}
X^{(\alpha)}{}_{Ai}=\mycom{{}_\ddd}{{}_{n\in A}}X^{(\alpha_n)}{}_{ni_n}, \vphantom{\bigg]}
\label{newcirc5}
\end{equation}
where $X^{(-1)}{}_{ni}=D_{ni}$, $X^{(+1)}{}_{ni}=S_{ni}$ and $A+\alpha=\{n+\alpha_n|n\in A\}\subset\bbN$. 

\begin{defi} \label{defi:pasqc}
Let $p\in\bbN$, $p>0$. A $p$--ary simplicial quantum circuit consists of a collection of unitary operators
$U_A\in\msU(\scH_A)$ with $A\subset\bbN$ and $|A|=p$ such that for all $\alpha\in F_A$ and $i\in\prod_{n\in A}\bbN_n$ 
\begin{equation}
X^{(\alpha)}{}_{Ai}U_A=U_{A+\alpha}X^{(\alpha)}{}_{Ai}. 
\label{newcirc6}
\end{equation}
\end{defi}

\noindent
The simple simplicial quantum circuits introduced in def. \cref{defi:ssqc}
are just $1$--ary simplicial quantum circuits. 

A $p$-ary quantum circuit $\{U_A\}$ encodes a family of simplicial quantum gates, 
\begin{equation}
U^{(A)}=U_A\oplus\mycom{{}_\ddd}{{}_{n\not\in A}}1_n. 
\label{newcirc7}
\end{equation}
Unlike in the simple case, these gates generally do not commute since the subspaces $\scH_A$ may have non trivial
intersections. The unitary operator $U\in\msU(\scH^{(\infty)})$ of the circuit 
is gotten by multiplying some subset of simplicial gates in a prescribed order.

If $\{U_n\}$ is a simple simplicial quantum circuit, the operators $U_A=\bigoplus_{n\in A}U_n$,
$A\subset\bbN$ and $|A|=p$, constitute a $p$--ary simplicial quantum circuit.
More generally, given a collection $\{U_{\alpha A_\alpha}\}$ of $p_\alpha$--ary simplicial quantum circuits,
$\alpha=1,\ldots,a$, one can construct a $p$--ary simplicial quantum circuit $\{U_A\}$ with $p=\sum_\alpha p_\alpha$
as follows. Every subset $A\subset\bbN$ with $|A|=p$ has a unique partition $A=\bigcup_\alpha A_\alpha$
such that  $A_\alpha\subset\bbN$ with $|A_\alpha|=p_\alpha$ and that for every $\alpha<\beta$, $m\in A_\alpha$, $n\in A_\beta$
one has $m<n$. Then, $U_A=\bigoplus_\alpha U_{\alpha A_\alpha}$. 

\begin{defi}
A $p$--ary simplicial quantum circuit $\{U_A\}$ of kind constructed above is called reducible.
\end{defi}

\noindent
They are so because \pagebreak they have a fixed non trivial
block diagonal structure on each simplicial $A$--subregister.

We present now a template for generating interesting examples of simple simplicial quantum circuits. 
The data of the construction are the following:

\begin{enumerate} 

\vspace{-1mm}\item a pair of parafinite simplicial sets $\sfX$, $\sfX'$;

\vspace{-2mm}\item a simplicial morphism $\phi:\sfX\rightarrow\sfX'$;

\vspace{-2mm}\item a structure of simplicial group on $\sfX'$.

\vspace{-1mm}
  
\end{enumerate}

\noindent
The role played by the simplicial group structure of $\sfX'$ (cf. ex. \cref{exa:simpgr}) is essential.

We now turn to the Cartesian product $\sfX\times\sfX'$ of $\sfX$, $\sfX'$ (cf. def. \cref{def:simplprod}). 
By means of the components $\phi_n$ of $\phi$, we define maps
$\hat\phi_n:\sfX\times\sfX'{}_n\rightarrow\sfX\times\sfX'{}_n$ by setting
\begin{equation}
\hat\phi_n(\sigma_n,\sigma'{}_n)=(\sigma_n,\sigma'{}_n\phi_n(\sigma_n)).
\label{simpcirc8}  
\end{equation}
The second component of the pair in the right hand side exhibits the product of the simplices
$\sigma'{}_n$, $\phi_n(\sigma_n)\in\sfX'{}_n$ in the group $\sfX'{}_n$.
Unlike the $\phi_n$, the $\hat\phi_n$ are always invertible, as they are injective and the sets
$\sfX\times\sfX'{}_n$ are finite. The $\hat\phi_n$ constitute indeed a reversible form of the $\phi_n$
with a structure analogous to that of similar maps employed in reversible computation.

Since $\sfX'$ is a simplicial group, its face and degeneracy maps $d'{}_{ni}$, $s'{}_{ni}$
are group morphisms. Furthermore, as $\phi$ is a simplicial morphism, its components $\phi_n$
satisfy the simplicial relations \ceqref{smplsets6/7}. Exploiting these properties, it is
straightforward to check
that the maps $\hat\phi_n$ obey the \ceqref{smplsets6/7} as well and so
are the components of a simplicial morphism $\hat\phi:\sfX\times\sfX'\rightarrow\sfX\times\sfX'$.
Being the $\hat\phi_n$ invertible, $\hat\phi$ is in fact an isomorphism. 

We define next operators $\hat U_{\phi n}:\scH\otimes\scH'{}_n\rightarrow\scH\otimes\scH'{}_n$ by 
\begin{equation}
\hat U_{\phi n}=\mycom{{}_\sss}{{}_{(\sigma_n,\sigma'{}_n)\in\sfX_n\times\sfX'{}_n}}
\ket{\hat\phi_n(\sigma_n,\sigma'{}_n)}\bra{(\sigma_n,\sigma'{}_n)},
\label{simpcirc9}  
\end{equation}
where the kets $\ket{(\sigma_n,\sigma'{}_n)}=\ket{\sigma_n}\otimes\ket{\sigma'{}_n}$ are
the $n$--simplex basis of $\scH\otimes\scH'{}_n$.

\begin{prop}
The operators $\hat U_{\phi n}$, $n\in\bbN$, constitute a simple simplicial quantum circuit $\hat U_\phi$ 
of $\sfX\times\sfX'$.
\end{prop}

\begin{proof}
The $\hat U_{\phi n}$ are \pagebreak unitary operators since the maps $\hat\phi_n$ are invertible.
Using that $\hat\phi$ is a simplicial morphism, it is now straightforward to verify that 
the $\hat U_{\phi n}$ obey relations \ceqref{simpcirc1}, \ceqref{simpcirc2} and are therefore
the component of a simplicial quantum circuit $\hat U_\phi$ of $\sfX\times\sfX'$.
\end{proof}

\noindent Indeed, the $\hat U_{\phi n}$ are just the components of the simplicial Hilbert automorphism
$\hat U_\phi:\scH\otimes\scH'\rightarrow\scH\otimes\scH'$ associated with the simplicial
isomorphism $\hat\phi$. 

The eventual relevance of the construction just outlined for useful applications
is still to be clarified because of its very special nature.
The example derived by it which we illustrated next is anyway worthy of mentioning.


\begin{exa} \label{exa:sqcsmtsg}
Simplicial quantum circuits of a simplicial group.

\noindent
{\rm
Let $\sfX$ be a simplicial group. Since the simplex sets $\sfX_n$ are groups, 
multiplication and inversion maps $\mu_n:\sfX_n\times\sfX_n\rightarrow\sfX_n$
and $\iota_n:\sfX_n\rightarrow\sfX_n$ are defined at each degree $n$. The property of the
face and degeneracy maps $d_{ni}$, $s_{ni}$ as group morphisms entails that the $\mu_n$
and $\iota_n$ are the components of simplicial morphisms $\mu:\sfX\times\sfX\rightarrow\sfX$,
and $\iota:\sfX\rightarrow\sfX$. Application of the construction scheme illustrated above
shows that with these there are associated simplicial quantum circuits $U_\mu$ and
$U_\iota$ of $\sfX\times\sfX\times\sfX$ and $\sfX\times\sfX$, respectively.
}
\end{exa}

The challenge facing one presently is constructing irreducible simplicial quantum circuits beyond
the simple ones and more broadly to devise a classification scheme of such circuits
reflecting basic properties of simplicial set theory



\subsection{\textcolor{blue}{\sffamily Simplicial quantum circuits and homology}}\label{subsec:simpcirchom}

In this final subsection, we shall examine the interplay of simplicial quantum circuits
studied in sect. \cref{subsec:simpcirc} and simplicial Hilbert homology as analyzed in subsects.
\cref{subsec:smplxhomol}, \cref{subsec:normal}. Our discussion will be limited
to simple simplicial quantum circuits (cf. def. \cref{defi:ssqc}), its extension to
$p$--ary simplicial quantum circuits (cf. def. \cref{defi:pasqc}) being straightforward.

Owing to  \ceqref{smplxhomol1} and \ceqref{simpcirc1}, the components $U_n$ of a
simple simplicial quantum circuit intertwine
also the simplicial Hilbert face coboundary operators $Q_{Dn}$,
\begin{equation}
U_{n-1}Q_{Dn}=Q_{Dn}U_n, 
\label{simpcirc3}  
\end{equation}
and hence define by eq. \ceqref{smplxhomol12}
a unitary chain endomorphism of the simplicial Hilbert face chain complex $(\scH,Q_D)$
studied in subsect. \cref{subsec:smplxhomol}. 
For each degree $n$, so, 
there exists an automorphism of the homology space $\rmH_{Dn}(\scH)$ associated with $U_n$
and hence, by the isomorphism \ceqref{smplxhomo-3}, also one of the simplicial homology $\rmH_n(\sfX,\bbC)$.
Such automorphism can be concretely described as follows. 
Relations \ceqref{smplxlapl1} and \ceqref{simpcirc3} entail that   
simplicial Hilbert Laplacian $H_{DDn}$ commutes with the $U_n$,
\begin{equation}
H_{DDn}U_n=U_nH_{DDn}. 
\label{simpcirc4}  
\end{equation}
By virtue of this relation it holds that that $U_n\ker H_{DDn}\subseteq\ker H_{DDn}$. Under the isomorphisms
$\rmH_{Dn}(\scH)\simeq\ker H_{DDn}$ of eq. \ceqref{smplxhomol16}, the resulting action of $U_n$ on
$\ker H_{DDn}$ re\-presents the aforementioned homology automorphism.

Simple simplicial quantum circuits can be analyzed in similar fashion also in the normalized simplicial
Hilbert homological framework elaborated in subsect. \cref{subsec:normal}, which as we have seen
is the one best suited for homology computation.
Let us see this in some detail. A straightforward verification using \ceqref{qusmplx32} and \ceqref{simpcirc2} shows
that for each $n$ the operator $U_n$ commutes with the projectors
$\varPi_{ni}$. By \ceqref{normal2}, $U_n$ commutes then also with the projector $\varPi_n$,
\begin{equation}
U_n\varPi_n=\varPi_nU_n.
\label{simpcirc5}  
\end{equation}
The degenerate $n$--simplex space ${}^s\!\scH_n$ in eq. \ceqref{normal0} is so invariant under $U_n$,
as $\varPi_n$ projects on ${}^s\!\scH_n$.
Its orthogonal complement ${}^c\!\scH_n$ is then
invariant too, by the unitarity of $U_n$, and the restriction ${}^cU_n$ of $U_n$ to ${}^c\!\scH_n$
is a unitary operator of ${}^c\!\scH_n$. By \ceqref{simpcirc3}, the normalized boundary
operator ${}^cQ_{Dn}$ defined in \ceqref{normal4} satisfies then
\begin{equation}
{}^cU_{n-1}{}^cQ_{Dn}={}^cQ_{Dn}{}^cU_n.
\label{simpcirc6}  
\end{equation}
Relation \ceqref{simpcirc6}, analogously to \ceqref{simpcirc3}, shows that the circuit defines a unitary chain
endomorphism of the normalized simplicial Hilbert face chain complex $({}^c\!\scH, {}^cQ_D)$
and so each circuit component $U_n$ yields an automorphism of the normalized homology space $\rmH_{Dn}({}^c\!\scH)$
and consequently, by virtue of the isomorphism \ceqref{normal-3} again, one
of the simplicial homology $\rmH_n(\sfX,\bbC)$.
The automorphism can be described similarly to the unnormalized case. From \ceqref{normal14} and \ceqref{simpcirc6},
the normalized \pagebreak simplicial Hilbert Laplacian ${}^cH_{DDn}$ obeys the relations 
\begin{equation}
{}^cH_{DDn}{}^cU_n={}^cU_n{}^cH_{DDn}
\label{simpcirc7}  
\end{equation}
analogous to \ceqref{simpcirc4}. Proceeding as done earlier, one finds that the homology 
automorphism is realized as an action of ${}^cU_n$ on $\ker{}^cH_{DDn}$. 

\subsection{\textcolor{blue}{\sffamily Expected extensions of the quantum simplicial framework
}}\label{subsec:qredux}

In this concluding subsection, we discus a few possible ramifications of the quantum simplicial framework
we have developed above. 

In computational topology, one aims to the computation of interesting topological invariants
of a relevant topological space, in particular this latter's homology. This requires a simplicial model
of the topological space. By its nature the model is not unique. While the results the model
provides cannot depend on its choice, the computational effort required to obtain those results does.
This brings to the forefront the problem of reduction: given a simplicial model of a topological space,
derive from it an equivalent simpler model that yields the same results at a lower computational cost.
The natural question arises about whether reduction can be implemented in the quantum simplicial framework
elaborated in this paper.

Our simplicial set--up involves a fixed parafinite simplicial set as an input datum, that is a collection
of simplices sorted according to their degrees together with a set of face and degeneracy maps. By
contrast, a reduction algorithm changes this datum. A possible way of circumventing the problem arising here
would be working within a large ``ambient' simplicial set, in which the given simplicial set and its reductions
are contained as subsets. It remains to be seen whether this is indeed feasible and
the resource cost of a set--up of this kind does not offset the computational advantage of reduction. 

We have characterized a simplicial quantum circuit as a quantum circuit based on the simplicial quantum register 
whose operation is compatible with the structure of the underlying simplicial set. 
Mathematically, a simplicial quantum circuit is a collection of unitary operators satisfying the simplicial conditions
stated in def. \cref{defi:ssqc} in the simple case and in def. \cref{defi:pasqc} in the general case.
Any simplicial quantum circuit performing an ordinary simplicial computation should satisfy the conditions
listed in those definitions. However, we ask, is any simplicial quantum circuit associated with one such calculation?

The problem posed in the previous paragraph is in a sense one of `reversed engineering'.  
We cannot provide a solution presently. We note however that a negative answer would show the existence of
new 'quantum' operations in algebraic topology, whose import for this discipline would have to be
explored.

\vfill\eject

\renewcommand{\sectionmark}[1]{\markright{\thesection\ ~~#1}}

\section{\textcolor{blue}{\sffamily Implementation of the quantum simplicial set framework}}\label{sec:simpcircappls}

In this section, relying on the quantum simplicial framework worked out in
sect. \cref{sec:setup}, we examine the problems which may arise in the implementation
of simplicial set theoretic topological algorithms in a quantum
computer taking into account its finite storage capabilities.
The results found are only preliminary and the issue will require further
closer examination in future work. We present no new algorithms and limit ourselves to
formulate some necessary conditions for their working.

The issues analyzed below range from truncation and skeletonization of a parafinite simplicial set to
its digital encoding by simplex counting and parametrizing. The way such operations 
are implemented in the quantum simplicial framework is then investigated. 
We also outline mainly for illustrative purposes
an algorithmic scheme combining a number of basic quantum algorithms
capable of computing the complex simplicial homology spaces and Betti numbers
of the simplicial set along the lines of that of ref. \ccite{Lloyd:2014lgz}. 


\subsection{\textcolor{blue}{\sffamily Truncating and skeletonizing simplicial sets and Hilbert spaces
}}\label{subsect:trunc}

Finite simplicial complexes have a finite total number of simplices.
By contrast, parafinite simplicial sets always have an infinite total
number of simplices. Even the non degenerate simplices
may be infinitely many.

Algorithms implemented on computers can process only a finite amount of input data.
Therefore, any algorithms of computational topology cannot be grounded on a simplicial set 
containing infinitely many simplices but only on an approximation of it including only finitely many
of them. This is achieved by setting a cut--off on simplicial degree so that only simplices
of degree not exceeding the cut-off are kept. This approach goes under the name of simplicial
truncation. Next, we describe this approach in more formal terms,  

Fix a cut--off integer $N\in\bbN$. 

\begin{defi} \label{def:trsset}
An $N$--truncated simplicial set $\sfX$ consists of a collection of sets $\sfX_n$, $0\leq n\leq N$,
and mappings $d_{ni}:\sfX_n\rightarrow \sfX_{n-1}$, $1\leq n\leq N$, $i=1,\dots,n$, and 
$s_{ni}:\sfX_n\rightarrow \sfX_{n+1}$, $0\leq n\leq N-1$, $i=1,\dots,n$, obeying the simplicial relations
\eqref{smplsets1/5} when defined. 
\end{defi}

\noindent Comparing defs. \cref{def:sset} and \cref{def:trsset}, we realize that an $N$--truncated simplicial set
is much like a simplicial set except for the existence of an upper bound $N$ to simplicial degree. 

\begin{defi} \label{def:trsmor}
A morphism $\phi:\sfX\rightarrow\sfX'$ of $N$--truncated simplicial sets consists in a collection of maps
$\phi_n:\sfX_n\rightarrow\sfX'{}_n$ with $0\leq n\leq N$ obeying the simplicial morphism relations
\ceqref{smplsets6/7} when defined. 
\end{defi}

\noindent Again, inspection of defs. \cref{def:smor} and \cref{def:trsmor} reveals that the simplicial morphisms
of $N$--truncated simplicial sets are the truncated analog of the simplicial morphisms of simplicial sets.
The operations of Cartesian product and disjoint union extend also to $N$--truncated simplicial sets.
$N$--truncated simplicial sets form so a bimonoidal category $\ul{\rm sSet}_N$ analogous to the
bimonoidal category $\ul{\rm sSet}$ of simplicial sets. The homology of $N$--truncated simplicial sets is defined
in the usual manner though it is limited to degree $n\leq N-1$.

There exists an obvious truncation functor $\tr_N:\ul{\rm sSet}\rightarrow\ul{\rm sSet}_N$
which discards all the simplices of degree $n>N$ of the simplicial sets on which it acts.  It can be shown
that $\tr_N$ admits a left adjoint functor $\lk_N:\ul{\rm sSet}_N\rightarrow\ul{\rm sSet}$, its so--called 
left Kan extension \ccite{Kan:1958afc}. The resulting composite functor
$\sk_N=\lk_N\circ\tr_N:\ul{\rm sSet}\rightarrow\ul{\rm sSet}$ goes under the name of $N$-skeleton functor 
and can be characterized as follows. 
If $\sfX\in\ul{\rm sSet}$ is a simpli\-cial set, then $\sk_N\sfX$ is the smallest simplicial subset of $\sfX$
such that $\sk_N\sfX_n=\sfX_n$ for $n\leq N$ and $\sk_N\sfX_n\subseteq{}^s\sfX_n$ for $n>N$,
where ${}^s\sfX_n\subseteq\sfX_n$ is the subset of degenerate $n$--simplices (cf. subsect.
\cref{subsec:smplsets}, def. \cref{defi:degsiplx}).
In this way, $\sk_N\sfX$ reproduces $\sfX$ in degree $n\leq N$ whilst in degree $n>N$ it keeps only certain degenerate
simplices of $\sfX$.


\begin{defi}
A simplicial set $\sfX$ is called $N$--skeletal if $\sfX=\sk_N\sfX^*$
for some simplicial set $\sfX^*$.   
\end{defi}

Thee following property of simplicial homology is now fairly
evident. 

\begin{prop}
Let $\msG$ be an Abelian group. Then, it holds that in degree $n\leq N-1$
$\rmH_n(\sfX,\msG)\simeq\rmH_n(\tr_N\sfX,\msG)\simeq\rmH_n(\sk_N\sfX,\msG)$. 
\end{prop}

\noindent Therefore, working with the $N$--truncation or $N$--skeletonization of a simplicial set $\sfX$
rather than with $\sfX$ itself still allows one to recover the homology of $\sfX$ up to degree $N-1$
inclusive.

The above discussion can be extended in an evident fashion to simplicial objects, in particular to
simplicial Hilbert spaces, which are our main focus. 

An $N$--truncated simplicial Hilbert space $\scH$ is a collection of Hilbert spaces $\scH_n$, $0\leq n\leq N$,
and operators $D_{ni}:\scH_n\rightarrow \scH_{n-1}$, $1\leq n\leq N$, $i=1,\dots,n$, and 
$S_{ni}:\scH_n\rightarrow \scH_{n+1}$, $0\leq n\leq N-1$, $i=1,\dots,n$, which constitute an $N$--truncated
simplicial set when these are regarded as sets and maps of sets.

A morphism $\varPhi:\scH\rightarrow\scH'$ of $N$--truncated simplicial Hilbert spaces consists in a collection of
linear operators $\varPhi_n:\scH_n\rightarrow\scH'{}_n$ with $0\leq n\leq N$, which constitute a map of
$N$--truncated simplicial sets when they are regarded as maps of sets.
 
The $N$--truncation and $N$--skeletonization functors can be built also for simplicial objects and in
particular for simplicial Hilbert spaces. Hence, with any simplicial Hilbert space $\scH$ we can associate
its truncation $\tr_N\scH$ and $N$--skeleton $\sk_N\scH$, which have the property that
$\tr_N\scH_n=\sk_N\scH_n=\scH_n$ for $n\leq N$ and $\sk_N\scH_n\subseteq {}^s\!\scH_n$
for $n>N$, ${}^s\!\scH_n$ being the degenerate $n$-simplex subspace (cf. def. \cref{def:degsub}).
A simplicial Hilbert space is $N$--skeletal when $\scH=\sk_N\scH^*$ for some
 simplicial Hilbert space $\scH^*$. 

In computational topology, setting a cut-off $N$ on the simplicial degree of the relevant parafinite
simplicial set $\sfX$ is tantamount to replacing $\sfX$ by its $N$--truncation $\tr_N\sfX$. $\tr_N\sfX$,
however, belongs to the category of $N$--truncated simplicial sets, which is related to but distinct
from the category of simplicial sets. To remain within this latter while essentially
keeping the essence of the truncation operation, one needs to consider, instead that $\tr_N\sfX$,
the $N$--skeleton $\sk\sfX_N$ of $\sfX$. Both the truncation $\tr_N\sfX$ and the skeleton $\sk\sfX_N$
may be viewed as an approximation of $\sfX$ in the appropriate sense. 
In practice, one works with $\tr_N\sfX$. In more formal considerations, dealing with $\sk\sfX_N$ may be
more natural, since it allows to use the analysis carried out in sect. \cref{sec:setup} simply by restricting
to $N$--skeletal simplicial sets. 

In the quantum simplicial set framework of subsects. \cref{subsec:qusmplx}, \cref{subsec:hilbfunct},
to each parafinite simplicial set $\sfX$ there corresponds a simplicial Hilbert space $\scH$.
The simplicial Hilbert encoding map, which defines the simplex basis,
is a simplicial set morphism $\varkappa:\sfX\rightarrow\scH$. The
$N$--truncation functor so yields a map
$\tr_N\varkappa:\tr_N\sfX\rightarrow\tr_N\scH$ of $N$--truncated simplicial sets 
with components $\tr_N\varkappa_n=\varkappa_n$ for $0\leq n\leq N$. 
Similarly, the $N$--skeletonization functor yields a map
$\sk_N\varkappa:\sk_N\sfX\rightarrow\sk_N\scH$ of $N$--skeletal simplicial sets 
with components $\sk_N\varkappa_n=\varkappa_n$ for $0\leq n\leq N$. Therefore, the
operations of $N$--truncation and $N$--skeletonization of parafinite simplicial sets
turn under simplicial Hilbert encoding
into the corresponding operations of the associated simplicial Hilbert spaces.



\subsection{\textcolor{blue}{\sffamily Simplicial digital encoding and quantum simplicial set framework
}}\label{subsect:encsimpl}

The digital encoding of the simplices of a given parafinite simplicial set is a precondition for
the implementation of simplicial set based algorithms of computational topology in a quantum computer.
This matter is analyzed in detail in the present subsection. 

Consider a parafinite simplicial set $\sfX$ and a truncation $\tr_N\sfX$ of it.
The full simplex set of the truncation is 
\begin{equation}
\sfX^{(N)}=\mycom{{}_\squuu}{{}_{0\leq n\leq N}}\sfX_n.
\label{encsimpl0}
\end{equation}
To encode the simplices of $\sfX^{(N)}$, one needs a $k$--bit register with $k\geq \kappa_{\sfX N}$, where
\begin{equation}
\kappa_{\sfX N}=\min\big\{l\hfpt\big|\hfpt l\in\bbN, |\sfX^{(N)}|\leq 2^l\big\}.
\label{encsimpl1}
\end{equation}
We shall show next how the encoding creates a digitized image of the whole simplicial structure
of $\tr_N\sfX$ in the register.  


Let $B_2=\{0,1\}$ be the digital Boolean domain. 

\begin{defi}\label{defi:digienc}
A digital encoding of $\tr_N\sfX$ in a length $k$ register is a bijective map
$\chi:\sfX^{(N)}\rightarrow \sfX_\chi{}^{(N)}$, where $\sfX_\chi{}^{(N)}\subseteq B_2{}^k$ is a $k$--bit string set
with $|\sfX_\chi{}^{(N)}|=|\sfX^{(N)}|$. 
\end{defi}
\noindent
The images via $\chi$ of the simplex sets $\sfX_n$ are the subsets
$\sfX_{\chi n}:=\chi(\sfX_n)\subseteq \sfX_\chi{}^{(N)}$.
They constitute a partition of $\sfX_\chi{}^{(N)}$, so that  
\begin{equation}
\sfX_\chi{}^{(N)}=\mycom{{}_\squuu}{{}_{0\leq n\leq N}}\sfX_{\chi n}.
\label{encsimpl9}
\end{equation}
From here, it is promptly verified that  the register contains a full digital image of $\tr_N\sfX$. 
The restrictions $\chi\big|_{\sfX_n}$ of $\chi$ to the $\sfX_n$ induce bijective maps
$\chi_n:\sfX_n\rightarrow \sfX_{\chi n}$, through which 
further maps $d_{\chi ni}:\sfX_{\chi n}\rightarrow \sfX_{\chi n-1}$, $1\leq n\leq N$, $i=1,\dots,n$, and 
$s_{\chi ni}:\sfX_{\chi n}\rightarrow \sfX_{\chi n+1}$, $0\leq n\leq N-1$, $i=1,\dots,n$, can be defined by
\begin{align}
d_{\chi ni}=\chi_{n-1}d_{ni}\chi_n{}^{-1},
\label{encsimpl10}
\\
s_{\chi ni}=\chi_{n+1}s_{ni}\chi_n{}^{-1}.
\label{encsimpl11}
\end{align}
These $d_{\chi ni}$, $s_{\chi ni}$ obey the simplicial relations \ceqref{smplsets1/5}
as a consequence of $d_{ni}$, $s_{ni}$ doing so.
The sets $\sfX_{\chi n}$ and the maps $d_{ni}$, $s_{ni}$ are so
the simplex sets and the face and degeneracy maps of an $N$--truncated 
simplicial set $\sfX_\chi$. 
Moreover, the maps $\chi_n$ are the components of an $N$--truncated simplicial set isomorphism
$\chi:\tr_N\sfX\rightarrow\sfX_\chi$.

Let $X_0\subseteq B_2{}^k$ be a $k$--bit string set
such that $|X_0|=|\sfX^{(N)}|$. If $\chi_0$ is a reference encoding of $\tr_N\sfX$ with $\sfX_{\chi_0}{}^{(N)}=X_0$
and $\pi$ is any permutation of $X_0$, then $\chi=\pi\chi_0$ \linebreak 
is an encoding of $\tr_N\sfX$ with $\sfX_\chi{}^{(N)}=X_0$ too. Furthermore, each encoding $\chi$  with $\sfX_\chi{}^{(N)}=X_0$
is of this form for precisely one permutation $\pi$,
viz $\pi=\chi\chi_0{}^{-1}$. Therefore, there are altogether $|\sfX^{(N)}|!$ encodings with a given range $X_0$. 

We now shall examine in detail how a digital encoding $\chi$ of the $N$--truncation of
a parafinite simplicial set $\sfX$ in a $k$--bit register is implemented in the quantum simplicial set 
framework of subsects. \cref{subsec:qusmplx}, \cref{subsec:hilbfunct}.
In subsect. \cref{subsect:trunc}, we saw that when $\sfX$ is replaced by
its $N$--truncation $\tr_N\sfX$, the simplicial Hilbert space of $\scH$ of $\sfX$ gets 
replaced by its $N$--truncation $\tr_N\scH$. The simplicial quantum register of $\sfX^{(N)}$ is
\begin{equation}
\scH^{(N)}=\mycom{{}_\ddd}{{}_{0\leq n\leq N}}\scH_n.
\label{encsimpl12}
\end{equation}
By \ceqref{encsimpl0}, $\dim\scH^{(N)}=|\sfX^{(N)}|$, because $\dim\scH_n=|\sfX_n|$. 
In the Hilbert set--up, the $k$--bit register turns into an $k$--qubit register with quantum 
Hilbert space $\scQ^{\hfpt\otimes\hfpt k}$, where $\scQ=\bbC^2$. 
The encoding $\chi$ yields a linear operator $U_\chi:\scH^{(N)}\rightarrow \scH_\chi{}^{(N)}$, where 
$\scH_\chi{}^{(N)}\subseteq\scQ^{\hfpt\otimes\hfpt k}$ is a $k$--bit string space of dimension $\dim\scH_\chi{}^{(N)}=|\sfX^{(N)}|$.
Explicitly
\begin{equation}
U_\chi=\mycom{{}_\sss}{{}_{0\leq n\leq N}}\mycom{{}_\sss}{{}_{\sigma_n\in\sfX_n}}\ket{\chi{\sigma_n}}_k\bra{\sigma_n},
\label{}
\end{equation}
where the kets $\ket{\xi}_k$, $\xi\in B_2{}^k$, are those of the computational basis of $\scQ^{\hfpt\otimes\hfpt k}$. 
$U_\chi$ is evidently unitary. $U_\chi$ generates an image of the direct sum structure of $\scH^{(N)}$
in $\scH_\chi{}^{(N)}$: setting $\scH_{\chi n}=U_\chi\scH_n$, one has indeed 
\begin{equation}
\scH_\chi{}^{(N)}=\mycom{{}_\ddd}{{}_{0\leq n\leq N}}\scH_{\chi n}.
\label{}
\end{equation}
The restrictions $U_\chi\big|_{\scH_n}$ of $U_\chi$ to the $\scH_n$ induce unitary operators 
$U_{\chi n}:\scH_n\rightarrow \scH_{\chi n}$ through which 
further operators $D_{\chi ni}:\scH_{\chi n}\rightarrow \scH_{\chi n-1}$, $1\leq n\leq N$, $i=1,\dots,n$, and 
$S_{\chi ni}:\scH_{\chi n}\rightarrow \scH_{\chi n+1}$, $0\leq n\leq N-1$, $i=1,\dots,n$, can be constructed by
\begin{align}
D_{\chi ni}=U_{\chi n-1}D_{ni}U_{\chi n}{}^{-1},
\label{}
\\
S_{\chi ni}=U_{\chi n+1}S_{ni}U_{\chi n}{}^{-1}.
\label{}
\end{align}
$D_{\chi ni}$, $S_{\chi ni}$ obey the simplicial identities \ceqref{qusmplx7/11}
as a consequence of $D_{ni}$, $S_{ni}$ doing so. In fact, in the computational basis of $\scQ^{\hfpt\otimes\hfpt k}$ 
these operators are given by expressions analogous to \ceqref{qusmplx1}, \ceqref{qusmplx2}, viz
\begin{align}
&D_{\chi ni}=\mycom{{}_\sss}{{}_{\xi_n\in \sfX_{\chi n}}}\ket{d_{\chi ni}\xi_n}_{k\hfpt k}\bra{\xi_n},
\label{}
\\
&S_{\chi ni}=\mycom{{}_\sss}{{}_{\xi_n\in \sfX_{\chi n}}}\ket{s_{\chi ni}\xi_n}_{k\hfpt k}\bra{\xi_n}.
\label{}
\end{align}
The Hilbert spaces $\scH_{\chi n}$ and the operators $D_{\chi ni}$, $S_{\chi ni}$ are thus 
the simplex spaces and the face and degeneracy operators of an $N$--truncated 
simplicial Hilbert space $\scH_\chi$. 
Moreover, the maps $U_{\chi n}$ are the components of an $N$--truncated simplicial Hilbert space
unitary operator $U_\chi:\tr_N\scH\rightarrow\scH_\chi$.  

The above analysis aims only to showing the possibility of creating trough a digital encoding $\chi$
a digitized image of the truncation $\tr_N\sfX$ of a parafinite simplicial set $\sfX$ and providing
a precise formal characterization of such image and similarly for the corresponding truncation $\tr_N\scH$
of the associated simplicial Hilbert space $\scH$. Depending on the specific features of $\sfX$, there may be
special instances of the encoding $\chi$ with high efficiency and distinctive formal properties.
In particular, $\chi$ should be selected
judiciously in such a way to render the face and degeneracy maps $d_{\chi ni}$, $s_{\chi ni}$ as simple as
possible. There is no general prescription for doing that and $\chi$ must be chosen on a case
by case basis. 
By contrast, in the simplicial complex framework of ref. \ccite{Lloyd:2014lgz}, there is a
canonical encoding of the simplices of the relevant simplicial complex in terms of which
the boundary maps have a simple form.


\subsection{\textcolor{blue}{\sffamily Counting and parametrizing simplices
}}\label{subsect:cnssimpset}

Subsect. \cref{subsect:encsimpl} provides a theoretical analysis of the  
digital encoding of a truncation of a simplicial set.
Explicit construction of an encoding requires however further scrutiny
of this matter, which we do in this subsection.



Let $\sfX$ be a parafinite simplicial set. 
In subsect. \cref{subsec:smplsets}, for each $n$ we considered the subset
${}^s\sfX_n\subseteq\sfX_n$ of degenerate simplices of $\sfX_n$.
${}^c\sfX_n=\sfX_n\setminus{}^s\sfX_n\subseteq\sfX_n$ is hence the subset of
non degenerate simplices of $\sfX_n$.

The following theorem is an important structural property of simplicial sets. 

\begin{theor} \label{prop:ezlemma} (Eilenberg--Zilber lemma \ccite{Eilenberg:sst1950})
For every $n$, each simplex $\sigma_n\in\sfX_n$ has a unique representation
$\sigma_n=s_{n-1j_{n-m-1}}\cdots s_{mj_0}\tau_m$, where $m\leq n$,
$\tau_m\in{}^c\sfX_n$ and $0\leq j_0<\ldots<j_{n-m-1}\leq n-1$. 
\end{theor}

\noindent When $\sigma_n$ is non degenerate the degeneracy map string $s_{n-1j_{n-m-1}}\cdots s_{mj_0}$ is empty.

By the Eilenberg--Zilber lemma, we have 
\begin{equation}
\sfX_n\simeq\mycom{{}_\uuu}{{}_{0\leq m\leq n}}J^n{}_m \times{}^c\sfX_m,
\label{cnssimphilb1}
\end{equation}
where for $m\leq n$ $J^n{}_m=\{(j_{n-m-1},\ldots,j_0)\hfpt|\hfpt 0\leq j_0<\ldots<j_{n-m-1}\leq n-1\}$
is the set of index strings of height $n-1$ and length $n-m$. 
Note that $J^n{}_n=\{\emptyset\}$, where $\emptyset$ denotes the empty index string.
It is a simple combinatorial exercise to show that $|J^n{}_m|=\binom{n}{m}$. By \ceqref{cnssimphilb1}, so, 
the number $|\sfX_n|$ of $n$--simplices can be expressed in terms of the 
numbers $|{}^c\sfX_m|$ of non degenerate $m$--simplices with $m\leq n$ as 
\begin{equation}
|\sfX_n|=\mycom{{}_\sss}{{}_{0\leq m\leq n}}\text{\small $\binom{n}{m}$}|{}^c\sfX_m|. 
\label{cnssimphilb2}
\end{equation}
The total to non degenerate $n$--simplex ratio
\begin{equation}
\varrho_{\sfX n}=|\sfX_n|/|{}^c\sfX_n|
\label{cnssimphilb3}
\end{equation}
is an important indicator of the incidence of degenerate $n$--simplices. While $\varrho_{\sfX 0}=1$,
$\varrho_{\sfX n}$ as a rule grows very rapidly as $n$ gets large.
The total number of simplices of an $N$--truncation of $\sfX$ reads from \ceqref{cnssimphilb2} as
\begin{equation}
|\sfX^{(N)}|=\mycom{{}_\sss}{{}_{0\leq n\leq N}}|\sfX_n|=\mycom{{}_\sss}{{}_{0\leq m\leq N}}
\text{\small $\binom{N+1}{m+1}$}|{}^c\sfX_m|
\label{cnssimphilb-1}
\end{equation}
(cf. eq. \ceqref{encsimpl0}).
The content of the non degenerate $n$--simplex sets ${}^c\sfX_m$ depends on the underlying simplicial
set $\sfX$. The above numerical measures of the simplex distribution and the size of a truncation of a
simplicial set can consequently be evaluated only on a case by case basis.

The following simple examples serve as an illustration of the general techniques we described above.

\begin{exa} \label{ex:deloopcount} The nerve of the delooping of a finite group. 

\noindent
{\rm Recall that a group $\msG$ can be viewed as a one--object groupoid $\rmB\msG$, the delooping of $\msG$. 

Consider the nerve $\sfN\rmB\msG$ of the delooping $\rmB\msG$
of a finite group $\msG$ (cf. ex \cref{ex:nerve}). Then, $\sfN_n\rmB\msG=\msG^n$.
The non degenerate $n$--simplices of $\sfN\rmB\msG$ are precisely
the $n$--tuples of $\msG^n$ which do not contain the identity of $\msG$. Therefore,
\begin{equation}
|{}^c\sfN_n\rmB\msG|=(|\msG|-1)^n.
\label{cnssimphilb-3}
\end{equation}
The number of non degenerate $n$--simplices such that $n\leq N$ is consequently
\begin{equation}
\mycom{{}_\sss}{{}_{0\leq n\leq N}}|{}^c\sfN_n\rmB\msG|=\frac{(|\msG|-1)^{N+1}-1}{|\msG|-2}
\label{}
\end{equation}
(when $|\msG|=2$, this takes the value $N+1$). 
Inserting \ceqref{cnssimphilb-3} in the general formula \ceqref{cnssimphilb2}, we recover
the known number of $n$--simplices of $\sfN\rmB\msG$ 
\begin{equation}
|\sfN_n\rmB\msG|=\mycom{{}_\sss}{{}_{0\leq m\leq n}}\text{\small $\binom{n}{m}$}|{}^c\sfN_m\rmB\msG|=|\msG|^n.
\label{}
\end{equation}
The total to non degenerate $n$--simplex ratio of $\sfN\rmB\msG$ is thus 
\begin{equation}
\varrho_{\sfN\rmB\msG n}=|\sfN_n\rmB\msG|/|{}^c\sfN_n\rmB\msG|=\bigg(\frac{|\msG|}{|\msG|-1}\bigg)^n.
\label{}
\end{equation}
This grows exponentially with $n$, but the larger $|\msG|$ is the slower this growth is.

The total number of simplices of the $N$--truncation $\tr_N\sfN\rmB\msG$ of $\sfN\rmB\msG$ reads as 
\begin{equation}
|\sfN^{(N)}\rmB\msG|=\frac{|\msG|^{N+1}-1}{|\msG|-1}
\label{}
\end{equation}
(when $|\msG|=1$, this takes the value $N+1$).
The encoding of 
$\tr_N\sfN\rmB\msG$ requires therefore a $k$--bit register with $k\geq \varkappa_{\sfN\rmB\msG N}$,
where $\varkappa_{\sfN\rmB\msG N}=\log_2|\sfN^{(N)}\rmB\msG|$. We note that
$\varkappa_{\sfN\rmB\msG N}=N\log_2|\msG|+O(1/|\msG|)$ when $|\msG|$ is large.
}
\end{exa}

\begin{exa} \label{ex:complxcount} The simplicial set of a finite ordered discrete simplicial complex.

\noindent
{\rm
Let $V=\{v_0,\ldots,v_d\}$ be a finite non empty set.
Let $\clP_V$ be the discrete simplicial complex of $V$, that is the simplicial complex
whose vertex set $\Vrtc_{\clP_V}$  is $V$ and whose simplex set $\Simp_{\clP_V}$ is the power set of $V$.
We assume that $V$ is endowed with a total ordering so that $v_a<v_b$ for $a<b$. 
$\clP_V$ is so an ordered simplicial complex.

Consider the simplicial set $\sfK\clP_V$ associated with the 
complex $\clP_V$ (cf. ex. \cref{ex:simp}). 
The $n$--simplices of $\clP_V$ constitute precisely the non degenerate $n$--simplices
of $\sfK\clP_V$. The number of $n$--simplices of $\clP_V$ is $\binom{d+1}{n+1}$ for $n\leq d$.
This indicates us also the number of the non degenerate $n$--simplices of $\sfK\clP_V$
\begin{align}
&|{}^c\sfK_n\clP_V|=\text{\small $\binom{d+1}{n+1}$} \quad\text{for $n\leq d$},
\label{cnssimphilb-2}
\\
&\hphantom{|{}^c\sfK_n\clP_V|\,=\,}0 \quad \text{for $n> d$}.
\nonumber
\end{align}
The number of non degenerate $n$--simplices such that $n\leq N$ with $N\leq d$ is found from here
to be given by the expression
\begin{equation}
\mycom{{}_\sss}{{}_{0\leq n\leq N}}|{}^c\sfK_n\clP_V|
 =2^{d+1}-1-\text{\small $\binom{d+1}{N+2}$}{}_2\mhfpt F{}_1(1,-d+N+1;N+3;-1).
\label{}
\end{equation}
The total number of non degenerate simplices is so $2^{d+1}-1$. 
Inserting \ceqref{cnssimphilb-2} in the general formula \ceqref{cnssimphilb2}, we obtain the number of $n$--simplices 
of $\sfK\clP_V$ for $n\leq d$
\begin{equation}
|\sfK_n\clP_V|=\mycom{{}_\sss}{{}_{0\leq m\leq n}}\text{\small $\binom{n}{m}$}|{}^c\sfK_m\clP_V|
=\text{\small $\binom{d+n+1}{n+1}$}.
\label{}
\end{equation}
The total to non degenerate $n$--simplex ratio of $\sfK\clP_V$
for $n\leq d$ is thus 
\begin{equation}
\varrho_{\sfK\clP_V n}=|\sfK_n\clP_V|/|{}^c\sfK_n\clP_V|=\text{\small $\binom{d+n+1}{n+1}\Big/\binom{d+1}{n+1}$}.
\label{}
\end{equation}
$\varrho_{\sfK\clP_V n}$ has the following expansions:
\begin{align}
&\varrho_{\sfK\clP_V n}=1+O(n^2/d) \quad\text{for $1\ll n\ll d^{1/2}$},
\label{}
\\
&\hphantom{\varrho_{\sfK\clP_V n}\,=\,}
\frac{2^{2d+1}}{(\pi d)^{1/2}}
[1+O(d^{-1},(n-d)\log_2d)]\quad \text{for $1\ll n\to d$}.
\nonumber
\end{align}
So, while for $n\ll d^{1/2}$ the numbers of degenerate and non degenerate $n$--simplices are
comparable, for $n\sim d$ the number of the degenerate simplices is exponentially greater that that of the
non degenerate ones.

The total number of simplices of the $N$--truncation of $\tr_N\sfK\clP_V$ of $\sfK\clP_V$ is 
\begin{equation}
|\sfK^{(N)}\clP_V|=\text{\small $\binom{d+N+2}{d+1}$}-1
\label{}
\end{equation}
provided $N\leq d$.
The encoding of 
$\tr_N\sfK\clP_V$ requires therefore a $k$--bit register with $k\geq \varkappa_{\sfK\clP_V N}$,
where $\varkappa_{\sfK\clP_V N}=\log_2|\sfK^{(N)}\clP_V|$. $\varkappa_{\sfK\clP_V N}$ has the expansion
\begin{align}
\varkappa_{\sfK\clP_V N}=
&\log_2\bigg[\Big(\frac{\rme d}{N}\Big)^N \frac{d}{(2\pi)^{1/2}N^{3/2}}\bigg]
+O(1/N,N^2/d) \quad\text{for $1\ll N\ll d^{1/2}$},
\label{}
\\
&\,\,2d+2-\tfrac{1}{2}\log_2(\pi d)+O(d^{-1}, 
(N-d)\log_2d) \quad\text{for $1\ll N\to d$}.
\nonumber
\end{align}
So, one needs a $2(d+1)$--bit register to encode all the simplex data in degree
less than $d$ of the simplicial set $\sfK\clP_V$. This is to be compared with the
$d+1$--bit register required to encode the simplex data
of the underlying simplicial complex $\clP_V$ \ccite{Lloyd:2014lgz}. 

If $\clS$ is a finite simplicial complex, then $\clS$ is a subcomplex of the discrete
simplicial complex $\clP_V$ with $V=\Vrtx_\clS$. The values of such quantities as
$|{}^c\sfK_n\clP_V|$, $|\sfK_n\clP_V|$, \dots\,computed above for $\sfK\clP_V$ are then upper bounds for the
values of the corresponding quantities of $|{}^c\sfK_n\clS|$, $|\sfK_n\clS|$, \dots\, of $\sfK\clS$.
}
\end{exa}

Once the number of simplices of the chosen truncation $\tr_N\sfX$ of s simplicial set $\sfX$ is ascertained,
it becomes possible to construct a digital encoding of it (cf. def. \cref{defi:digienc})
by allocating a suitably sized register. The choice of the encoding is not unique.
The following examples of encoding are presented here as an illustration of the theory.
No claim is made that they constitute the optimal choice with regard to efficiency. 

\begin{exa} The nerve of the delooping of a finite group. 

\noindent
{\rm
Consider again the nerve $\sfN\rmB\msG$ of the delooping $\rmB\msG$
of a finite group $\msG$ (cf. ex \cref{ex:deloopcount}).       
The simplices of a truncation $\tr_N\sfN\rmB\msG$ of $\sfN\rmB\msG$ can be digitally
encoded in a $Nq+r$--bit register with $q,r$ integers such that $q\geq\log_2|\msG|$ and $r\geq\log_2(N+1)$ 
as follows. We represent a $Nq+r$--bit string as $(x_1,\ldots,x_N;y)$, where the $x_a$ are $q$--bit strings
and $y$ is an $r$--bit string. For convenience, we use the enumerative indexation of such bit strings,
in terms of which the $x_a$ are represented by integers in the range $0$ to $2^q-1$ and $y$ as an integer
in the range $0$ to $2^r-1$.  Further, we select a digital encoding of $\msG$, i.e. a bijective map
$\varphi:\msG\rightarrow P_\varphi$ with $P_\varphi\subseteq B_2{}^q$, which we normalize by requiring that
$\varphi(e)=0$, where $e$ denotes the neutral element of $\msG$.
Using these elements, we can construct a digital encoding $\chi$ of $\tr_N\sfN\rmB\msG$.
This is the bijective map $\chi:\sfN^{(N)}\rmB\msG\rightarrow\sfN_\chi{}^{(N)}\rmB\msG$ defined as follows.
The encoding range is
\begin{equation}
\sfN_\chi{}^{(N)}\rmB\msG
=\mycom{{}_\squuu}{{}_{0\leq n\leq N}}P_\varphi{}^n\times\{0\}^{N-n}\times\{n\}\subseteq B_2{}^{Nq+r}.
\label{}
\end{equation}
Further, for $\sigma_n=(g_1,\ldots,g_n)\in\sfN_n\rmB\msG$ with $n\leq N$, 
\begin{align}
&\chi(\sigma_n)=(0,\ldots,0;0) \quad \text{for $n=0$},
\label{}
\\
&\chi(\sigma_n)=(\varphi(g_1),\ldots,\varphi(g_n),0,\ldots,0;n)\quad \text{for $0<n\leq N$}.
\end{align}
The encoding's face and degeneracy maps $d_{\chi ni}$, $s_{\chi ni}$
(cf. eqs. \ceqref{encsimpl10}, \ceqref{encsimpl11}) take the following form.
We notice preliminarily  that $\sfN_{\chi n}\rmB\msG=P_\varphi{}^n\times\{0\}^{N-n}\times\{n\}$. 
Let $(x_1,\ldots,x_n,0,\ldots,0;n)\in\sfN_{\chi n}\rmB\msG$ with $n\leq N$. Then, if $1\leq n$
\begin{subequations}
\begin{align}
&d_{\chi n0}(x_1,\ldots,x_n,0,\ldots,0;n)=(x_2,\ldots,x_n,0,\ldots,0;n-1),
\label{}
\\
&d_{\chi ni}(x_1,\ldots,x_n,0,\ldots,0;n)
\label{}
\\
&\hspace{2mm}
=(x_1,\ldots,x_{i-1},\varphi(\varphi^{-1}(x_i)\varphi^{-1}(x_{i+1})),x_{i+2},\ldots,x_n,0,\ldots,0;n-1)
~~~~\text{if $0<i<n$},
\nonumber
\\
&d_{\chi nn}(x_1,\ldots,x_n,0,\ldots,0;n)=(x_1,\ldots,x_{n-1},0,\ldots,0;n-1).
\label{}
\end{align}
\end{subequations}
Similarly, if $n\leq N-1$
\begin{align}
&s_{\chi ni}(x_1,\ldots,x_n,0,\ldots,0;n)=(x_1,\ldots,x_i,0,x_{i+1},\ldots,x_n,0,\ldots,0;n+1).
\label{}
\end{align}
}
\end{exa}

\begin{exa} The simplicial set of a finite ordered discrete simplicial complex.

\noindent
{\rm
Consider the simplicial set $\sfK\clP_V$ associated with the ordered simplicial complex $\clP_V$
studied earlier (cf. ex. \cref{ex:complxcount}).
The simplices of a truncation $\tr_N\sfK\clP_V$ of $\sfK\clP_V$ can be digitally
encoded in a $(d+1)r$--bit register with $r$ is an integer such that $r\geq\log_2(N+2)$ 
as follows. We represent a $(d+1)r$--bit string as $(x_0,\ldots,x_d)$, where the $x_a$ are $r$--bit strings,
and employ again the enumerative indexation of the $r$--bit strings representing them as
integers in the range $0$ to $2^r-1$. For each index $0\leq a\leq d$, let $\varphi_a:\sfK^{(\infty)}\clP_V\rightarrow\bbN$,
where $\sfK^{(\infty)}\clP_V=\bigsqcup_{0\leq n}\sfK_n\clP_V$, be the $a$--th vertex
counting map: if $\sigma_n\in\sfK_n\clP_V$,
then $\varphi_a(\sigma_n)$ is the number of occurrences of the vertex $v_a$ in $\sigma_n$. We note that 
for $\sigma_n\in\sfK_n\clP_V$ the integers $\varphi_a(\sigma_n)$ must obey the sum rule 
$\sum_{0\leq a\leq d}\varphi_a(\sigma_n)=n+1$.
Employing these elements, we can construct a digital encoding $\chi$ of $\tr_N\sfK\clP_V$.
This is the bijective map $\chi:\sfK^{(N)}\clP_V\rightarrow\sfK_\chi{}^{(N)}\clP_V$ defined as follows.
The range 
of the encoding is 
\begin{equation}
\sfK_\chi{}^{(N)}\clP_V=
\Big\{(x_0,\ldots,x_d)\Big|0\leq x_a\leq N+1,0<\mycom{{}_\sss}{{}_{0\leq a\leq d}}x_a\leq N+1\Big\}.
\label{}
\end{equation}
Further, for $\sigma_n\in\sfK_n\clP_V$ with $n\leq N$, we have 
\begin{equation}
\chi(\sigma_n)=(\varphi_0(\sigma_n),\ldots,\varphi_d(\sigma_n)).
\label{}
\end{equation}  
The face and degeneracy maps $d_{\chi ni}$, $s_{\chi ni}$ of the encoding 
(cf. eqs. \ceqref{encsimpl10}, \ceqref{encsimpl11}) read as follows.
We notice that $\sfK_{\chi n}\clP_V=\{(x_0,\ldots,x_d)|\sum_{0\leq a\leq d}x_a=n+1\}$. 
For $(x_0,\ldots,x_d)\in\bbN^{d+1}$, set 
\begin{align}
\vartheta_{ai}(x_0,\ldots,x_d)=\,&1\quad\text{if~~~}
\mycom{{}_\sss}{{}_{0\leq b<a}}x_b\leq i<\mycom{{}_\sss}{{}_{0\leq b\leq a}}x_b,
\label{}
\\  
&0\quad\text{else}.
\nonumber
\end{align}
Let $(x_0,\ldots,x_d)\in\sfK_{\chi n}\scP_V$ with $n\leq N$. Then, for $1\leq n$ 
\begin{equation}
d_{\chi ni}(x_0,\ldots,x_d)=(x_0-\vartheta_{0i}(x_0,\ldots,x_d),\ldots,x_d-\vartheta_{di}(x_0,\ldots,x_d)),
\label{}
\end{equation}
while for $n\leq N-1$
\begin{equation}
s_{\chi ni}(x_0,\ldots,x_d)=(x_0+\vartheta_{0i}(x_0,\ldots,x_d),\ldots,x_d+\vartheta_{di}(x_0,\ldots,x_d)).
\label{}
\end{equation}
}
\end{exa}


\subsection{\textcolor{blue}{\sffamily Homology computation via quantum algorithms
}}\label{subsec:qpe}

In subsect. \cref{subsec:normal}, we showed that the complex coefficient homology spaces of a parafinite
simplicial set $\sfX$, $\rmH_n(\sfX,\bbC)$, can be computed based on the isomorphism of these latter
and the normalized simplicial Hilbert homology spaces of the associated simplicial
Hilbert space $\scH$, $\rmH_{Dn}({}^c\!\scH)$ (cf. theor. \cref{prop:hilbnorm}).
In turn, such spaces are given by $\rmH_{Dn}({}^c\!\scH)\simeq\ker{}^cH_{DDn}$, where
${}^cH_{DDn}$ is the normalized simplicial Hilbert Laplacian
expressible through the simplicial Hilbert boundary operators ${}^cQ_{Dn}$, ${}^cQ_{Dn+1}$
and their adjoints  (cf. def. \cref{def:normsmplxlapl} and theor. \cref{prop:normhodge}).
The problem we have to tackle next is the determination
of $\ker{}^cH_{DDn}$ in an $N$--truncated setting as required by the implementation of simplicial set based 
quantum algorithms. This will be done along the lines of ref. \ccite{Lloyd:2014lgz}
by exploiting at two distinct stages two well--known quantum algorithms: 

\begin{enumerate} %

\vspace{-1mm}\item Grover multisolution quantum search algorithm \ccite{Grover:1996fqs};

\vspace{-2mm}\item Abrams' and Lloyd's quantum algorithm for the solution of the general
eigenvalue problem \ccite{Abrams:1998qee}, a refinement of Kitaev's basic quantum phase estimation algorithm
\ccite{Kitaev:1995qpe}. 

\vspace{-1mm}

\end{enumerate}

\noindent The implementation of such algorithms relies on a number of other algorithms, as will be explained
in due course. The total complexity of the whole algorithm is determined by the combined complexity
of all the algorithms which compose it. 

To study normalized Hilbert homology in an $N$--truncated setting, we have to consider the $N+1$--tuple of equations
${}^cH_{DDNn}\ket{\psi_n}=0$ with $\ket{\psi_n}\in{}^c\!\scH_n$, where 
\begin{equation}
{}^cH_{DDNn}={}^cH_{DDn}-\delta_{Nn}{}^cQ_{DN+1}{}^cQ_{DN+1}{}^+.
\label{qpe1}
\end{equation}
Note that ${}^cH_{DDNn}={}^cH_{DDn}$ for $0\leq n<N$ only while ${}^cH_{DDNN}\neq{}^cH_{DDN}$.
The operator ${}^cH_{DDNN}$ rather than ${}^cH_{DDN}$ appears here, as the operators
${}^cQ_{DN+1}$, ${}^cQ_{DN+1}{}^+$ are excluded by the truncation. Such operators enter only
the expression of ${}^cH_{DDN}$ (cf. eq. \ceqref{normal14} and their contribution is duly subtracted
out. 

We saw in subsect. \cref{subsect:encsimpl} that the $N$--truncation of the simplicial Hilbert space $\scH$
gives rise to the simplicial quantum register $\scH^{(N)}$ shown in eq. \ceqref{encsimpl12}.
The normalized simplicial Hilbert homology rests on the subspace ${}^c\!\scH^{(N)}$ of $\scH^{(N)}$
spanned by the non degenerate $n$--simplex subspaces ${}^c\!\scH_n$ of the $\scH_n$ with $0\leq n\leq N$, 
\begin{equation}
{}^c\!\scH^{(N)}=\mycom{{}_\ddd}{{}_{0\leq n\leq N}}{}^c\!\scH_n.
\label{qpe2}
\end{equation}
By systematically exploiting the canonical projections $R_n:{}^c\!\scH^{(N)}\rightarrow{}^c\!\scH_n$
and injections $I_n:{}^c\!\scH_n\rightarrow{}^c\!\scH^{(N)}$, 
one can view the vector spaces ${}^c\!\scH_n$ as subspaces of the space ${}^c\!\scH^{(N)}$
and similarly the operator spaces $\Hom({}^c\!\scH_n,{}^c\!\scH_m)$ as subspaces
of the space $\End({}^c\!\scH^{(N)})$. 
In what follows we shall thoroughly adhere to this perspective. Accordingly, we shall use the same notation
to denote vectors of ${}^c\!\scH_n$ and operators of $\Hom({}^c\!\scH_n,{}^c\!\scH_m)$,
e.g. $\ket{{}^c\psi_n}$, $\ket{{}^c\phi_n}$, \ldots and ${}^cA_{n,m}$, ${}^cB_{n,m}$, \ldots,  
irrespective of whether ${}^c\!\scH_n$ and $\Hom({}^c\!\scH_n,{}^c\!\scH_m)$ are considered as vector
spaces in their own or as subspaces of ${}^c\!\scH^{(N)}$ and $\End({}^c\!\scH^{(N)})$, respectively
\footnote{$\vphantom{\dot{dot{\dot{f}}}}$The notation
$\ket{n}\otimes\ket{{}^c\psi_n}$, $\ket{n}\otimes\ket{{}^c\phi_n}$,\ldots and
$\ket{n}\bra{m}\otimes{}^cA_{n,m}$, $\ket{n}\bra{m}\otimes{}^cB_{n,m}$, \ldots is instead more commonly used
for vectors of ${}^c\!\scH_n$ and operators of $\Hom({}^c\!\scH_n,{}^c\!\scH_m)$, when these latter are
regarded as subspaces of ${}^c\!\scH^{(N)}$ and $\End({}^c\!\scH^{(N)})$ to keep track of simplicial degree,
where $\ket{n}$ is the canonical basis of an auxiliary Hilbert space $\scA^{(N)}\simeq\bbC^{N+1}$.    
}.





The quantum algorithm computing the normalized Hilbert homologies $\rmH_{Dn}({}^c\!\scH)$
requires as a preliminary step the implementation of the projection of $\scH^{(N)}$ onto ${}^c\!\scH^{(N)}$.
This operation is done parallelly within each subspace $\scH_n$ with $0\leq n\leq N$
by the orthogonal projection operator $1_n-\varPi_n$, where $\varPi_n$ is the orthogonal projector
of $\scH_n$ onto ${}^s\!\scH_n={}^c\!\scH_n{}^\perp$ given in \ceqref{normal2}. That operator however
cannot be a component of any quantum circuit as such, as it is not unitary. Luckily, the projection can be achieved 
in the appropriate form compatible with unitarity using Grover's quantum search algorithm \ccite{Grover:1996fqs}
in the variant based on amplitude amplification \ccite{Brassard:1997pts}.
We illustrate briefly how the algorithm works within $\scH_n$.
The quantum computer is initialized in a state $\ket{\xi_{0n}}\in\scH_n$ 
\begin{equation}
\ket{\xi_{0n}}=\mycom{{}_\sss}{{}_{\sigma_n\in\sfX_n}}\ket{\sigma_n}|\sfX_n|^{-1/2},
\label{qpe3}
\end{equation}
that is a uniform superposition of all $n$--simplex states $\ket{\sigma_n}$.
Through the algorithm, the state $\ket{\xi_{0n}}$ evolves unitarily toward the final state
\begin{equation}
\ket{{}^c\xi_{0n}}=(1_n-\varPi_n)\ket{\xi_{0n}}(|\sfX_n|/|{}^c\sfX_n|)^{1/2} 
=\mycom{{}_\sss}{{}_{\sigma_n\in{}^c\sfX_n}}\ket{\sigma_n}|{}^c\sfX_n|^{-1/2},
\label{qpe4}
\end{equation}
which constitutes a uniform superposition of all non degenerate $n$--simplex states $\ket{\sigma_n}$.

The algorithm comprises two stages: $i)$ the preparation of the state $\ket{\xi_{0n}}$ and $ii)$
the production of the state $\ket{{}^c\xi_{0n}}$ from $\ket{\xi_{0n}}$. These stages contribute additively 
to the algorithm's complexity. In stage $i$, the state $\ket{\xi_{0n}}$ is yielded by the action of
some unitary operator $W_n$ on a fiducial reference state $\ket{o_n}$, so that $\ket{\xi_{0n}}=W_n\ket{o_n}$.
In stage $ii$, the algorithm is implemented through the action on the state
$\ket{\xi_{0n}}$ of a unitary operator $G_n{}^{p_n}$ that is the $p_n$--th power
of an elementary unitary operator $G_n$, showing the algorithm's iterative nature. 
The Grover operator $G_n$ is of the form $G_n=-W_nD_{0n}W_n{}^+D_n$, where $D_{0n}$, $D_n$
are unitary operators. $D_{0n}=1_n-2\ket{o_n}\bra{o_n}$ is the conditional sign flip operator of the
reference state $\ket{o_n}$. $D_n$ is the conditional sign flip operator of the
non degenerate simplex states $\ket{\sigma_n}$ and is oracular in nature.
The Grover iteration number $p_n$ reads as $p_n=\left[\frac{\pi}{4}(|\sfx_n|/|{}^c\sfx_n|)^{1/2}\right]$.

The overall complexity $C_n$ of the above Grover type state preparation
algorithm is given by $C_n=C(W_n)+p_n(2C(W_n)+C(D_{0n})+C(D_n))$, where
$C(W_n)$, $C(D_{0n})$, $C(D_n)$ are the complexities of the unitaries $W_n$, $D_{0n}$, $D_n$.
The values of $C(W_n)$, $C(D_n)$ depend on the underlying simplicial set $\sfX$.
Further, the values of $C(W_n)$, $C(D_{0n})$ depend on the reference state $\ket{o_n}$ chosen.
$C(W_n)$, $C(D_{0n})$, $C(D_n)$ may depend also on the digital encoding of the truncation $\tr_N\sfX$ used.
The performance of the algorithm hinges in particular on that of the oracle $D_n$.
Note that the iteration number $p_n$ depends on the total to non degenerate $n$--simplex ratio
$\varrho_{\sfX n}=|\sfX_n|/|{}^c\sfX_n|$ introduced in subsect. \cref{subsect:cnssimpset}
(cf. eq. \ceqref{cnssimphilb3}). As $\varrho_{\sfX n}$ typically grows very rapidly with $n$, the algorithm fails
when $n$ is large enough, if a maximal number of iterations is allowed. 
$\varrho_{\sfX n}$ is unknown in general. In such a case, it must be determined previously 
using a quantum counting algorithm \ccite{Brassard:1998qca}, which is an instance of phase estimation algorithm
\ccite{Kitaev:1995qpe} using the Grover operator $G_n$ as the underlying unitary operator, since
the eigenvalues $\rme^{\pm i\theta_n}$ of $G_n$ are given by $\sin(\theta_n/2)=\varrho_{\sfX n}{}^{-1/2}$.



The effective realization of the projection of $\scH^{(N)}$ onto ${}^c\!\scH^{(N)}$ renders possible the
implementation of operators on ${}^c\!\scH^{(N)}$ such as ${}^cQ_{DN+1}$, ${}^cQ_{DN+1}{}^+$ 
and ${}^cH_{DDNN}$ (cf. eqs. \ceqref{normal4} and \ceqref{normal14}) and many more as constitutive
elements of a quantum homological algorithms. Henceforth, so, we work ${}^c\!\scH^{(N)}$. 


By the isomorphism \ceqref{normal15} and \ceqref{qpe1}, $\rmH_{Dn}({}^c\!\scH)=\ker{}^cH_{DDNn}$ for $0\leq n<N$.
To compute the normalized simplicial Hilbert homology spaces $\rmH_{Dn}({}^c\!\scH)$,
we thus have to solve the $N$ equations ${}^cH_{DDNn}\ket{{}^c\psi_n}=0$ with $\ket{{}^c\psi_n}\in{}^c\!\scH_n$. 
To that end, let us introduce the following linear operator ${}^cH_{DD}{}^{(N)}$ of ${}^c\!\scH^{(N)}$: 
\begin{equation}
{}^cH_{DD}{}^{(N)}=\mycom{{}_\sss}{{}_{0\leq n\leq N}}{}^cH_{DDNn}. 
\label{qpe7}
\end{equation}
Then, the equations mentioned earlier are equivalent to a single equation, namely 
${}^cH_{DD}{}^{(N)}\ket{\psi^{(N)}}=0$, since
${}^cH_{DD}{}^{(N)}\ket{\psi^{(N)}}=\sum_{0\leq n\leq N}{}^cH_{DDNn}\ket{\psi_n}$. 
In this way, the calculation of the homology spaces $\rmH_{Dn}({}^c\!\scH)$ in degree $n<N$ \pagebreak 
is reduced  to that of the kernel of $\ker{}^cH_{DD}{}^{(N)}$.

The determination of $\ker{}^cH_{DD}{}^{(N)}$ by the quantum phase estimation methods of 
ref. \ccite{Abrams:1998qee} involves the unitary operators $\exp(i\tau{}^cH_{DD}{}^{(N)})$ for varying
$\tau$. The construction of this requires the use of a suitable Hamiltonian
simulation algorithm \ccite{Feynman:1982spc}. For the sake of computational efficiency,
it is convenient when possible to replace ${}^cH_{DD}{}^{(N)}$ with a 
Hermitian operator ${}^cB^{(N)}$ sparser than ${}^cH_{DD}{}^{(N)}$
such that $\ker{}^cB^{(N)}=\ker{}^cH_{DD}{}^{(N)}$
and construct $\exp(i\tau{}^cB^{(N)})$ instead, since the complexity of the algorithm depends
inversely on the sparsity of the exponentiated operator
\ccite{Llooyd:1996uqs,Aharonov:2003qsg}. 
A standard choice of ${}^cB^{(N)}$ is provided by the Dirac operator ${}^cB_D{}^{(N)}$
of ${}^cH_{DD}{}^{(N)}$, which is the operator of ${}^c\!\scH^{(N)}$ given by 
\begin{equation}
{}^cB_D{}^{(N)}=\mycom{{}_\sss}{{}_{0\leq n\leq N-1}}\big({}^cQ_{Dn+1}+{}^cQ_{Dn+1}{}^+\big).
\label{qpe8}
\end{equation}
${}^cB_D{}^{(N)}$ is evidently Hermitian and satisfies 
\begin{equation}
{}^cB_D{}^{(N)}{}^2={}^cH_{DD}{}^{(N)}.
\label{qpe9}
\end{equation}
Consequently, $\ker{}^cH_{DD}{}^{(N)}=\ker{}^cB_D{}^{(N)}$ as required.

The Hamiltonian simulation algorithm constructing 
$\exp(i\tau{}^cB_D{}^{(N)})$ involves an oracle unitary operator $O_{DN}$ providing
the non zero matrix elements of ${}^cB_D{}^{(N)}$. The algorithm's total complexity is 
$C_D{}^{(N)}(\tau)=G(\tau)+Y_D{}^{(N)}(\tau)C(O_{DN})$,
where $G(\tau)$, $Y_D{}^{(N)}(\tau)$ and $C(O_{DN})$ are the
algorithm's gate and query complexity and the complexity of $O_{DN}$, respectively.
$G(\tau)$, $Y_D{}^{(N)}(\tau)$ depend, besides the value of the parameter $\tau$, 
on the precision desired and the specific algorithm employed
\ccite{Berry:2014esh,Berry:2014hts,Berry:2015hod,Low:2017hsp,Childs:2017qss,Chen:2021qhs}.
$Y_D{}^{(N)}(\tau)$ depends further on the sparseness index of ${}^cB_D{}^{(N)}$,
measured as the number of non zero matrix elements of ${}^cB_D{}^{(N)}$ per row or column,
and the maximum magnitude of the matrix elements. Finally, $C(O_{DN})$ may
depend through $O_{DN}$ on the underlying simplicial set $\sfX$ and the
digital encoding of the truncation $\tr_N\sfX$ used. 


We review briefly how the algorithm of ref. \ccite{Abrams:1998qee} is implemented to find 
homology spaces $\rmH_{Dn}({}^c\!\scH)$ in our framework. The algorithm proceeds by determining the eigenvalues
and eigenvectors of the unitary operators $\exp(i\tau{}^cB_D{}^{(N)})$ introduced above 
using quantum phase estimation. It operates specifically with the density operators. For each $n$ with
$0\leq n\leq N$, let ${}^c\rho_{0n}$ be the uniform mixture of all non degenerate $n$--simplex states
$\ket{\sigma_n}\bra{\sigma_n}$. Explicitly, ${}^c\rho_{0n}$ reads as  
\begin{equation}
{}^c\rho_{0n}=|{}^c\sfX_n|^{-1}{}^c1_n
=\mycom{{}_\sss}{{}_{\sigma_n\in{}^c\sfX_n}}\ket{\sigma_n}|{}^c\sfX_n|^{-1}\bra{\sigma_n}
\label{qpe5}
\end{equation}
${}^c\rho_{0n}$ can be straightforwardly obtained from the state $\ket{{}^c\xi_{0n}}$
of eq. \ceqref{qpe4} construct\-ed earlier by adding an ancilla,
copying the simplex data into the ancilla to construct the state
$\sum_{\sigma_n\in{}^c\sfX_n}\ket{\sigma_n}\otimes\ket{\sigma_n}|{}^c\sfX_n|^{-1/2}$
and then tracing the ancilla out \ccite{Lloyd:2014lgz}. 
One adjoins next to the 'vector' register
${}^c\!\scH^{(N)}$ a $b_t$--bit 'clock' register $\scQ^{\hfpt\otimes\hfpt b_t}$ with 
$b_t$ suitably large, so that the total Hilbert space is $\scQ^{\hfpt\otimes\hfpt b_t}\otimes{}^c\!\scH^{(N)}$.
The quantum computer is initialized in the mixed state
\begin{equation}
{}^{tc}\rho_{0n}=\ket{0}_{t\,t}\bra{0}\otimes {}^c\rho_{0n},
\label{qpe10}
\end{equation}
where $\ket{\lambda}_t$, $0\leq \lambda\leq 2^{b_t}-1$, is the computational basis of $\scQ^{\hfpt\otimes\hfpt b_t}$
(in the enumerative parametrization). The algorithm evolves unitarily the state ${}^{tc}\rho_{0n}$ and
ends with a measurement of the clock register. 
Upon reiteration of the algorithm, the relevant clock value $\lambda=0$
is found with probability $\dim\ker{}^cH_{DDNn}/|{}^c\sfX_n|$. After $\lambda=0$ is
obtained, the computer is in the mixed state
\begin{equation}
{}^{tc}\rho_{DDn}=\ket{0}_{t\,t}\bra{0}\otimes{}^c\rho_{DDn},
\label{qpe11}
\end{equation}
where ${}^c\rho_{DDn}$ is the density operator of ${}^c\scH^{(N)}$ reading as 
\begin{equation}
{}^c\rho_{DDn}=\dim\ker{}^cH_{DDNn}{}^{-1}{}^cP_{DDNn},
\label{qpe12}
\end{equation}
${}^cP_{DDNn}$ denoting the orthogonal projection operator of ${}^c\!\scH_n$ onto $\ker{}^cH_{DDNn}$. 
By the isomorphism $\ker{}^cH_{DDn}\simeq\rmH_{Dn}({}^c\!\scH)$, the final state of the quantum 
computer for $\lambda=0$ encodes the homology space $\rmH_{Dn}({}^c\!\scH)$.
Further, the frequency with which $\lambda=0$ occurs furnishes 
the Betti numbers $\beta_n(\sfX,\bbC)=\dim\ker{}^cH_{DDNn}$ (cf. eq. \ceqref{betti}). 

In the algorithm described in the previous paragraph,
the value of $b_t$ depends on the number of bits and the precision desired for the estimation of the
eigenvalues of ${}^cB_D{}^{(N)}$. The algorithm involves the use of $b_t$--bit
Welsh--Hadamard and quantum Fourier transforms with combined complexity
$C_{WHT}(b_t)+C_{QFT}(b_t)=O(b_t{}^2)$ and one call of an oracle unitary operator 
$U_{DNj}$ computing $\exp(i2^j{}^cB_D{}^{(N)})$ for each $j$ with $0\leq j\leq b_t-1$. If
$C(U_{DNj})$ is the complexity of $U_{DNj}$, the total complexity of the algorithm is
$C_{QPE}{}^{(N)}=C_{WHT}(b_t)+C_{QFT}(b_t)+\sum_{0\leq j\leq b_t-1}C(U_{DNj})$.
The values of the $C(U_{DNj})$ depend on the Hamiltonian simulation algorithm employed.

We conclude this subsection with the following remark. The complexity analysis of the quantum 
algorithmic scheme for the computation of the homology of a parafinite simplicial set that
we have presented above is only preliminary. Many of the contribution to scheme's complexity
depend on the simplicial set considered and on the digital encoding of its truncation used
in a way whose precise theoretical understanding and quantitative estimation definitely
requires further investigation.

\vfill\eject

\renewcommand{\sectionmark}[1]{\markright{\thesection\ ~~#1}}

\section{\textcolor{blue}{\sffamily Open problems and outlook
}}\label{sec:outlook}

In this paper, we have attempted to construct a theoretical model of a simplicial quantum 
computer (SQC). It is an open question whether a computer of this kind can be eventually built and employed
in practice in computational topology. In this final section, we briefly review a number of issues which may arise
in this respect \footnote{$\vphantom{dot{\dot{f}}}$
The author thanks one of the referees of the paper for raising these points},
without advancing any claim for a complete solutions.


\subsection{\textcolor{blue}{\sffamily Structural issues of the SQC model
}}\label{subsec:sqcimplement}

The following objections concerning the overall structure of the SQC model can be raised. 
  
\begin{enumerate}[label=\roman*),font=\itshape] 

\vspace{-1mm}\item The applicability of the model is limited to parafinite
simplicial sets, restricting its potential usefulness for a wider range of
topological problems.

\vspace{-2mm}\item Associating a finite dimensional simplicial Hilbert space with a simplicial
set and implementing the simplicial operator set--up is arduous in a quantum computational setting.
 
\vspace{-1mm}
  
\end{enumerate}

\noindent
The reply to the first objection is provided already in the inception of sect. \cref{sec:setup}, which we quote here.
Extending the quantum simplicial framework to non parafinite simplicial sets is not feasible
because the intrinsic limitations of implementable computation:  by the way they are defined
(cf. def. \ref{def:parafinite}), parafinite simplicial sets
are the most general kind of simplicial sets that can be handled by a SQC  because 
any such device, regardless the way it is conceived and built,
necessarily can operate only on a finite number of simplicial data of each given degree. Indeed, the simplicial
complexes routinely used in topological data analysis are all instances of parafinite simplicial sets. 

The reply to the second objection is that configuring the simplicial Hilbert space and arranging its 
operator scheme in a quantum computational setting is evidently impossible because of the overall
infinite dimensional nature of these entities! 
The simplicial framework is in fact only a conceptual construct devised to simplify and useful for the analysis.
It is its truncated form, introduced in subsect. \cref{subsect:trunc}, that by virtue
of its finite dimensionality can be assembled at least in principle on a SQC,
with an effort comparable to that required by other quantum
computational set--ups. The crucial problem to be addressed here, in our judgement,
is rather the implementability of simplicial quantum circuits, perhaps in a form less idealized than that
presented in subsect. \cref{subsec:simpcirc}.


\subsection{\textcolor{blue}{\sffamily Practical realization of the SQC model
}}\label{subsec:practical}

The SQC model is admittedly abstract: its viability depends to a significant extent on theoretical assumptions
which may not hold in practice. This raises a number of issues concerning generically the implementation and
the practical realization of the SQC model, as we list next.

\begin{enumerate}[label=\roman*),font=\itshape] 

\vspace{-1mm}\item An investigation of noise and decoherence in a SQC and their the effect on the reliability of the
computations performed by it should be carried out. 

\vspace{-2mm}\item An analysis of error rates occurring in a SQC is also required; methods of error corrections
specific for a SQC should be developed. 

\vspace{-2mm}\item A detailed study of the limitations and feasibility of the quantum hardware required by
the effective functioning of a SQC should be undertaken.

\vspace{-2mm}\item Resource requirements for running algorithms on a SQC have also to be taken in due consideration. 

\vspace{-1mm}
\end{enumerate}

\noindent
All these questions are undoubtedly important, but their resolution lies beyond the reach of the analysis
carried out in the present work, which is theoretical in nature.
Anyhow, the eventual validation of the effective functioning of a SQC can ultimately be
only empirical and experimental.


\subsection{\textcolor{blue}{\sffamily Relevance of the SQC model for computational topology
}}\label{subsec:relcomptop}

The eventual relevance of the SQC model for computational topology 
should also be appraised. In particular, it is important to gauge the potential impact of the SQC model 
on the development \pagebreak of dedicated quantum algorithms for computational topology and more broadly 
its implications for the future development and advancement of algebraic topology. 
In this respect, the following issues are especially relevant:

\begin{enumerate}[label=\roman*),font=\itshape] 

\vspace{-1mm} \item the potential advantages and disadvantages of using algorithms designed for implementation on a SQC 
should be appraised in depth.

\vspace{-2mm}\item the robustness of algorithms devised for a SQC requires verification.

\vspace{-2mm}\item a comparative analysis of the effectiveness of such algorithms and their classical counterparts is required 
to evaluate the merits of the SQC model;

\vspace{-2mm}\item a comparison of the effectiveness of new SQC algorithms
and existing ones of quantum topological data analysis is also due. 

\vspace{-1mm}
\end{enumerate}

\noindent
As we have explained in sect. \cref{sec:intro}, the greater flexibility of simplicial set theory
affords in general several simplicial models of a given topological space. Each of these, in turn,
may have a variety of digital encodings. All these possibilities affect in different ways the size of the
simplex sets at hand, impinging on the complexity of quantum search algorithms, and on the sparsity properties and
matrix element size of the relevant Hamiltonians, having so a bearing on the complexity of 
quantum Hamiltonian simulation algorithms, just to mention a few of the
many issues involved. For these reasons, it is difficult to make definite general statements
on these matters. Such statements are possible only on a case by case basis and as a result of
extensive research work. 


Further points concern specifically the algorithm for the computation of simplicial
homology we have outlined in subsect. \cref{subsec:qpe}.

\begin{enumerate}[label=\roman*),font=\itshape] 

\vspace{-1mm}\item The scalability of SQC quantum algorithmic schemes should be studied to assess
its efficiency for large simplicial data sets.

\vspace{-2mm}\item The computational complexity of SQC quantum algorithmic schemes is to be analyzed 

\vspace{-1mm}
\end{enumerate}

\noindent
We have provided only an elementary treatment of these topics. A more refined analysis
is definitely required. For reasons explained in the previous paragraph, this can only be left to
future work.

\subsection{\textcolor{blue}{\sffamily Extensions and generalizations of the SQC model
}}\label{subsec:sqcfuture}

We conclude this section by assessing prospective enhancements of the SQC model.
It would be useful:

\begin{enumerate}[label=\roman*),font=\itshape] 

\vspace{-1mm}\item to study whether the SQC framework is generalizable to understand its broader ramifications;

\vspace{-2mm}\item to explore potential applications of the SQC model beyond
simplicial sets. 

\vspace{-1mm}
\end{enumerate}

We have no answer to these questions presently. We only observe here that simplicial sets as well as  
cubical sets are instances of a more general type of combinatorial structure: the so called simploidal sets
\ccite{Peltier:2009doc,Peltier:2018sbs}. Simploidal sets have so far found useful applications mostly in
computational geometry, to the best of our knowledge. In principle, the SQC model might be generalized
to simploidal sets. In practice, this can be confirmed only by a detailed analysis.

\vfil\eject

\begin{appendix} 

\section{\textcolor{blue}{\sffamily Appendixes}}\label{app:app}

The following appendixes provide the more technical details of the proofs
of several propositions and theorems
stated in the main text of the paper. They are required for the completeness of the
analysis presented and are illustrative of the kind of techniques used.


\subsection{\textcolor{blue}{\sffamily The no degeneracy defect theorem 
}}\label{app:vanidect1}

In this appendix, we provide a proof of the no degeneracy defect theorem \cref{prop:vanidect1}.


On account of \ceqref{qusmplx19}, to demonstrate \ceqref{qusmplx20}, one has just to show that 
for $0\leq i,j\leq n$, $i<j$, if $\sigma_n,\omega_n\in\sfX_n$ and $s_{ni}\omega_n=s_{nj}\sigma_n$,
then there exists a unique $\zeta_{n-1}\in\sfX_{n-1}$ with $s_{n-1i}\zeta_{n-1}=\sigma_n$,
$s_{n-1j-1}\zeta_{n-1}=\omega_n$, We begin with recalling that for any $\rho_n\in\sfX_n$ the set
$\sfS_{n-1i}(\rho_n)$ is either empty or contains precisely one element, which by \ceqref{smplsets3}
is then $\eta_{n-1}=d_{ni}\rho_n=d_{ni+1}\rho_n$. By the simplicial identities \ceqref{smplsets2}--\ceqref{smplsets4},
we have  $s_{ni}\omega_n=s_{nj}\sigma_n\Rightarrow d_{ni}\sigma_n\in\sfS_{n-1j-1}(\omega_n)=\{d_{nj}\omega_n\}$
and $s_{ni}\omega_n=s_{nj}\sigma_n\Rightarrow d_{nj}\omega_n\in\sfS_{n-1i}(\sigma_n)=\{d_{ni}\sigma_n\}$.
Then, $\zeta_{n-1}=d_{nj}\omega_n=d_{ni}\sigma_n\in\sfX_n$ is the unique element such that $s_{n-1i}\zeta_{n-1}=\sigma_n$,
$s_{n-1j-1}\zeta_{n-1}=\omega_n$. 


\subsection{\textcolor{blue}{\sffamily The perfectness propositions 
}}\label{app:vanidect2}

We provide below the proofs of the perfectness propositions  
\cref{prop:vanidect2a}, \cref{prop:vanidect2b}.

We suppose first that the simplicial set $\sfX$ under consideration
is the nerve of a finite category $\clC$, $\sfX=\sfN\clC$
(cf. subsect. \cref{subsec:smplsets}). We have to verify whether
conditions \ceqref{qusmplx22}--\ceqref{qusmplx24} are fulfilled.

Below, we shall adopt the following convention:
any sequence $f_r,f_{r+1}\ldots,f_{s-1},f_s$ of morphisms
or identities between morphisms of $\clC$  displayed in the course of the analysis is tacitly
assumed to be absent whenever $r>s$. 

Because of \ceqref{qusmplx17}, to verify \ceqref{qusmplx22} it is enough to prove that
for $0\leq i,j\leq n$, $i\leq j$, if $\sigma_n\in\sfN_n\clC$, $\omega_{n+2}\in\sfN_{n+2}\clC$
and $d_{n+2i}\omega_{n+2}=s_{nj}\sigma_n$, then there is a unique $\zeta_{n+1}\in\sfN_{n+1}\clC$
such that $d_{n+1i}\zeta_{n+1}=\sigma_n$, $s_{n+1j+1}\zeta_{n+1}=\omega_{n+2}$. Write
\begin{align}
\sigma_n=(g_1,\ldots,g_n), \qquad \omega_{n+2}=(h_1,\ldots,h_{n+2})
\label{prop1/6}
\end{align}
in terms of morphisms $g_k$, $h_l$ of $\clC$. The condition $d_{n+2i}\omega_{n+2}=s_{nj}\sigma_n$
implies that
\begin{multline}
h_1=g_1,\ldots,h_{i-1}=g_{i-1}, h_{i+1}\circ h_i=g_i, h_{i+2}=g_{i+1},
\label{prop1/7}
\\
\ldots,h_{j+1}=g_j,h_{j+2}=\id,h_{j+3}=g_{j+1},\ldots,h_{n+2}=g_n.
\end{multline}
Using relations \ceqref{prop1/7}, it is straightforward to check that the morphism string
\begin{align}
\zeta_{n+1}=(h_1,\ldots,h_{i+1},g_{i+1},\ldots,g_n)=(h_1,\ldots,h_{j+1},g_{j+1},\ldots,g_n)
\label{prop1/8}
\end{align}
is composable and so that $\in\sfN_{n+1}\clC$ and that $d_{n+1i}\zeta_{n+1}=\sigma_n$, $s_{n+1j+1}\zeta_{n+1}=\omega_{n+2}$.
The uniqueness of $\zeta_{n+1}$ is guaranteed by the injectivity of $s_{n+1j+1}$. 

From \ceqref{qusmplx18}, to demonstrate \ceqref{qusmplx23} it suffices to prove that
for $0\leq i,j\leq n$, $i+1< j$, if $\sigma_n\in\sfN_n\clC$, $\omega_{n-2}\in\sfN_{n-2}\clC$
and $s_{n-2i}\omega_{n-2}=d_{nj}\sigma_n$, then there is a unique $\zeta_{n-1}\in\sfN_{n-1}\clC$
such that $s_{n-1i}\zeta_{n-1}=\sigma_n$, $d_{n-1j-1}\zeta_{n-1}=\omega_{n-2}$. Write
\begin{align}
\sigma_n=(g_1,\ldots,g_n), \qquad \omega_{n-2}=(h_1,\ldots,h_{n-2})
\label{prop1/9}
\end{align}
in terms of morphisms $g_k$, $h_l$ of $\clC$. The condition $s_{n-2i}\omega_{n-2}=d_{nj}\sigma_n$ 
implies that
\begin{multline}
h_1=g_1,\ldots,h_i=g_i, g_{i+1}=\id, h_{i+1}=g_{i+2},\ldots,
\label{prop1/10}
\\
h_{j-2}=g_{j-1},h_{j-1}=g_{j+1}\circ g_j, h_j=g_{j+2},\ldots,h_{n-2}=g_n.
\end{multline}
Using relations \ceqref{prop1/10}, one checks readily that the morphism sequence
\begin{align}
\zeta_{n-1}=(h_1,\ldots,h_i,g_{i+2},\ldots,g_n)=(h_1,\ldots,h_{j-2},g_j,\ldots,g_n)
\label{prop1/11}
\end{align}
is composable so that $\in\sfN_{n-1}\clC$ and that $s_{n-1i}\zeta_{n-1}=\sigma_n$, $d_{n-1j-1}\zeta_{n-1}=\omega_{n-2}$. 
The uniqueness of $\zeta_{n-1}$ is guaranteed by the injectivity of $s_{n-1i}$.

By virtue of \ceqref{qusmplx16}, to show \ceqref{qusmplx21}, \ceqref{qusmplx24} it is sufficient to show
that for $0\leq i,j\leq n$, $i\leq j$, if $\sigma_n,\omega_n\in \sfN_n\clC$ and $d_{ni}\omega_n=d_{nj}\sigma_n$,
then there is a unique $\zeta_{n+1}\in\sfN_{n+1}\clC$ such that $d_{n+1i}\zeta_{n+1}=\sigma_n$,
$d_{n+1j+1}\zeta_{n+1}=\omega_n$. Let
\begin{align}
\sigma_n=(g_1,\ldots,g_n), \qquad \omega_n=(h_1,\ldots,h_n)
\label{prop1/1}
\end{align}
in terms of morphisms $g_k$, $h_l$ of $\clC$. The condition $d_{ni}\omega_n=d_{nj}\sigma_n$
implies that
\begin{multline}
h_1=g_1,\ldots,h_{i-1}=g_{i-1}, h_{i+1}\circ h_i=g_i, h_{i+2}=g_{i+1},\ldots, 
\label{prop1/2}
\\
h_j=g_{j-1},h_{j+1}=g_{j+1}\circ g_j,h_{j+2}=g_{j+2},\ldots,h_n=g_n
\end{multline}
for $i<j$ and 
\begin{multline}
h_1=g_1,\ldots,h_{i-1}=g_{i-1}, h_{i+1}\circ h_i=g_{i+1}\circ g_i,h_{i+2}=g_{i+2},\ldots,h_n=g_n
\label{prop1/3}
\end{multline}
for $i=j$. 
When $i<j$, it is simple to check exploiting relations \ceqref{prop1/2} that
\begin{equation}
\zeta_{n+1}=(h_1,\ldots,h_{i+1},g_{i+1},\ldots,g_n)=(h_1,\ldots,h_j,g_j,\ldots,g_n)
\label{prop1/4}
\end{equation}
is composable and so $\zeta_{n+1}\in\sfN_{n+1}\clC$
and that $d_{n+1i}\zeta_{n+1}=\sigma_n$, $d_{n+1j+1}\zeta_{n+1}=\omega_n$ as claimed. Letting
$\zeta_{n+1}=(z_1,\ldots,z_{n+1})$ for generic morphisms $z_k$
and imposing that $d_{n+1i}\zeta_{n+1}=\sigma_n$, $d_{n+1j+1}\zeta_{n+1}=\omega_n$ we find further that
$\zeta_{n+1}$ is precisely of the form \ceqref{prop1/4}, showing the uniqueness of $\zeta_{n+1}$.
When $i=j$, provided $\clC$ is a groupoid, one similarly verifies from \ceqref{prop1/3} that,
setting $f_{i+1}=g_i\circ h_i{}^{-1}=g_{i+1}{}^{-1}\circ h_{i+1}$ and 
\begin{align}
&\zeta_{n+1}=(h_1,\ldots,h_i,f_{i+1},g_{i+1},\ldots,g_n), 
\label{prop1/5}
\end{align}
$\zeta_{n+1}\in\sfN_{n+1}\clC$ 
and that $d_{n+1i}\zeta_{n+1}=\sigma_n$, $d_{n+1j+1}\zeta_{n+1}=\omega_n$ again. 
When $i=0$, only the second expression of $f_{i+1}$ holds and when $i=n$ only the first.
The equality of the two expressions of $f_{i+1}$ for $0<i<n$ is guarantied by the $i$th relation \ceqref{prop1/3}.
Writing
$\zeta_{n+1}=(z_1,\ldots,z_{n+1})$ for generic morphisms $z_k$
and imposing that $d_{n+1i}\zeta_{n+1}=\sigma_n$, $d_{n+1j+1}\zeta_{n+1}=\omega_n$ we find further that
$\zeta_{n+1}$ is precisely of the form \ceqref{prop1/5}, showing once more the uniqueness of $\zeta_{n+1}$.

We suppose next that the relevant simplicial set $\sfX$
considered is the simplicial set of an ordered finite simplicial complex 
$\clS$, $\sfX=\sfK\clS$ (cf. subsect. \cref{subsec:smplsets}). We have to verify whether conditions
\ceqref{qusmplx22}--\ceqref{qusmplx23} are fulfilled.

Below, analogously to what done earlier, we shall adopt the convention that 
any sequence $v_r,v_{r+1}\ldots,v_{s-1},v_s$ of vertices
or identities between vertices of $\clS$ displayed in the analysis is tacitly
assumed to be absent whenever $r>s$.

Because of \ceqref{qusmplx17} again, verifying \ceqref{qusmplx22} reduces just to demonstrating that
for $0\leq i,j\leq n$, $i\leq j$, if $\sigma_n\in\sfK_n\clS$, $\omega_{n+2}\in\sfK_{n+2}\clS$
and $d_{n+2i}\omega_{n+2}=s_{nj}\sigma_n$, then there is a unique $\zeta_{n+1}\in\sfK_{n+1}\clS$
such that $d_{n+1i}\zeta_{n+1}=\sigma_n$, $s_{n+1j+1}\zeta_{n+1}=\omega_{n+2}$.
Write 
\begin{align}
\sigma_n=(a_0,\ldots,a_n), \qquad \omega_{n+2}=(b_0,\ldots,b_{n+2})
\label{prop1/15}
\end{align}
in terms of vertices $a_k$, $b_l$ of $\clS$. The condition $d_{n+2i}\omega_{n+2}=s_{nj}\sigma_n$
implies that
\begin{equation}
b_0=a_0,\ldots,b_{i-1}=a_{i-1}, b_{i+1}=a_i,\ldots, b_{j+1}=a_j,b_{j+2}=a_j,\ldots,b_{n+2}=a_n.
\label{prop1/16}
\end{equation}
Using relations \ceqref{prop1/16}, it is straightforward to check that
\begin{align}
\zeta_{n+1}=(b_0,\ldots,b_i,a_i,\ldots,a_n)=(b_0,\ldots,b_{j+1},a_{j+1},\ldots,a_n) 
\label{prop1/17}
\end{align}
does the job. 
Since $\zeta_{n+1}=(b_0,\ldots,b_{j+1},b_{j+3},\ldots,b_{n+2})$, it holds that $\zeta_{n+1}\in\sfK_{n+1}\clS$.
Further, $d_{n+1i}\zeta_{n+1}=\sigma_n$, $s_{n+1j+1}\zeta_{n+1}=\omega_{n+2}$ as is readily checked.
The uniqueness of $\zeta_{n+1}$ is guaranteed by the injectivity of $s_{n+1j+1}$. 

By \ceqref{qusmplx18}, showing \ceqref{qusmplx23} amounts to proving that
for $0\leq i,j\leq n$, $i+1< j$, if $\sigma_n\in\sfK_n\clS$, $\omega_{n-2}\in\sfK_{n-2}\clS$
and $s_{n-2i}\omega_{n-2}=d_{nj}\sigma_n$, then there exists a unique $\zeta_{n-1}\in\sfK_{n-1}\clS$
such that $s_{n-1i}\zeta_{n-1}=\sigma_n$, $d_{n-1j-1}\zeta_{n-1}=\omega_{n-2}$.
Write 
\begin{align}
\sigma_n=(a_0,\ldots,a_n), \qquad \omega_{n-2}=(b_0,\ldots,b_{n-2})
\label{prop1/12}
\end{align}
in terms of vertices $a_k$, $b_l$ of $\clS$. The condition $s_{n-2i}\omega_{n-2}=d_{nj}\sigma_n$ 
implies that
\begin{equation}
b_0=a_0,\ldots,b_i=a_i, b_i=a_{i+1},\ldots,b_{j-2}=a_{j-1},b_{j-1}=a_{j+1},\ldots,b_{n-2}=a_n.
\label{prop1/13}
\end{equation}
Owing to the relations \ceqref{prop1/13}, it is straightforward to check that
\begin{align}
\zeta_{n-1}=(b_0,\ldots,b_i,a_{i+2},\ldots,a_n)=(b_0,\ldots,b_{j-2},a_j,\ldots,a_n) 
\label{prop1/14}
\end{align}
has the required properties. As $\zeta_{n-1}=(a_0,\ldots,a_i,a_{i+2},\ldots,a_n)$, $\zeta_{n-1}\in\sfK_{n-1}\clS$.
Fur\-ther, the identities $s_{n-1i}\zeta_{n-1}=\sigma_n$, $d_{n-1j-1}\zeta_{n-1}=\omega_{n-2}$ hold. 
The uniqueness of $\zeta_{n-1}$ is guaranteed by the injectivity of $s_{n-1i}$.

Though the above analysis completes the proof of the perfectness proposition for the simplicial
set $\sfK\clS$, it may be interesting to get an intuition of the reason why $\sfK\clS$ fails to be
quasi perfect or perfect. By \ceqref{qusmplx16},
showing \ceqref{qusmplx21}, \ceqref{qusmplx24} would require proving that
for $0\leq i,j\leq n$, $i\leq j$, if $\sigma_n,\omega_n\in \sfK_n\clC$ and $d_{ni}\omega_n=d_{nj}\sigma_n$,
then there is a unique $\zeta_{n+1}\in\sfK_{n+1}\clS$ such that $d_{n+1i}\zeta_{n+1}=\sigma_n$,
$d_{n+1j+1}\zeta_{n+1}=\omega_n$. Let 
\begin{align}
\sigma_n=(a_0,\ldots,a_n), \qquad \omega_n=(b_0,\ldots,b_n).
\label{prop1/18}
\end{align}
Then, considerations analogous to those carried out above would yield for $i<j$
\begin{align}
\zeta_{n+1}=(b_0,\ldots,b_i,a_i,\ldots,a_n)=(b_0,\ldots,b_j,a_j,\ldots,a_n).
\label{prop1/19}
\end{align}
This simplex, however, cannot be written in terms of either the $a_k$ or the $b_k$ only, so there
is no guarantee that $\zeta_{n+1}\in\sfK_{n+1}\clS$ by lack of increasing vertex ordering.


\subsection{\textcolor{blue}{\sffamily The simplicial Hilbert Hodge theorem  
}}\label{app:hodge}

The proof of the simplicial Hilbert Hodge theorem \cref{prop:hodge} follows immediately from
the following proposition.
~
\begin{prop} \label{prop:finitehodge}
(Finite dimensional Hodge theorem) Let $\scK_n$, $n\geq 0$, a sequence of finite dimensional Hilbert spaces.
Let further 
\begin{align}
&\xymatrix@C=2.5pc
{\cdots
&\ar[l]_{Q_2}\scK_2&
\ar[l]_{Q_1}
\scK_1&\ar[l]_{Q_0}\scK_0}
\label{fdht1}
\\
\intertext{be a Hilbert cochain complex and}
&\xymatrix@C=2.5pc
{\cdots
\ar[r]^{Q_2{}^+}
&\scK_2
\ar[r]^{Q_1{}^+}
&\scK_1\ar[r]^{Q_0{}^+}&\scK_0}
\label{fdht2}
\end{align}
be its adjoint chain complex. For $n\geq 0$, denote by $\rmH^n(\scK,Q)=\ker Q_n/\ran Q_{n-1}$ 
and $\rmH_n(\scK,Q{}^+)=\ker Q_{n-1}{}^+/\ran Q_n{}^+$ their (co)homology spaces, where by convention 
$\ran Q_{-1}=0$ and $\ker Q_{-1}{}^+=\scK_0$ as usual. Then, 
\begin{align}
&\rmH^n(\scK,Q)\simeq\ker H_n,  
\label{fdht3}
\\
&\rmH_n(\scK,Q{}^+)\simeq\ker H_n 
\label{fdht4}
\end{align}
for any $n\geq 0$, where $H_n:\scK_n\rightarrow\scK_n$ are the Laplacians 
\begin{align}
H_n=Q_n{}^+Q_n+Q_{n-1}Q_{n-1}{}^+,
\label{fdht5}
\end{align}
the second tern being missing when $n=0$.
\end{prop}  

We show the theorem. Let $n$ be fixed. We notice that $H_n$ is a Hermitian operator. So is then
any real valued function of $H_n$. Two operators of this kind will be relevant for the proof.
The first is the orthogonal projector $P_n$ onto the kernel of $H_n$. 
The second is the generalized inverse $H_n{}^{-1}$ of $H_n$. They satisfy the
following basic identities which will be employed repeatedly in the proof:
\begin{align}
&H_n{}^{-1}H_n=1_n-P_n,
\label{fdht6}
\\
&H_nP_n=0, \qquad H_n{}^{-1}P_n=0, 
\label{fdht7}
\end{align}
where $1_n=1_{\scK_n}$. 

The kernel of $H_n$ is \pagebreak contained in that of $Q_n$ and the range of $Q_n$ is
orthogonal to the kernel of $H_{n+1}$, since for $n\geq 0$
\begin{equation}
Q_nP_n=P_{n+1}Q_n=0.
\label{fdht9}
\end{equation}
Indeed, owing to the first relation \ceqref{fdht7}, 
\begin{align}
0&=P_nH_nP_n
\label{fdht10}
\\
&=P_n(Q_n{}^+Q_n+Q_{n-1}Q_{n-1}{}^+)P_n
\noindent
\\
&=(Q_nP_n)^+Q_nP_n+P_nQ_{n-1}(P_nQ_{n-1})^+.
\noindent
\end{align}
As $(Q_nP_n)^+Q_nP_n,\,P_nQ_{n-1}(P_nQ_{n-1})^+\geq 0$ as operators, we find that $Q_nP_n=0$, $P_nQ_{n-1}=0$,
the second relation being missing for $n=0$. 

From relation \ceqref{fdht5}, using that $Q_nQ_{n-1}=0$, it is immediate to verify that 
\begin{equation}
Q_nH_n=H_{n+1}Q_n
\label{fdht8}
\end{equation}  
for $n\geq 0$. From \ceqref{fdht8}, it follow that 
\begin{equation}
Q_nH_n{}^{-1}=H_{n+1}{}^{-1}Q_n.
\label{fdht11}
\end{equation}
Indeed, on account of \ceqref{fdht6}, \ceqref{fdht9}, 
\begin{align}
H_{n+1}{}^{-1}Q_n-Q_nH_n{}^{-1}&=H_{n+1}{}^{-1}Q_n(1_n-P_n)-(1_{n+1}-P_{n+1})Q_nH_n{}^{-1}
\label{fdht12}
\\
&=H_{n+1}{}^{-1}Q_nH_nH_n{}^{-1}-H_{n+1}{}^{-1}H_{n+1}Q_nH_n{}^{-1}
\nonumber
\\
&=H_{n+1}{}^{-1}(Q_nH_n-H_{n+1}Q_n)H_n{}^{-1}=0.
\nonumber
\end{align}

Using \ceqref{fdht11}, we can show that for $n\geq 0$
\begin{equation}
1_n=P_n+Q_n{}^+H_{n+1}{}^{-1}Q_n+Q_{n-1}H_{n-1}{}^{-1}Q_{n-1}{}^+,
\label{fdht13}
\end{equation}
where for $n=0$ the third term is missing. Indeed, 
\begin{align}
Q_n{}^+H_{n+1}{}^{-1}Q_n+Q_{n-1}H_{n-1}{}^{-1}Q_{n-1}{}^+
&=(Q_n{}^+Q_n+Q_{n-1}Q_{n-1}{}^+)H_n{}^{-1}
\label{fdht14}
\\
&=H_nH_n{}^{-1}=1_n-P_n, 
\nonumber
\end{align}
whence \ceqref{fdht13} follows. 

We have already observed that \pagebreak $\ker H_n\subseteq \ker Q_n$, 
since $Q_nP_n=0$ by \ceqref{fdht9}.
As a consequence, a linear map $\mu_n:\ker H_n\rightarrow\rmH^n(\scK,Q)$ is defined given by
\begin{equation}
\mu_n\ket{\psi_n}=\ket{\psi_n}+\ran Q_{n-1},\
\label{fdht15}
\end{equation}
with $\ket{\psi_n}\in\ker H_n$. 
We are going to show that $\mu_n$ is an isomorphism. 
We prove first that $\mu_n$ is injective. 
Suppose that $\ket{\psi_n},\ket{\phi_n}\in\ker H_n$ and $\ket{\psi_n}-\ket{\phi_n}=Q_{n-1}\ket{\chi_{n-1}}$
for some $\ket{\chi_{n-1}}\in\scK_{n-1}$, where $Q_{-1}\ket{\chi_{-1}}=0$ by convention. Then, by the second
relation \ceqref{fdht9},
\begin{equation}
\ket{\psi_n}-\ket{\phi_n}=P_n(\ket{\psi_n}-\ket{\phi_n})=P_nQ_{n-1}\ket{\chi_{n-1}}=0,
\label{fdht16}
\end{equation}
proving the stated injectivity of $\mu_n$. 
We demonstrate next that $\mu_n$ is a surjective map. If $Q_n\ket{\psi_n}=0$, then by \ceqref{fdht13}
\begin{equation}
\ket{\psi_n}=P_n\ket{\psi_n}+Q_{n-1}H_{n-1}{}^{-1}Q_{n-1}{}^+\ket{\psi_n}
\label{fdht17}
\end{equation}
and consequently 
\begin{equation}
\ket{\psi_n}+\ran Q_{n-1}=P_n\ket{\psi_n}+\ran Q_{n-1}=\mu_nP_n\ket{\psi_n},
\label{fdht18}
\end{equation}
showing the stated surjectivity of $\mu_n$. So, $\mu_n$ is an isomorphism as claimed. The isomorphism
\ceqref{fdht3} follows.

The isomorphism \ceqref{fdht4} is shown similarly. We have that 
$\ker H_n\subseteq \ker Q_{n-1}{}^+$, 
by the relation $Q_{n-1}{}^+P_n=0$ from \ceqref{fdht9} if $n>0$ or trivially if $n=0$.
As a consequence, a linear map $\nu_n:\ker H_n\rightarrow\rmH_n(\scK,Q{}^+)$
\begin{equation}
\nu_n\ket{\psi_n}=\ket{\psi_n}+\ran Q_n{}^+
\label{fdht19}
\end{equation}
with $\ket{\psi_n}\in\ker H_n$, is defined. $\nu_n$ can be shown to be an isomorphism
by verifying its injectivity and surjectivity as done earlier for $\mu_n$.


\subsection{\textcolor{blue}{\sffamily Triviality of the kernel of the
    degeneracy simplicial Hodge Laplacian 
}}\label{app:hstriv}

Combining \ceqref{susysmplx12}, \ceqref{qusmplx37} and \ceqref{qusmplx33}
\begin{equation}
H_{SSn}=\mycom{{}_\sss}{{}_{\sigma_n\in \sfX_n}}\ket{\sigma_n}\Big(n+1
-\mycom{{}_\sss}{{}_{0\leq i\leq n-1}}|\sfS_{n-1i}(\sigma_n)|\Big)\bra{\sigma_n}.
\label{}
\end{equation}
Since $|\sfS_{n-1i}(\sigma_n)|\leq 1$, we have $H_{SSn}\geq 1_n$. Consequently,
$\ker H_{SSn}=0$ as required.


\subsection{\textcolor{blue}{\sffamily The normalized Hilbert homology theorem
}}\label{app:deginv}

We begin with showing identity \ceqref{normal2} providing an expression of the orthogonal projector $\varPi_n$ 
on the degenerate $n$-simplex space ${}^s\!\scH_n\subseteq\scH_n$ shown in eq. \ceqref{normal0}.
Recall that the orthogonal projectors on the ranges $\ran S_{n-1i}$ of the degeneracy operators 
$S_{n-1i}$ are the operators $\varPi_{ni}$ given in \ceqref{qusmplx32}. Recall also that the $\varPi_{ni}$
commute pairwise. Now, from \ceqref{normal0}, the space ${}^s\!\scH_n$ can be expressed as 
\begin{equation}
{}^s\!\scH_n=\Big(\nnn_{i=0}^{n-1}\ran S_{n-1i}{}^\perp\Big)^\perp.
\label{normal-1}
\end{equation} 
Now, the projectors onto the subspaces $\ran S_{n-1i}{}^\perp$ are the operators $1_n-\varPi_{ni}$.
By the commutativity of the $\varPi_{ni}$, the projector onto the subspace $\bigcap_{i=0}^{n-1}\ran S_{n-1i}{}^\perp$
is then $\prod_{i=0}^{n-1}(1_n-\varPi_{ni})$. \ceqref{normal-1} then implies that the projector $\varPi_n$ is
given by \ceqref{normal2} as claimed.

Next, we demonstrate the homological relation 
\begin{equation}
{}^cQ_{Dn-1}{}^cQ_{Dn}=0.
\label{normal5}
\end{equation}
To that purpose, the relations 
\begin{equation}
(1_{n-1}-\varPi_{n-1})Q_{Dn}\varPi_n=0
\label{normal3}
\end{equation}
valid for $n\geq 1$ are required. The proof of \ceqref{normal3} runs as follows.
Let $0\leq k\leq n$. From \ceqref{qusmplx32}, using relations \ceqref{qusmplx8}--\ceqref{qusmplx10} yields 
\begin{align}
Q_{Dn}\varPi_{nk}&=\mycom{{}_\sss}{{}_{0\leq i\leq n}}(-1)^iD_{ni}S_{n-1k}S_{n-1k}{}^+
\label{norma14}
\\
&=\mycom{{}_\sss}{{}_{0\leq i<k}}(-1)^iS_{n-2k-1}D_{n-1i}S_{n-1k}{}^+
+\mycom{{}_\sss}{{}_{k+1< i\leq n}}(-1)^iS_{n-2k}D_{n-1i-1}S_{n-1k}{}^+,
\nonumber
\end{align}
where the first and second term are $0$ when respectively $k=0$ and $k=n$.
This expressions shows that that $\ran Q_{Dn}\varPi_{nk}\subseteq\ran S_{n-2k-1}+\ran S_{n-2k}$.
Consequently, 
\begin{equation}
(1_{n-1}-\varPi_{n-1k-1})(1_{n-1}-\varPi_{n-1k})Q_{Dn}\varPi_{nk}=0.
\label{normal16}
\end{equation}
By \ceqref{normal2}, we can replace the first two factors by $1_{n-1}-\varPi_{n-1}$.
By \ceqref{normal2} again, $\varPi_n$ is a polynomial in the $\varPi_{ni}$
with vanishing degree $0$ term. Hence, we can replace the last factor by $\varPi_{n-1}$.
Relation \ceqref{normal3} follows. \pagebreak Using this latter, we can readily prove \ceqref{normal5}:
from \ceqref{normal4}, 
\begin{align}
{}^cQ_{Dn-1}&{}^cQ_{Dn}
=(1_{n-2}-\varPi_{n-2})Q_{Dn-1}(1_{n-1}-\varPi_{n-1})Q_{Dn}\big|_{{}^c\!\scH_n}
\label{}
\\  
&=(1_{n-2}-\varPi_{n-2})Q_{Dn-1}Q_{Dn}\big|_{{}^c\!\scH_n}
-(1_{n-2}-\varPi_{n-2})Q_{Dn-1}\varPi_{n-1}Q_{Dn}\big|_{{}^c\!\scH_n}=0.
\nonumber
\end{align}

Next, we provide the details of the proof of prop. \cref{prop:hilbnorm/1}.
The design of the proof is already outlined in subsect \cref{subsec:normal} 
and rests on the construction of a chain equivalence
of the abstract and concrete Hilbert complexes $(\ol{\scH},\ol{Q}_D)$, $({}^c\!\scH, {}^cQ_D)$
(see subsect. \cref{subsec:normal} for their definition).
The chain equivalence given by a sequence of chain operators
$I_n:\ol{\scH}_n\rightarrow{}^c\!\scH_n$,
$J_n:{}^c\!\scH_n\rightarrow\ol{\scH}_n$, $n\geq 0$, with 
$J_nI_n$, $I_nJ_n$ chain homotopic to $\ol{1}_n$, ${}^c1_n$ respectively (cf. eqs. \ceqref{normal8/9}
and \ceqref{normal10/11}).
$I_n$ is the operator from $\ol{\scH}_n$ to ${}^c\!\scH_n$
induced by the orthogonal projector $1_n-\varPi_n$ by virtue of the fact that  ${}^s\!\scH_n=\ker(1_n-\varPi_n)$. 
$J_n$ is the canonical projection of ${}^c\!\scH_n$ onto $\ol{\scH}_n$.

We show first that the chain operator relations \ceqref{normal8/9} are
fulfilled. By \ceqref{normal4}, for $\ket{\psi_n}\in\scH_n$, we have 
\begin{align}
I_{n-1}\ol{Q}_{Dn}(\ket{\psi_n}+{}^s\!\scH_n)&=I_{n-1}(Q_{Dn}\ket{\psi_n}+{}^s\!\scH_{n-1})
\label{}
\\  
&=(1_{n-1}-\varPi_{n-1})Q_{Dn}\ket{\psi_n}
\nonumber
\\
&=(1_{n-1}-\varPi_{n-1})Q_{Dn}(1_n-\varPi_n)\ket{\psi_n}
\nonumber
\\
&={}^cQ_{Dn}(1_n-\varPi_n)\ket{\psi_n}
\nonumber
\\
&={}^cQ_{Dn}I_n(\ket{\psi_n}+{}^s\!\scH_n),
\nonumber
\end{align}
showing \ceqref{normal8}. Similarly, by \ceqref{normal4} again, for $\ket{\psi_n}\in{}^c\!\scH_n$
\begin{align}
J_{n-1}{}^cQ_{Dn}\ket{\psi_n}&=J_{n-1}(1_{n-1}-\varPi_{n-1})Q_{Dn}\ket{\psi_n}
\label{}
\\  
&=(1_{n-1}-\varPi_{n-1})Q_{Dn}\ket{\psi_n}+{}^s\!\scH_{n-1}
\nonumber
\\
&=Q_{Dn}\ket{\psi_n}+{}^s\!\scH_{n-1}
\nonumber
\\
&=\ol{Q}_{Dn}(\ket{\psi_n}+{}^s\!\scH_n),
\nonumber
\\
&=\ol{Q}_{Dn}J_n\ket{\psi_n},
\nonumber
\end{align}
proving  \ceqref{normal9}.

We show next that the chain operator homotopy relations \ceqref{normal10/11} hold too. The proof is 
simple. The maps $I_n$, $J_n$ turn out to be reciprocally inverse. So, \ceqref{normal10}, \ceqref{normal11}
hold trivially with $\ol{W}_k=0$, ${}^cW_k=0$. The isomorphism \ceqref{normal7} follows.

\end{appendix}

\vfil\eject

\vspace{.5cm}

\noindent
\markright{\textcolor{blue}{\sffamily Acknowledgements}}

\noindent
\textcolor{blue}{\sffamily Acknowledgements.} 
The author acknowledges financial support from INFN Research Agency
under the provisions of the agreement between Alma Mater Studiorum University of Bologna and INFN. 
He is grateful to the INFN Galileo Galilei Institute in Florence, where part of this work was done,
for hospitality and support. He further thanks the Staff of the Center for Quantum and Topological Systems
at New York University Abu Dhabi, and in particular Hisham Sati, for inviting him as speaker of their Quantum Colloquium
to present this work. Finally, he thanks Jim Stasheff and Urs Schreiber for correspondence.

Most of the analytic calculations presented in this paper
have carried out employing the WolframAlpha computational platform.

\vspace{.5cm}

\noindent
\textcolor{blue}{\sffamily Conflict of interest statement.}
The author declares that the results of the current study do not involve any conflict of interest.

\vspace{.5cm}

\noindent
\textcolor{blue}{\sffamily Data availability statement.}
The author declares that no datasets were generated or analysed during the current study.

\vfill\eject

\noindent
\textcolor{blue}{\bf\sffamily References}


\end{document}